\documentclass[draftcls,11pt]{IEEEtran}
\onecolumn 
\usepackage[cmex10]{amsmath}
\usepackage{xcolor}

\usepackage{mdwtab}
\usepackage{eqparbox}
\usepackage{graphicx,amssymb,rangecite,upgreek,graphicx,dsfont,mathrsfs}

\usepackage{algorithm}
\usepackage{algorithmic} 
\usepackage[usenames,dvipsnames]{pstricks}
 \usepackage{epsfig}

 \usepackage{pst-grad} 
 \usepackage{pst-plot} 
\newcommand {\exe} {\stackrel{\cdot} {=}}

\newcommand {\bx} {\mbox{\boldmath $x$}}
\newcommand {\by} {\mbox{\boldmath $y$}}

\newcommand {\bA} {\mbox{\boldmath $A$}}

\newcommand {\bE} {\mbox{\boldmath $E$}}
\newcommand {\bF} {\mbox{\boldmath $F$}}

\newcommand {\bI} {\mbox{\boldmath $I$}}

\newcommand {\bN} {\mbox{\boldmath $N$}}

\newcommand {\bP} {\mbox{\boldmath $P$}}

\newcommand {\bS} {\mbox{\boldmath $S$}}
\newcommand {\bT} {\mbox{\boldmath $T$}}

\newcommand {\bX} {\mbox{\boldmath $X$}}
\newcommand {\bY} {\mbox{\boldmath $Y$}}

\newcommand {\bSig} {\mbox{\boldmath $\Sigma$}}

\newcommand {\blam} {\mbox{\boldmath $\lambda$}}
\newcommand {\bXi} {\mbox{\boldmath $\Xi$}}

\newcommand {\blamt} {\mbox{\boldmath \footnotesize $\lambda$}}

\newcommand{\calA}{{\cal A}}

\newcommand{\calC}{{\cal C}}

\newcommand{\calG}{{\cal G}}

\newcommand{\calI}{{\cal I}}

\newcommand{\calN}{{\cal N}}

\newcommand{\calT}{{\cal T}}

\newcommand{\calX}{{\cal X}}
\newcommand{\calY}{{\cal Y}}

\newcommand{\Amat}{{\bf{A}}}
\newcommand{\Bmat}{{\bf{B}}}
\newcommand{\Cmat}{{\bf{C}}}
\newcommand{\Dmat}{{\bf{D}}}

\newcommand{\Hmat}{{\bf{H}}}
\newcommand{\Jmat}{{\bf{J}}}

\newcommand{\Qmat}{{\bf{Q}}}

\newcommand{\Smat}{{\bf{S}}}
\newcommand{\Tmat}{{\bf{T}}}

\newcommand{\Rmat}{{\bf{R}}}

\newcommand{\define}{\stackrel{\triangle}{=}}

\newcommand{\be}{\begin{equation}}
\newcommand{\ee}{\end{equation}}
\newcommand{\beqna}{\begin{eqnarray}}
\newcommand{\eeqna}{\end{eqnarray}}

\usepackage{theorem}
\DeclareFontFamily{U}{mathx}{\hyphenchar\font45}
\DeclareFontShape{U}{mathx}{m}{n}{
      <5> <6> <7> <8> <9> <10>
      <10.95> <12> <14.4> <17.28> <20.74> <24.88>
      mathx10
      }{}
\DeclareSymbolFont{mathx}{U}{mathx}{m}{n}
\DeclareMathSymbol{\bigtimes}{1}{mathx}{"91}
\theorembodyfont{\rmfamily}
 
\newcommand{\Kmat}{{\bf{K}}}
\newcommand{\Ind}{{\mathds{1}}}

\newcommand{\abs}[1]{\left|#1\right|}

\newcommand {\bze} {\mbox{\boldmath $0$}}

\newcommand{\diag}{\mathop{\mathrm{diag}}}
\DeclareMathOperator{\re}{Re}

\theoremheaderfont{\itshape}
\newtheorem{definition}{Definition}
\newtheorem{theorem}{Theorem}
\newtheorem{proof}{Proof}
\newtheorem{example}{Example}

\newtheorem{lemma}{Lemma} 

\newtheorem{prop}{Proposition}
\newtheorem{remark}{Remark}
\newcommand{\p}[1]{\left(#1\right)}
\newcommand{\pp}[1]{\left[#1\right]}
\newcommand{\ppp}[1]{\left\{#1\right\}}
\newcommand{\norm}[1]{\left\|#1\right\|}

\begin{document}

\title{Analysis of Mismatched Estimation Errors Using Gradients of Partition Functions$^\ast$}
\author{Wasim~Huleihel
        and~Neri~Merhav
				\\
        Department of Electrical Engineering \\
Technion - Israel Institute of Technology \\
Haifa 32000, ISRAEL\\
E-mail: \{wh@tx, merhav@ee\}.technion.ac.il
\thanks{$^\ast$This research was partially supported by The Israeli Science Foundation (ISF), grant no. 412/12.}
}

\maketitle

\begin{abstract}
\boldmath We consider the problem of signal estimation (denoising) from a statistical-mechanical perspective, in continuation to a recent work on the analysis of mean-square error (MSE) estimation using a direct relationship between optimum estimation and certain partition functions. The paper consists of essentially two parts. In the first part, using the aforementioned relationship, we derive single-letter expressions of the mismatched MSE of a codeword (from a randomly selected code), corrupted by a Gaussian vector channel. In the second part, we provide several examples to demonstrate phase transitions in the behavior of the MSE. These examples enable us to understand more deeply and to gather intuition regarding the roles of the real and the mismatched probability measures in creating these phase transitions. 
\end{abstract}


\IEEEpeerreviewmaketitle

\section{Introduction}

\IEEEPARstart{T}{he} connections and the interplay between information theory, statistical physics and signal estimation have been known for several decades \cite{Bucy,Duncan,Kailath,Seidler}, and they are still being studied from a variety of aspects, see, for example \cite{Neri1,Neri2,Guo1,Guo2,Guo3,Guo4,Palomar1,Palomar2,Raginsky,Verdu1,Weissman1,Atar1,ronit} and many references therein. 

Recently, in \cite{Neri2}, the well known I-MMSE relation \cite{Guo2}, which relates the mutual information and the derivative of the minimum mean-square error (MMSE), was further explored using a statistical physics perspective. Specifically, in their analysis, the authors of \cite{Neri2} exploit the natural ``mapping" between information theory problems and certain models of many-particle systems in statistical mechanics (see, e.g., \cite{Sourlas1,Sourlas2}). One of the main contributions in \cite{Neri2} is the demonstration of the usefulness of statistical-mechanical tools (in particular, utilizing the fact that the mutual information can be viewed as the partition function of a certain physical system) in assessing MMSE via the I-MMSE relation of \cite{Guo2}. More recently, Merhav \cite{Neri1} proposed a more flexible method, whose main idea is that, for the purpose of evaluating the covariance matrix of the MMSE estimator, one may use other information measures, which have the form of a partition function and hence can be analyzed using methods of statistical physics (see, e.g., \cite{Sourlas1,Sourlas2,Cousseau,Hosaka,Iba,Kabashima1,Kabashima2,Tanaka1,Tanaka2} and many references therein). The main advantage of the proposed approach over the I-MMSE relations, is its full generality: Any joint probability function $P\p{\bx,\by}$, where $\bx$ and $\by$ designate the channel input to be estimated and the channel output, respectively, can be handled (for example, the channel does not have to be additive or Gaussian). Moreover, using this approach, any mismatch, both in the source and the channel, can be considered. 

This paper is a further development of \cite{Neri1} in the above described direction. Particularly, in \cite[Section IV. A]{Neri1}, the problem of mismatched estimation of a codeword, transmitted over an additive white Gaussian (AWGN) channel, was considered. It was shown that the mismatched MSE exhibits phase transitions at some rate thresholds, which depend upon the real and the mismatched parameters of the problem, and the behavior of the receiver. To wit, the mismatched MSE acts inherently differently for a \emph{pessimistic} and \emph{optimistic} receivers, where in the example considered in \cite[Section IV. A]{Neri1} pessimism literally means that the estimator assumes that the channel is worse than it really is (in terms of signal-to-noise ratio (SNR)), and the vice versa for optimism. In this paper, we extend the above described model to a much more general one; the Gaussian vector channel, which has a plenty of applications in communications and signal processing. It is important to emphasize that compared to \cite{Neri1,Neri2}, it will be seen that: (1) the mathematical analysis is much more complicated (consisting of some new concepts), and (2) the notions of pessimism and optimism described above, also play a significant role in this model, although their physical meanings in general are not obvious. Moreover, in contrast to previous work on mismatched estimation, in this paper, the interesting case of channel mismatch is explored, namely, the receiver has a wrong assumption on the channel. In order to demonstrate the usefulness of the theoretical results derived for the general model, we also provide a few examples associated with some specific channel transfer functions, and draw conclusions and insights regarding the threshold effects in the behavior of the partition function and the MSE.

As was mentioned earlier, we consider the Gaussian vector channel model
\begin{align}
\bY = \bA\bX+\bN,
\label{modint}
\end{align}
where $\bN\in\mathbb{R}^n$ is a Gaussian white noise vector and $\bA$ is a deterministic $n\times n$ matrix representing a linear transformation induced by a given linear system. The vector $\bX\in\mathbb{R}^n$ is chosen uniformly at random from a codebook (which is itself selected at random as well). There are several motivations for codeword estimation. One example is that of a user that, in addition to its desired signal, receives also a relatively strong interference signal, which carries digital information intended to other users, and which comes from a codebook whose rate exceeds the capacity of this crosstalk channel between the interferer and our user, so that the user cannot fully decode this interference. Nevertheless, our user would like to estimate the interference as accurately as possible for the purpose of cancellation. Furthermore, we believe that the tools/concepts developed in this paper for handling matched and mismatched problems, can be used in other applications in signal processing and communication. Such examples are denoising (see for example, \cite{Ordentlich,Gemelos,Jalali}), mismatched decoding (for example, \cite{Ganti}), blind deconvolution (for example, \cite{DonohoB,Majid}), and many other applications. Note that although the aforementioned examples are radically different (in terms of their basic models and systematization), they will all suffer from mismatch when estimating the input signals.  

In the special case of matched estimation, it will be shown that the MMSE is asymptotically given by
\begin{align}
\lim_{n\to\infty}\frac{\text{mmse}\p{\bX\mid\bY}}{n}
=\begin{cases}
\frac{1}{2\pi}\int_0^{2\pi}\frac{P_x}{1+\abs{\Hmat\p{\omega}}^2P_x\beta}\mathrm{d}\omega,\ \ &\text{if}\ \ R>R_c\\
0,\ \ &\text{if}\ \ R\leq R_c\\ 
\end{cases}
\end{align}
where 
\begin{align}
R_c\define \frac{1}{4\pi}\int_0^{2\pi}\ln\p{1+\abs{\Hmat\p{\omega}}^2P_x\beta}\mathrm{d}\omega,
\end{align}
in which $\text{mmse}\p{\bX\mid\bY}$ is the estimation error results from estimating $\bX$ based on $\bY$, using the MMSE estimator, $1/\beta$ and $P_x$ denote the noise variance and the transmitted power, respectively, and $\Hmat\p{\omega}$ is the frequency response of the linear system $\bA$. As can be seen from the above formula, for $R<R_c$ the MMSE essentially vanishes since the correct codeword can be reliably decoded, whereas for $R>R_c$, the MMSE is simply the estimation error which results by the Wiener filter that would have been applied had the input been a zero-mean, i.i.d. Gaussian process, with variance $1/\beta$. Accordingly, it will be seen that for $R>R_c$ the MMSE estimator is simply the Wiener filter. It is important to emphasize that while the above result may seem to be a natural generalization of the results in \cite{Neri1,Neri2} (where $\bA$ is taken to be identity matrix), the analysis (and results) of the mismatched case is by far more complicated and non-trivial. Indeed, it will be seen that in the mismatched case, the MSE is essentially separated into two cases, each exhibiting a completely different behavior. Further physical insights regarding the above result and other results will be presented later on. 

The remaining part of this paper is organized as follows. In Section \ref{sec:notation}, we first establish notation conventions. Then, the model considered is presented and the problem is formulated. In Section \ref{sec:body}, the main results are stated and discussed. In Section \ref{sec:exmp}, we provide a few examples which illustrate the theoretical results. In Section \ref{sec:proofOut}, we discuss the techniques and methodologies that are utilized in order to prove the main results, along with a brief background and summary on the basic relations between the conditional mean estimator, as well as its error covariance matrix and the aforementioned partition function, which were derived in \cite{Neri1}. In Section \ref{sec:proofs}, the main results are proved. Finally, our conclusions appear in Section \ref{sec:Conclusion}.

\section{Notation Conventions and Problem Formulation}\label{sec:notation}
\subsection{Notation Conventions}
Throughout this paper, scalar random variables (RV's) will be denoted by capital letters, their sample values will be denoted by the respective lower case letters and their alphabets will be denoted by the respective calligraphic letters. A similar convention will apply to random vectors and their sample values, which will be denoted with same symbols in the bold face font. Thus, for example, $\bX$ will denote a random vector $\p{X_1,\ldots,X_n}$ and $\bx = \p{x_1,\ldots,x_n}$ is a specific vector value in $\calX^n$, the $n$-th Cartesian power of $\calX$. The notations $\bx_i^j$ and $\bX_i^j$, where $i$ and $j$ are integers and $i\leq j$, will designate segments $\p{x_i,\ldots,x_j}$ and $\p{X_i,\ldots,X_j}$, respectively. Probability functions will be denoted generically by the letter $P$ or $P'$. In particular, $P\p{\bx,\by}$ is the joint probability mass function (in the discrete case) or the joint density (in the continuous case) of the desired channel input vector $\bx$ and the observed channel output vector $\by$. Accordingly, $P\p{\bx}$ will denote the marginal of $\bx$, $P\p{\by\mid\bx}$ will denote the conditional probability or density of $\by$ given $\bx$, induced by the channel, and so on. 

The expectation operator of a generic function $f\p{\bx,\by}$ with respect to (w.r.t.) the joint distribution of $\bX$ and $\bY$, $P\p{\bx,\by},$ will be denoted by $\bE\ppp{f\p{\bX,\bY}}$. Accordingly, $\bE'\ppp{f\p{\bX,\bY}}$ means that the expectation is performed w.r.t. $P'\p{\bx,\by}$. The conditional expectation of the same function given that $\bY = \by$, denoted $\bE\ppp{f\p{\bX,\bY}\mid \bY = \by}$ and which is obviously identical to $\bE\ppp{f\p{\bX,\by}\mid \bY = \by}$, is, of course, a function of $\by$. On substituting $\bY$ in this function, this becomes a random variable which will be denoted by $\bE\ppp{f\p{\bX,\bY}\mid\bY}$. When using vectors and matrices in a linear-algebraic format, $n$-dimensional vectors, like $\bx$ (and $\bX$), will be understood as column vectors, the operators $\p{\cdot}^T$ and $\p{\cdot}^H$ will denote vector or matrix transposition and vector or matrix conjugate transposition, respectively, and so, $\bx^T$ would be a row vector. For two positive sequences $\ppp{a_n}$ and $\ppp{b_n}$, the notation $a_n\exe b_n$ means equivalence in the exponential order, i.e., $\lim_{n\to\infty}\frac{1}{n}\log\p{a_n/b_n} = 0$. For two sequences $\ppp{a_n}$ and $\ppp{b_n}$, the notations $a_n\sim b_n$ and $a_n\lesssim b_n$ mean $\lim_{n\to\infty}\p{a_n/b_n} = 1$ and $\lim_{n\to\infty}\p{a_n/b_n} \leq 1$, respectively. Finally, the indicator function of an event $\calA$ will be denoted by $\Ind{\ppp{\calA}}$.

\subsection{Model and Problem Formulation}\label{sub:model}
Let $\calC = \ppp{\bx_0,\ldots,\bx_{M-1}}$ denote a codebook of size $M = e^{nR}$, which is selected at random (and then revealed to the estimator) in the following manner: Each $\bx_i$ is drawn independently under the uniform distribution over the surface of the $n$-dimensional hyperesphere, which is centered at the origin, and whose radius is $\sqrt{nP_x}$. Finally, let $\bX$ assume a uniform distribution over $\calC$. We consider the Gaussian vector channel model
\begin{align}
\bY = \bA\bX + \bN,
\label{MatModell}
\end{align}
where $\bY$, $\bX$ and $\bN$ are random vectors in $\mathbb{R}^n$, designating the channel output vector, the transmitted codeword and the noise vector, respectively. It is assumed that the components of the noise vector, $\bN$,  are i.i.d., zero-mean, Gaussian random variables with variance $1/\beta$, where $\beta$ is a given positive constant designating the signal-to-noise ratio (SNR) (for $P_x=1$), or the inverse temperature in the statistical-mechanical jargon. We further assume that $\bX$ and $\bN$ are statistically independent. Finally, the channel matrix, $\bA\in\mathbb{R}^{n\times n}$, is assumed to be a given deterministic Toeplitz matrix, whose entries are given by the coefficients of the impulse response of a given linear system. Specifically, let $\ppp{h_{k}}$ denote the generating sequence (or impulse response) of $\bA$, so that $\bA= \ppp{a_{i,j}}_{i,j}=\ppp{h_{i-j}}_{i,j}$, and let $\Hmat\p{\omega}$ designate the frequency response (Fourier transform) of $\ppp{h_k}$. 

As was mentioned previously, we analyze the problem of mismatched codeword estimation which is formulated as follows: Consider a mismatched estimator which is the conditional mean of $\bX$ given $\bY$, based on an incorrect joint distribution $P'\p{\bx,\by}$, whereas the true joint distribution continues to be $P\p{\bx,\by}$. Accordingly, the \emph{mismatched MSE} is defined as
\begin{align}
\text{mse}\p{\bX\mid\bY} &\define \bE\norm{\bX-\bE'\ppp{\bX\mid\bY}}^2
\end{align}
where $\bE'\ppp{\bX\mid\bY}$ is the conditional expectation w.r.t. the mismatched measure $P'$. In this paper, the following mismatch mechanism is assumed: The input measure is matched, i.e., $P\p{\bx} = P'\p{\bx}$ (namely, the mismatched estimator knows the true code), both conditional measures (``channels") $P\p{\cdot\mid\bx}$ and $P'\p{\cdot\mid\bx}$ are Gaussian, but are associated with different channel matrices. More precisely, while the true channel matrix (under $P$) is $\bA$, the assumed channel matrix (under $P'$) is $\bA'$, another Toeplitz matrix, generated by the impulse response $\ppp{h_k'}$, whose frequency response is $\Hmat'\p{\omega}$. It should be pointed out, however, that the analysis in this paper can be easily carried out also for the case of mismatch in the input distribution, or mismatch in the noise distribution, which has been already considered in \cite{Neri1}. Using the theoretical tools derived in \cite{Neri1}, the mismatched MSE (and the MMSE as a special case) will be derived for the model described above.

A very important function, which will be pivotal to our derivation of both the mismatched estimator and the MSE, is the \emph{partition function}, which is defined as follows.
\begin{definition}[Partition Function]\label{def:1}
Let $\blam = \p{\lambda_1,\ldots,\lambda_n}^T$ be a column vector of $n$ real-valued parameters. The partition function w.r.t. the joint distribution $P\p{\bx,\by}$, denoted by $Z\p{\by,\blam}$, is defined as
\begin{align}
Z\p{\by,\blam}&\define\sum_{\bx\in\calX^n}\exp\ppp{\blam^T\bx}P\p{\bx,\by}.
\label{PartFunc}
\end{align} 
\end{definition}
In the above definition, it is assumed that the sum (or integral, in the continuous case) converges uniformly at least in some neighborhood of $\blam = \bze$ \footnote{In case that this assumption does not hold, one can instead, parametrize each component $\lambda_i$ of $\blam$ as a purely imaginary number $\lambda_i = j\omega_i$ where $i=\sqrt{-1}$, similarly to the definition of the characteristics function.}. Accordingly, under the above described model, the mismatched partition function is given by
\begin{align}
Z'\p{\by,\blam}&\define\sum_{\bx\in\calC}\exp\ppp{\blam^T\bx}P'\p{\bx,\by}\\
&=\p{2\pi/\beta}^{-n/2}\sum_{\bx\in\calC}e^{-nR}\exp\pp{-\beta\norm{\by-\bA'\bx}^2/2+\blam^T\bx}.
\label{PartFuncmod}
\end{align} 
\begin{remark}
In the above definition, the role of $\blam$ will be understood later on. In a nutshell, the idea \cite{Neri1} is that the gradient of $\ln Z'\p{\by,\blam}$ w.r.t. $\blam$, computed at $\blam=\bze$, simply gives the mismatched MSE estimator, $\bE'\ppp{\bX\mid\by}$, and the expectation of the Hessian of $\ln Z'\p{\by,\blam}$ w.r.t. $\blam$, computed at $\blam=\bze$, gives the MSE. Nevertheless, in the next section, where we present the main results, the dependency of the different quantities in $\blam$ will not be apparent, as they will already be computed at $\blam=\bze$.  
\end{remark}

\section{Main Results and Discussion}\label{sec:body}

\allowdisplaybreaks

In this section, our main results are presented and discussed. The proofs of these results are provided in Section \ref{sec:proofs}. The asymptotic MMSE, which is obtained as a special case of the mismatched case ($P=P'$), is given in the following theorem.
\begin{theorem}[Asymptotic MMSE]\label{th:1} 
Consider the model defined in Subsection \ref{sub:model}, and assume that the sequence $\ppp{h_k}_k$ is square summable. Then, the asymptotic MMSE is given by
\begin{align}
\lim_{n\to\infty}\frac{\text{mmse}\p{\bX\mid\bY}}{n}=
\begin{cases}
\frac{1}{2\pi}\int_0^{2\pi}\frac{P_x}{1+\abs{\Hmat\p{\omega}}^2P_x\beta}\mathrm{d}\omega,\ \ &\ \ R>R_c\\
0,\ \ &\ \ R\leq R_c\\ 
\end{cases}
\label{finalTheRes}
\end{align}
where
\begin{align}
R_c\define \frac{1}{4\pi}\int_0^{2\pi}\ln\p{1+\abs{\Hmat\p{\omega}}^2P_x\beta}\mathrm{d}\omega.
\end{align}
\end{theorem}

From the above result, it can be seen that for $R>R_c$ the MMSE is simply the estimation error which results by the Wiener filter that would have been applied had the input been a zero-mean, i.i.d. Gaussian process, with variance $1/\beta$. Accordingly, it is also shown in Section \ref{sec:proofs} that the MMSE estimator is exactly the Wiener filter. 

In the next theorem, we present the mismatched MSE. In contrast to the MMSE, unfortunately, the MSE does not lend itself to a simple closed-form expression. As will be seen in Section \ref{sec:proofs}, this complexity stems from the complicated dependence of the partition function on $\blam$. Nevertheless, despite of the following non-trivial expressions, it should be emphasized that the obtained MSE expression has a single-letter formula, and thus, practically, it can be easily calculated at least numerically. Let us define the following auxiliary variables
\begin{align}
P_a\p{\omega} \define& \frac{\abs{\Hmat'\p{\omega}}^2\beta\p{2+P_x\beta\abs{\Hmat\p{\omega}}^2}+\gamma_0}{\p{\abs{\Hmat'\p{\omega}}^2\beta+\gamma_0}^2}
\label{psiM}
\end{align}
where $\gamma_0$ is chosen such that $\int_0^{2\pi}P_a\p{\omega}\mathrm{d}\omega = 2\pi P_x$. Next define
\begin{align}
&\Bmat\p{\omega} \define \frac{\p{\abs{\Hmat'\p{\omega}}^2+\gamma_0}-2\p{\abs{\Hmat'\p{\omega}}^2\beta\p{2+P_x\beta\abs{\Hmat\p{\omega}}^2}+\gamma_0}}{\p{\abs{\Hmat'\p{\omega}}^2+\gamma_0}^3}\\
&\Cmat\p{\omega} \define \frac{2\beta^2\abs{\Hmat'\p{\omega}}^2\p{\abs{\Hmat\p{\omega}}^2+\frac{1}{\beta}}\Bmat\p{\omega}}{\sqrt{1+4\beta^2\abs{\Hmat'\p{\omega}}^2P_a\p{\omega}\p{\abs{\Hmat\p{\omega}}^2+\frac{1}{\beta}}}}\\
&\vartheta \define 2+\frac{\int_0^{2\pi}\pp{\frac{P_x}{P_x\gamma_0+P_x\abs{\Hmat'\p{\omega}}^2\beta}-\Cmat\p{\omega}}\mathrm{d}\omega}{\int_0^{2\pi}\Bmat\p{\omega}\mathrm{d}\omega},
\end{align}
and
\begin{align}
&\bXi_1\p{\omega} \define -\frac{\beta\Hmat'^*\p{\omega}}{\p{\abs{\Hmat'\p{\omega}}^2+\gamma_0}^2}\pp{\vartheta-\frac{2\p{\abs{\Hmat'\p{\omega}}^2+\gamma_0}^2P_a\p{\omega}+2\beta^2\abs{\Hmat'\p{\omega}}^2\p{\abs{\Hmat\p{\omega}}^2+\frac{1}{\beta}}}{\sqrt{1+4\beta^2\abs{\Hmat'\p{\omega}}^2P_a\p{\omega}\p{\abs{\Hmat\p{\omega}}^2+\frac{1}{\beta}}}}}.
\label{Chi1}
\end{align}
Let $\epsilon_{s,0}$, $\alpha_{1,0}$ and $\alpha_{2,0}$ be the solution of the following set of three simultaneous equations:
\begin{align}
\label{seteq1}
&R+\frac{1}{4\pi}\int_0^{2\pi}\ln\p{\frac{2\epsilon_{s,0}}{P_x\abs{\Hmat'\p{\omega}}^2\alpha_{2,0}+2P_x\alpha_{1,0}\epsilon_{s,0}}}\mathrm{d}\omega=0\\
&\frac{1}{2\pi}\int_0^{2\pi}\frac{4\alpha_{1,0}\epsilon_{0,s}^2+\abs{\Hmat'\p{\omega}}^2\alpha_{2,0}\pp{\p{\abs{\Hmat\p{\omega}}^2P_x+\frac{1}{\beta}}\alpha_{2,0}+2\epsilon_{s,0}}}{\p{\abs{\Hmat'\p{\omega}}^2\alpha_{2,0}+2\alpha_{1,0}\epsilon_{s,0}}^2}\mathrm{d}\omega = P_x\\
&\frac{1}{2\pi}\int_0^{2\pi}\frac{4\alpha_{1,0}^2\epsilon_{s,0}^2\p{1+P_x\beta\abs{\Hmat\p{\omega}}^2}+4\abs{\Hmat'\p{\omega}}^2\alpha_{1,0}\epsilon_{s,0}^2\beta+2\abs{\Hmat'\p{\omega}}^4\alpha_{2,0}\epsilon_{s,0}\beta}{2\beta\epsilon_{s,0}\p{\abs{\Hmat'\p{\omega}}^2\alpha_{2,0}+2\alpha_{1,0}\epsilon_{s,0}}^2}\mathrm{d}\omega = 1.
\label{seteq3}
\end{align}
Then, we define
\begin{align}
&\Kmat\p{\omega} \define 2\beta\epsilon_{s,0}\p{\abs{\Hmat'\p{\omega}}^2\alpha_{2,0}+2\alpha_{1,0}\epsilon_{s,0}}^2\\
&\Tmat\p{\omega} \define 4\alpha_{1,0}^2\epsilon_{s,0}^2\p{1+P_x\beta\abs{\Hmat\p{\omega}}^2}+4\beta\abs{\Hmat'\p{\omega}}^2\alpha_{1,0}\epsilon_{s,0}^2+2\abs{\Hmat'\p{\omega}}^4\alpha_{2,0}\beta\epsilon_{s,0}\\
&\Dmat\p{\omega} \define \abs{\Hmat'\p{\omega}}^2\alpha_{2,0}+2\alpha_{1,0}\epsilon_{s,0}\\
&\Rmat\p{\omega} \define 4\alpha_{1,0}\epsilon_{s,0}^2 + \abs{\Hmat'\p{\omega}}^2\alpha_{2,0}\pp{\alpha_{2,0}\p{\abs{\Hmat\p{\omega}}^2P_x+2\epsilon_{s,0}}}\\
&\Qmat\p{\omega} \define \epsilon_{s,0}\p{P_x\abs{\Hmat'\p{\omega}}^2\alpha_{2,0}+2P_x\alpha_{1,0}\epsilon_{s,0}}\\
&V\define \frac{1}{2\pi}\int_0^{2\pi}\frac{P_x\abs{\Hmat'\p{\omega}}^2\alpha_{2,0}}{\epsilon_{s,0}\p{P_x\abs{\Hmat'\p{\omega}}^2\alpha_{2,0}+2P_x\alpha_{1,0}\epsilon_{s,0}}}\mathrm{d}\omega\\
&F \define \frac{1}{V}\frac{1}{2\pi}\int_0^{2\pi}\frac{P_x\abs{\Hmat'\p{\omega}}^2r_2 + 2P_x\epsilon_{s,0}r_1}{P_x\abs{\Hmat'\p{\omega}}^2\alpha_{2,0}+2P_x\alpha_{1,0}\epsilon_{s,0}}\mathrm{d}\omega\\
&\gamma_1 \define \frac{1}{2\pi}\int_0^{2\pi}\left[\frac{8\alpha_{1,0}\epsilon_{s,0}^2\p{1+P_x\beta\abs{\Hmat\p{\omega}}^2}+4\beta\abs{\Hmat'\p{\omega}}^2\epsilon_{s,0}^2}{\Kmat\p{\omega}}\right.\nonumber\\
& \ \ \ \ \ \ \ \ \ \ \ \ \ \ \ \ \ \  \left.-\frac{8\Tmat\p{\omega}\beta\epsilon_{s,0}^2\p{\abs{\Hmat'\p{\omega}}^2\alpha_{2,0}+2\alpha_{1,0}\epsilon_{s,0}}}{\Kmat^2\p{\omega}}\right]\mathrm{d}\omega\\
&\gamma_2 \define \frac{1}{2\pi}\int_0^{2\pi}\frac{2\Kmat\p{\omega}\beta\epsilon_{s,0}\abs{\Hmat'\p{\omega}}^4-4\Tmat\p{\omega}\beta\epsilon_{s,0}\p{\abs{\Hmat'\p{\omega}}^2\alpha_{2,0}+2\alpha_{1,0}\epsilon_{s,0}}\abs{\Hmat'\p{\omega}}^2}{\Kmat^2\p{\omega}}\mathrm{d}\omega\\
&\gamma_3 \define \frac{1}{2\pi}\int_0^{2\pi}\left[\frac{8\alpha_{1,0}^2\epsilon_{s,0}\p{1+P_x\beta\abs{\Hmat\p{\omega}}^2}+8\beta\epsilon_{s,0}\abs{\Hmat'\p{\omega}}^2\alpha_{1,0}+2\beta\alpha_{2,0}\abs{\Hmat'\p{\omega}}^4}{\Kmat\p{\omega}}\right.\nonumber\\
&\left.-\frac{\Tmat\p{\omega}\pp{2\beta\p{\abs{\Hmat'\p{\omega}}^2\alpha_{2,0}+2\alpha_{1,0}\epsilon_{s,0}}^2+8\beta\epsilon_{s,0}\alpha_{1,0}\p{\abs{\Hmat'\p{\omega}}^2\alpha_{2,0}+2\alpha_{1,0}\epsilon_{s,0}}}}{\Kmat^2\p{\omega}}\right]\mathrm{d}\omega\\
&\Upsilon\p{\omega} \define \frac{-4\beta\alpha_{1,0}\epsilon_{s,0}\alpha_{2,0}-\beta\alpha_{2,0}^2\abs{\Hmat'\p{\omega}}^2}{\Kmat^2\p{\omega}}\\
&\eta_1 \define \frac{1}{2\pi}\int_0^{2\pi}\frac{4\Dmat\p{\omega}\epsilon_{s,0}^2-4\Rmat\p{\omega}\epsilon_{s,0}}{\Dmat^3\p{\omega}}\mathrm{d}\omega\\
&\eta_2 \define \frac{1}{2\pi}\int_0^{2\pi}\left[\frac{\abs{\Hmat'\p{\omega}}^2\pp{\p{\abs{\Hmat\p{\omega}}^2P_x+\frac{1}{\beta}}\alpha_{2,0}+2\epsilon_{s,0}}}{\Dmat^2\p{\omega}}\right.\nonumber\\
&\left.\ \ \ \ \  \ \ \ \ \ \ \ \ +\frac{\abs{\Hmat'\p{\omega}}^2\alpha_{2,0}\p{\abs{\Hmat\p{\omega}}^2P_x+\frac{1}{\beta}}\Dmat\p{\omega}-2\Rmat\p{\omega}\abs{\Hmat'\p{\omega}}^2}{\Dmat^3\p{\omega}}\mathrm{d}\omega\right]\\
&\eta_3 \define \frac{1}{2\pi}\int_0^{2\pi}\frac{8\Dmat\p{\omega}\alpha_{1,0}\epsilon_{s,0}+2\Dmat\p{\omega}\abs{\Hmat'\p{\omega}}^2\alpha_{2,0}-4\Rmat\p{\omega}\alpha_{1,0}}{\Dmat^3\p{\omega}}\mathrm{d}\omega\\
&\Lambda\p{\omega} \define \frac{\alpha_{2,0}^2}{\Dmat^2\p{\omega}}\\
&r_1  = \frac{\eta_2\gamma_3-\gamma_2\eta_3}{\gamma_2\eta_1-\eta_2\gamma_1}\\
&r_2  = \frac{\eta_1\gamma_3-\gamma_1\eta_3}{\gamma_1\eta_2-\eta_1\gamma_2}\\
&\Jmat_1\p{\omega} = \frac{\eta_2\Upsilon\p{\omega}-\gamma_2\Lambda\p{\omega}}{\gamma_2\eta_1-\eta_2\gamma_1}\\
&\Jmat_2\p{\omega} = \frac{\eta_1\Upsilon\p{\omega}-\gamma_1\Lambda\p{\omega}}{\gamma_1\eta_2-\eta_1\gamma_2}\\
&\Jmat\p{\omega}\define\frac{1}{2\pi}\frac{\Jmat_1\p{\omega}\int_0^{2\pi}\frac{2\epsilon_{s,0}^2P_x}{\Qmat\p{\omega}}\mathrm{d}\omega+\Jmat_2\p{\omega}\int_0^{2\pi}\frac{2\epsilon_{s,0}P_x\abs{\Hmat'\p{\omega}}^2}{\Qmat\p{\omega}}\mathrm{d}\omega}{V\p{1-F}}\\
&\bXi_2\p{\omega}\define -2\Jmat\p{\omega}\Hmat'^*\p{\omega}.
\label{Chi2}
\end{align}
Finally, let
\begin{align}
E_{g} \define P_x-\re\p{\frac{1}{\pi}\int_0^{2\pi}\bXi_{2}\p{\omega}^*\Hmat^*\p{\omega}P_x\mathrm{d}\omega}+\frac{1}{2\pi}\int_0^{2\pi}\abs{\bXi_2\p{\omega}}^2\p{\abs{\Hmat\p{\omega}}^2P_x+\frac{1}{\beta}}\mathrm{d}\omega,
\end{align}
and
\begin{align}
E_{p} \define P_x-\re\p{\frac{1}{\pi}\int_0^{2\pi}\bXi_1\p{\omega}^*\Hmat^*\p{\omega}P_x\mathrm{d}\omega}+\frac{1}{2\pi}\int_0^{2\pi}\abs{\bXi_1\p{\omega}}^2\p{\abs{\Hmat\p{\omega}}^2P_x+\frac{1}{\beta}}\mathrm{d}\omega,
\end{align}
and we define the following critical rates
\begin{align}
R_e\define&\frac{1}{4\pi}\int_0^{2\pi}\ln\p{P_x\gamma_0+P_x\beta\abs{\Hmat'\p{\omega}}^2}\mathrm{d}\omega\\
R_d\define&\frac{1}{2} + \beta P_x\frac{1}{2\pi}\int_0^{2\pi}\re\p{\Hmat'^{*}\p{\omega}\Hmat\p{\omega}}\mathrm{d}\omega\nonumber\\
&+\frac{1}{4\pi}\int_0^{2\pi}\abs{\Hmat'\p{\omega}}^2\beta\p{\frac{\abs{\Hmat'\p{\omega}}^2\beta\p{2+\beta\abs{\Hmat\p{\omega}}^2P_x}+\gamma_0}{\p{\abs{\Hmat'\p{\omega}}^2\beta+\gamma_0}^2}-P_x}\nonumber\\
&-\frac{1}{4\pi}\int_0^{2\pi}\frac{\abs{\Hmat'\p{\omega}}^2\beta\p{3+2P_x\beta\abs{\Hmat\p{\omega}}^2}+\gamma_0}{\p{\abs{\Hmat'\p{\omega}}^2\beta+\gamma_0}}\\
R_c\define& R_e+R_d\\
R_g\define&-\frac{1}{4\pi}\int_0^{2\pi}\ln\p{\frac{2\tilde{\epsilon}}{P_x\abs{\Hmat'\p{\omega}}^2\tilde{\alpha}_2+2P_x\tilde{\alpha}_1\tilde{\epsilon}}}\mathrm{d}\omega
\label{R_gDefcont}
\end{align}
where $\tilde{\alpha}_1$ and $\tilde{\alpha}_2$ solve the set of two simultaneous equations
\begin{align}
&\frac{1}{2\pi}\int_0^{2\pi}\frac{4\tilde{\alpha}_1\tilde{\epsilon}^2+\abs{\Hmat'\p{\omega}}^2\tilde{\alpha}_2\pp{\p{\abs{\Hmat\p{\omega}}^2P_x+\frac{1}{\beta}}\tilde{\alpha}_2+2\tilde{\epsilon}}}{\p{\abs{\Hmat'\p{\omega}}^2\tilde{\alpha}_2+2\tilde{\alpha}_1\tilde{\epsilon}}^2}\mathrm{d}\omega = P_x\\
&\frac{1}{2\pi}\int_0^{2\pi}\frac{4\tilde{\alpha}_1^2\tilde{\epsilon}^2\p{1+P_x\beta\abs{\Hmat\p{\omega}}^2}+4\abs{\Hmat'\p{\omega}}^2\tilde{\alpha}_1\tilde{\epsilon}^2\beta+2\abs{\Hmat'\p{\omega}}^4\tilde{\alpha}_2\tilde{\epsilon}\beta}{2\beta\tilde{\epsilon}\p{\abs{\Hmat'\p{\omega}}^2\tilde{\alpha}_2+2\tilde{\alpha}_1\tilde{\epsilon}}^2}\mathrm{d}\omega = 1.
\label{twoequations}
\end{align}
and
\begin{align}
\tilde{\epsilon} = \frac{1}{2\beta}+\frac{P_x}{4\pi}\int_0^{2\pi}\abs{\Hmat'\p{\omega}-\Hmat\p{\omega}}^2\mathrm{d}\omega.
\label{tildeepsilon}
\end{align}

We are now in a position to state our main theorem.
\begin{theorem}[Mismatched MSE]\label{th:2}
Consider the model defined in Subsection \ref{sub:model}, and assume that the sequence $\ppp{h_k}_k$ is square summable. The (asymptotic) mismatched MSE is given as follows:
\newline a) For $R_d\geq0$ 
 \begin{align}
&\lim_{n\to\infty}\frac{\text{mse}\p{\bX\mid\bY}}{n} =
\begin{cases}
0,\ \ &\ \ R\leq R_c\\
E_{p},\ \ &\ \ R> R_c 
\end{cases}.
\end{align}
b) For $R_d<0$ 
\begin{align}
&\lim_{n\to\infty}\frac{\text{mse}\p{\bX\mid\bY}}{n}=
\begin{cases}
0,\ \ &\ \ R\leq R_g\\
E_{g},\ \ &\ \ R_g<R\leq R_e\\ 
E_{p},\ \ &\ \ R> R_e
\end{cases}.
\end{align}
\end{theorem}

In the jargon of statistical mechanics of spin arrays (see for example \cite[Ch. 6]{Mezard}), the ranges of rates $R\leq R_c$ for $R_d\geq0$, $R\leq R_g$ for $R_d<0$, and $R\leq R_c$ in the matched case, correspond to the ordered phase (or ferromagnetic phase) in which the partition function is dominated by the correct codeword (and hence so is the posterior). Accordingly, in this range the MSE asymptotically vanishes, which literally means reliable communication. The intermediate range, $R_g<R\leq R_e$, which appears only in the mismatched case and only for $R_d<0$, is analogous to the glassy phase (or ``frozen" phase), in which the partition function is dominated by a sub-exponential number of wrong codewords. Intuitively, in this range, we may have the illusion that there is relatively little uncertainty about the transmitted codeword, but this is wrong due to the mismatch (as the main support of the mismatched posterior belongs to incorrect codewords). The remaining range corresponds to the paramagnetic phase, in which the partition function is dominated by an exponential number of wrong codewords. In Section \ref{sec:exmp}, we will link between each one of the two cases $R_d\geq0$ and $R_d<0$, to ``pessimistic" and ``optimistic" behaviors of the receiver, which were already mentioned in the Introduction. 

It is tempting to think that there should not be a range of rates for which the MSE (MMSE) vanishes, as we deal with an estimation problem rather than a decoding problem. Nonetheless, since codewords are being estimated, and there are a finite number of them, for low enough rates (up to some critical rate) the posterior is dominated by the correct codeword, and thus asymptotically, the estimation can be regarded as a maximum a posteriori probability (MAP) estimation, and so the MSE vanishes. In the same breath, note that this is not the case if mismatch in the input distribution is considered. For example, if the receiver's assumption on the transmitted energy is wrong, then no matter how low the rate is, there will always be an inherent error which stems from the fallacious averaging over a hypersphere with wrong radius (wrong codebook). Precisely, in this case, the estimated codeword will differ from the real one by an inevitable scaling of $\sqrt{P_x'/P_x}$, where $P_x'$ is the mismatched power. 

Finally, it is important to emphasize that the mismatched MSE estimator and the MMSE estimator can also be obtained as a byproduct of the analysis. However, since they will add only little further insights into the problem, we do not present them here. The interested reader can find their explicit expressions in Section \ref{sec:proofs}. 

\begin{remark}
Although we have assumed that the transmitted codeword has a flat spectrum, the analysis can readily be extended to any input spectral density $S_x\p{\omega}$. In Section \ref{sec:proofs}, we discuss the technical issues that should be considered in order to modify the analysis to hold for this generalization. As a concrete simple example, in the case of MMSE estimation, one obtains
\begin{align}
\lim_{n\to\infty}\frac{\text{mmse}\p{\bX\mid\bY}}{n}
=\begin{cases}
\frac{1}{2\pi}\int_0^{2\pi}\frac{S_x\p{\omega}}{1+\abs{\Hmat\p{\omega}}^2S_x\p{\omega}\beta}\mathrm{d}\omega,\ \ &\ \ R>R_c\\
0,\ \ &\ \ R\leq R_c\\ 
\end{cases}
\label{finalTheRes2}
\end{align}
where
\begin{align}
R_c\define \frac{1}{4\pi}\int_0^{2\pi}\ln\p{1+\abs{\Hmat\p{\omega}}^2S_x\p{\omega}\beta}\mathrm{d}\omega.
\end{align}
Nevertheless, our assumption on flat input spectrum is reasonable when there is uncertainty at the encoder concerning the frequency response of the channel, as there are no ``preferred" frequencies. Finally, note that as an application of the above issue, one may wish to consider the minimization of the MMSE w.r.t. the input spectral density.  
\end{remark}


\section{Examples}\label{sec:exmp}
In this section, we provide a few examples in order to illustrate the theoretical results presented in the previous section. In particular, we present and explore the phase diagrams and the MSE's as functions of the rate and some parameters of the mismatched channel. The main goal in these examples is further understanding of the role of the true and the mismatched probability measures in creating phase transitions. 

\begin{example}\label{exmp2}
We start with a simple example where both $\Hmat\p{\omega}$ and $\Hmat'\p{\omega}$ are low-pass filters (LPFs) that differ in their cutoff frequencies and gains
\begin{align}
\Hmat\p{\omega} = 
\begin{cases}
1, \ &\abs{\omega}\leq\frac{\pi}{2}\\
0, \ &\text{else}
\end{cases},
\end{align}
and
\begin{align}
\Hmat'\p{\omega} = 
\begin{cases}
\chi, \ &\abs{\omega}\leq\omega_c\\
0, \ &\text{else}
\end{cases}
\end{align}
for some $\chi>0$ and $0\leq\omega_c\leq\pi$. In the numerical calculations, we chose $\beta = P_x =1$. Figures \ref{fig:2} and \ref{fig:2MSE} show, respectively, the phase diagrams and the MSE's as functions of $R$ and $\omega_c$, for various values of the gain $\chi$. The first obvious observation is that the maximum range of rates for which the ferromagnetic phase dominates the partition function occurs at $\omega_c = \pi/2$ for each gain, as expected. Next, consider the case of $\chi=1$, which means that the gain is matched. In this case, it is observed that for $\omega_c\leq\pi/2$, there are two phases: the ferromagnetic phase and the paramagnetic phase, and hence, based on Theorem \ref{th:2}, $R_d\geq0$. On the other hand, for $\omega_c>\pi/2$, the glassy phase begins to play a role, and thus $R_d<0$. Intuitively speaking, the case of $\omega_c\leq\pi/2$ corresponds to a pessimistic assumption of the receiver - lower bandwidth which translates to lower effective SNR, while $\omega_c>\pi/2$ corresponds to an optimistic assumption - higher effective SNR. These behaviors are consistent with the results obtained in \cite{Neri1}, where the case of mismatch in the noise variance was considered (while assuming that $\bA=\bA'$ is the identity matrix). 

In \cite{Neri1}, $R_d>0$ simply translates to $\beta>\beta'$ (the mismatched noise variance is larger than the actual one), namely, the estimator is pessimistic, while in the case of the reversed inequality it is overly optimistic. Accordingly, in the pessimistic case, the partition function exhibits a single phase transition, but at the price of a lower critical rate (compared to the matched case), which means that the range of rates for which reliable communication is possible is smaller. In the optimistic case, however, there is no loss in the critical rate, but there is a price of an additional phase transition. Now, for $\chi\neq1$, the notions of pessimism and optimism are not a priori obvious. For example, it can be seen that for $\chi<1$, and for a large enough cutoff frequency $\omega_c$, the mismatched estimator can be regarded as an optimistic one. Also, for $\chi>1$, apparently, the ``price" of being too optimistic in the gain results in a dominant range of the glassy phase. Finally, note that the fact that the range of rates for which the ferromagnetic region dominates the partition function (namely, vanishing MSE) is decreasing with the excess of the optimism (e.g., for $\chi=1$ and increasing of the cutoff frequency) is reasonable\footnote{In \cite{Neri1}, in contrast to our case, for $\beta<\beta'$ ($R_d<0$), the critical rate $R_e$ is fixed for any mismatched noise variance value, namely, it is independent of the optimistic behavior of the receiver.}. Indeed, the uncertainty in the frequency domain, causes the receiver to assume that the codewords are distributed in some subspace of the $n$-dimensional hypersphere. The size of this subspace is, of course, increasing as the receiver's assumption is more optimistic. Accordingly, the probability of error also increases, and thus the threshold rate for reliable communication decreases. 

\begin{figure}[!t]
\begin{minipage}[b]{1.0\linewidth}
  \centering
	\centerline{\includegraphics[width=15cm,height = 12cm]{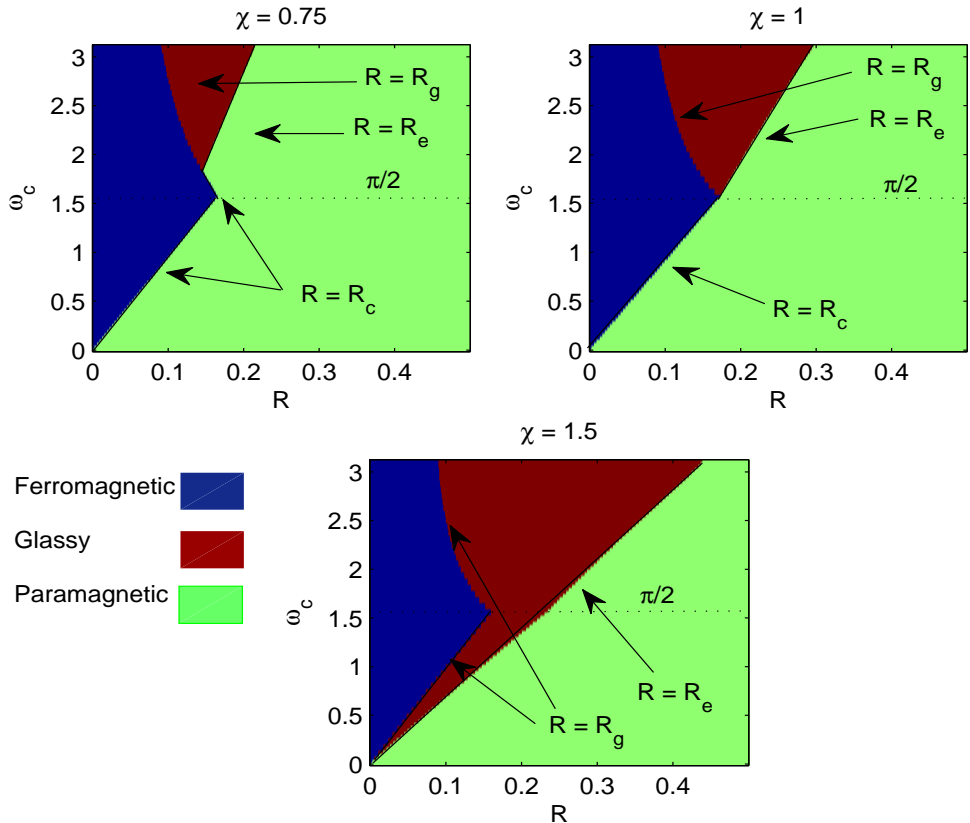}}
	
\end{minipage}
\caption{Example \ref{exmp2}: Phase diagram in the plane of $R$ vs. $\omega_c$ with various gain values. The arrows are directed towards the boundaries of the various phase transitions.}
\label{fig:2}
\end{figure}

\begin{figure}[!t]
\begin{minipage}[b]{1.0\linewidth}
  \centering
	\centerline{\includegraphics[width=15cm,height = 10cm]{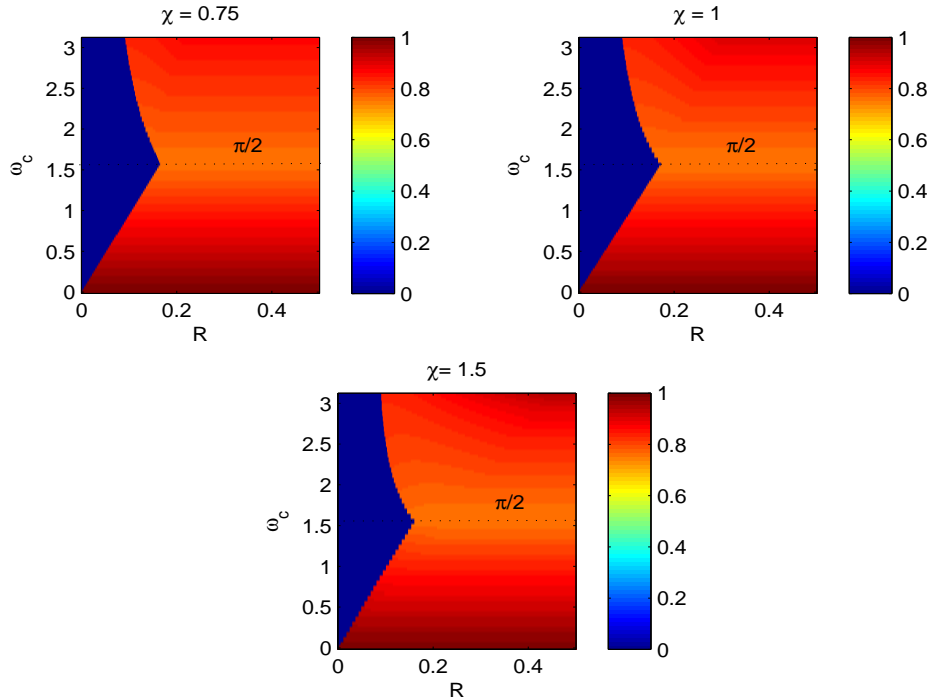}}
\end{minipage}
\caption{Example \ref{exmp2}: Mismatched MSE as a function of $R$ and $\omega_c$ with various gain values.}
\label{fig:2MSE}
\end{figure}

\end{example}

\begin{example}\label{exmp:2}
Let $\Hmat\p{\omega}$ be a multiband filter given by
\begin{align}
\Hmat\p{\omega} = 
\begin{cases}
1, \ &\abs{\omega\pm\frac{3\pi}{8}}\leq\frac{\pi}{8}\;\text{or}\;\abs{\omega\pm\frac{7\pi}{8}}\leq\frac{\pi}{8}\\
0, \ &\text{else}
\end{cases},
\end{align}
and let the mismatched filter be given by a band-pass filter
\begin{align}
\Hmat'\p{\omega} = 
\begin{cases}
1, \ &\omega_L\leq\abs{\omega}\leq\omega_R\\
0, \ &\text{else}
\end{cases},
\end{align}
with constant bandwidth, $\omega_R-\omega_L = \pi/8$, i.e., smaller than the real one. In the numerical calculations, we again chose $\beta=P_x=1$. Figures \ref{fig:4} and \ref{fig:4MSE} show, respectively, the phase diagram and the MSE as functions of $R$ and $\omega_L$. First, observe that for $\omega_R<\pi/4$, which means that $\Hmat'\p{\omega}$ and $\Hmat\p{\omega}$ are equal to one over non intersecting frequency ranges, there is no ferromagnetic phase, as expected. Accordingly, for $\omega_R>\pi/4$, the ferromagnetic phase begins to play a role, and it can be seen that for $\pi/4+\pi/8<\omega_R<\pi/2$, which means maximal intersection between the two filters, the range of rates for which the ferromagnetic phase dominates the partition function is maximal. Since the matched filter has two bands, obviously, the same behavior appears also in the second band. Thus, in this example, we actually obtain two disjoint glassy (and ferromagnetic) regions, which correspond to the two bands of the matched filter. Also, as shown in Fig. \ref{fig:4MSE}, in the ranges where no ferromagnetic phase exists, the MSE within the paramagnetic phase is larger than the MSE within the regions where ferromagnetic phase does exists, as one would expect. 
\begin{remark}
Example \ref{exmp:2} actually demonstrates that there can be arbitrarily many phase transitions. Generally speaking, for a matched multiband filter with $N$ disjoint bands, and a mismatched bandpass filter (with small enough bandwidth), there are $N$ disjoint glassy and ferromagnetic phases.
\end{remark}
\begin{figure}[!t]
\begin{minipage}[b]{1.0\linewidth}
  \centering
  \centerline{\includegraphics[width=14cm,height = 10cm]{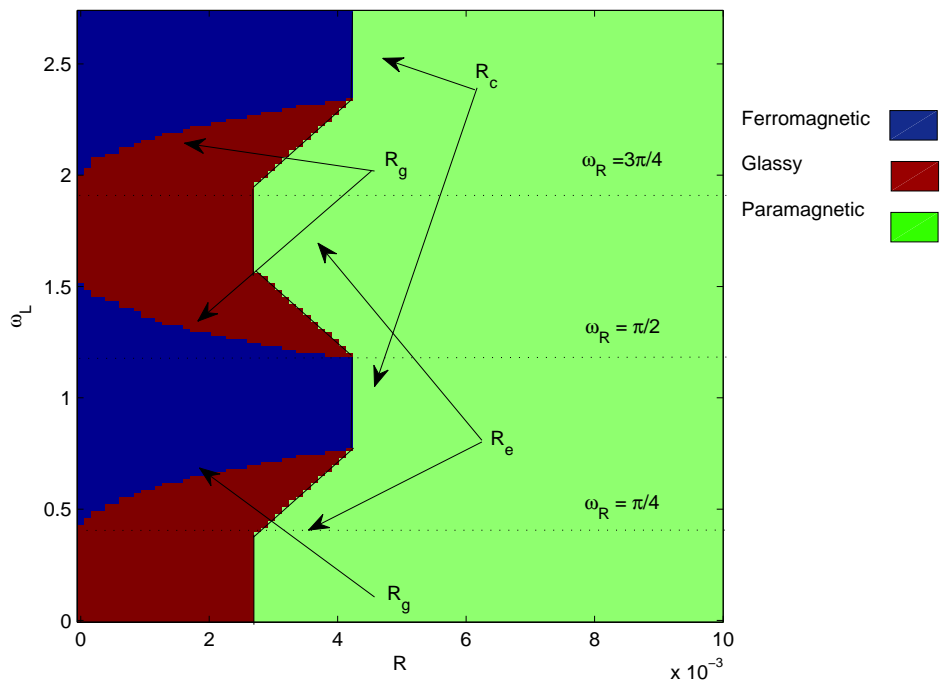}}
\end{minipage}
\caption{Example \ref{exmp:2}: Phase diagram in the plane of $R$ vs. $\omega_L$.}
\label{fig:4}
\end{figure}
\begin{figure}[!t]
\begin{minipage}[b]{1.0\linewidth}
  \centering
  \centerline{\includegraphics[width=14cm,height = 10cm]{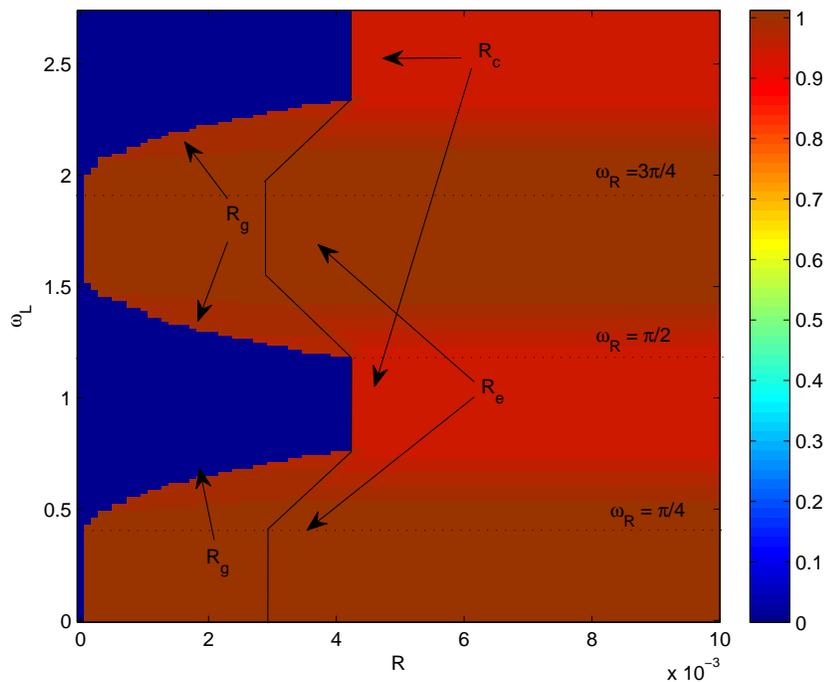}}
\end{minipage}
\caption{Example \ref{exmp:2}: Mismatched MSE as a function of $R$ and $\omega_L$.}
\label{fig:4MSE}
\end{figure}

\end{example}

\begin{example}\label{exmp3}
In this example, we consider more realistic filters. Let $\Hmat\p{z}$ denote a Type-II FIR filter given by (in the $\mathcal{Z}$ domain)
\begin{align}
\Hmat\p{z} = \p{1-e^{j0.8\pi}z^{-1}}^2\p{1-e^{-j0.8\pi}z^{-1}}^2,
\end{align}
and let the mismatched filter be given is
\begin{align}
\Hmat'\p{z} = \p{1-z_0z^{-1}}\p{1-z_0^*z^{-1}}\p{1-e^{j0.8\pi}z^{-1}}\p{1-e^{-j0.8\pi}z^{-1}}
\end{align}
where $z_0$ is a mismatched zero. In the numerical calculations, we chose again $\beta=P_x=1$. Fig. \ref{fig:amp} shows the amplitude response of the real and the mismatched filters for various angular frequencies defined as $\phi\define\arg\p{z_0}$. Figures \ref{fig:FIR} and \ref{fig:FIR2} show, respectively, the phase diagram and the MSE as functions of $R$ and $\phi$. In this example, the roles of the differences between the true and mismatched filters, are emphasized. Starting with the obvious, observe that the maximal range of rates for which the ferromagnetic region dominates the partition function occurs at $\phi = 0.8\pi$, as expected. Less trivially, for angular frequencies within the range $\pp{0.2\pi,0.25\pi}$, the ferromagnetic region is negligible. Looking at Fig. \ref{fig:amp}, it can be seen that within this range of angular frequencies, the true and the mismatched filters are ``almost orthogonal" in the $L_2$ sense, namely, their inner product is almost zero. Accordingly, using the methods in Section \ref{sec:proofs}, it can be easily shown that for orthogonal filters we have that $R_g=0$, namely, no ferromagnetic region exists (note that in this example, $R_g$ is never equal to zero since the filters are never orthogonal). Finally, for angular frequencies within the range $\pp{0,0.2\pi}$, the ferromagnetic region returns to play a role. Indeed, Fig. \ref{fig:amp} shows that, within this range, the matched and the mismatched filters ``share" more similarities (in the sense of larger inner product). 

\begin{figure}[!t]
\begin{minipage}[b]{1.0\linewidth}
  \centering
  \centerline{\includegraphics[width=12cm,height = 10cm]{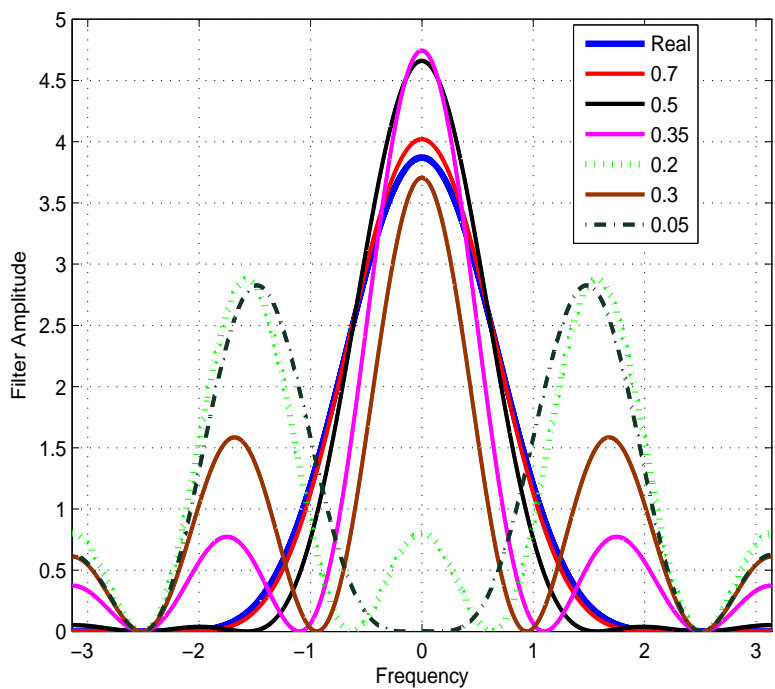}}
\end{minipage}
\caption{Example \ref{exmp3}: Amplitude of the real filter and mismatch filters for several phases.}
\label{fig:amp}
\end{figure}

\begin{figure}[!t]
\begin{minipage}[b]{1.0\linewidth}
  \centering
  \centerline{\includegraphics[width=12cm,height = 8cm]{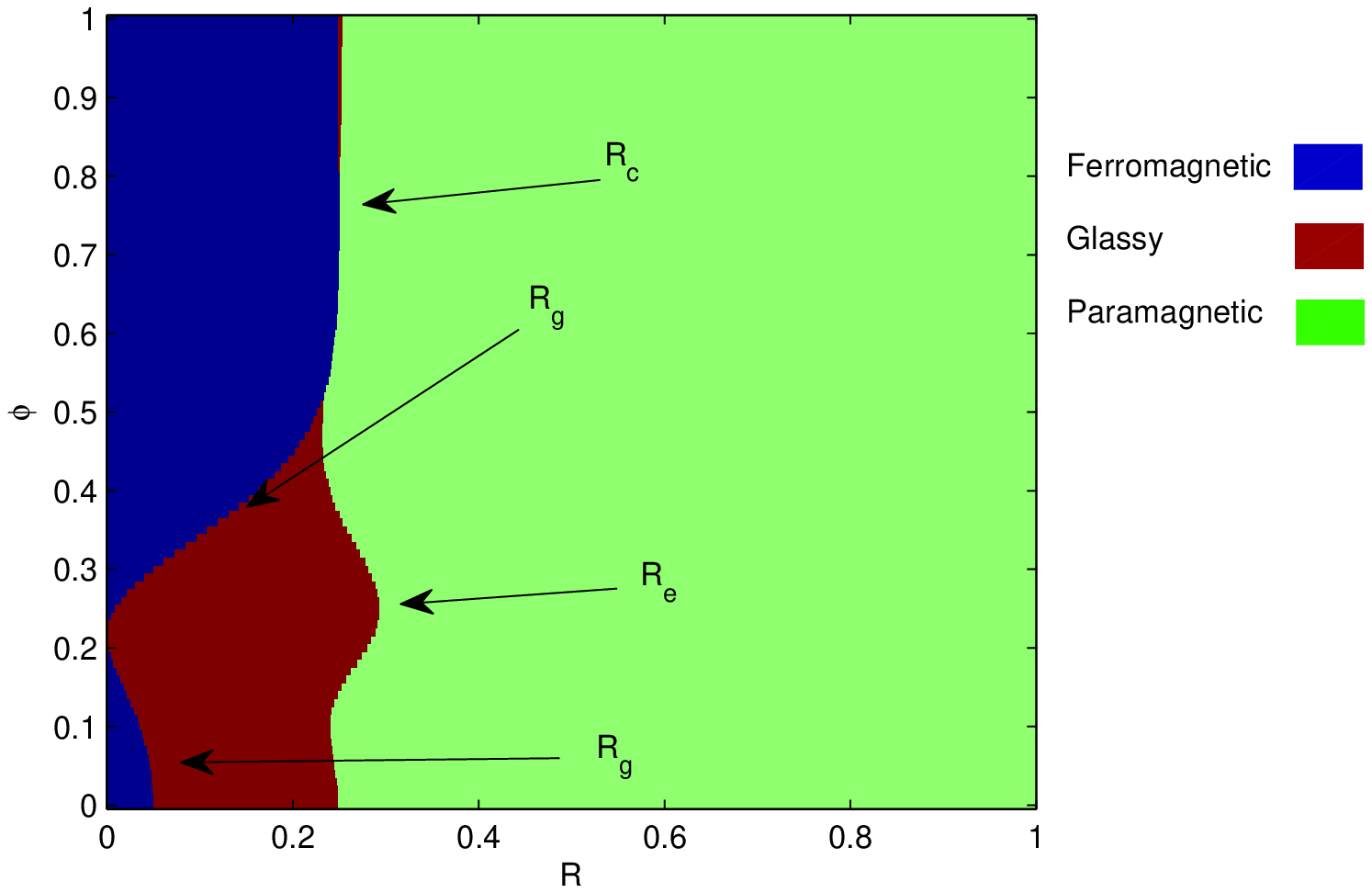}}
\end{minipage}
\caption{Example \ref{exmp3}: Phase diagram in the plane of $R$ vs. $\phi$.}
\label{fig:FIR}
\end{figure}

\begin{figure}[!t]
\begin{minipage}[b]{1.0\linewidth}
  \centering
  \centerline{\includegraphics[width=12cm,height = 8cm]{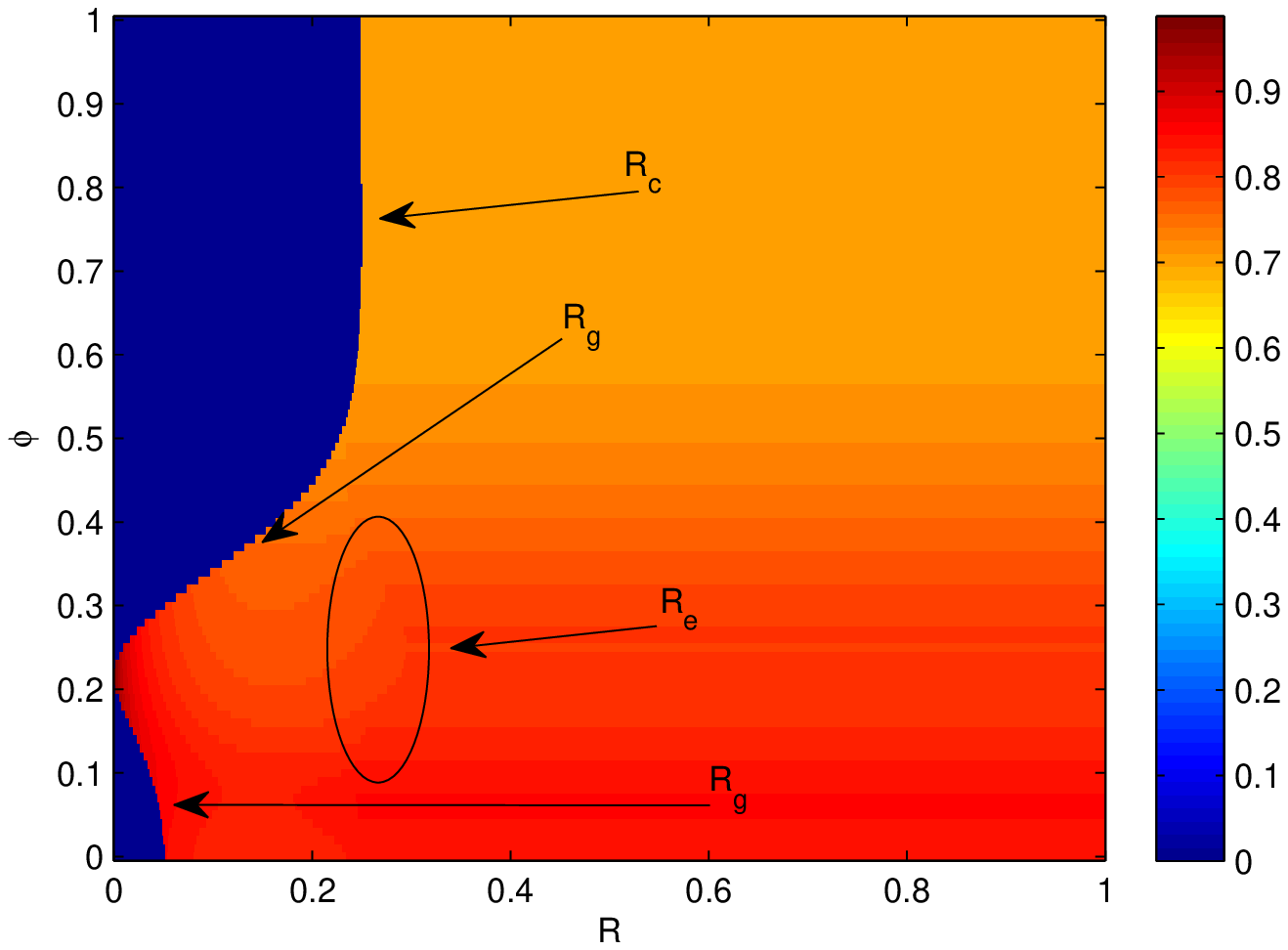}}
\end{minipage}
\caption{Example \ref{exmp3}: Mismatched MSE as a function of $R$ and $\phi$.}
\label{fig:FIR2}
\end{figure}
\end{example}

\begin{example}\label{exmp4}
Let $\Hmat\p{z}$ be given by 
\begin{align}
\Hmat\p{z} &= z-2\cos\p{0.8\pi}+z^{-1}\nonumber\\
&=z\cdot\p{1-e^{j0.8\pi}z^{-1}}\p{1-e^{-j0.8\pi}z^{-1}}
\end{align}
and let the mismatched filter be given as
\begin{align}
\Hmat'\p{z} = \Hmat\p{z}z^{-d}
\end{align}
where $d\in\mathbb{Z}$ is a mismatched delay. As before, in the numerical calculations, we chose $\beta=P_x=1$. Figures \ref{fig:9} and \ref{fig:9MSE} show, respectively, the phase diagram and the MSE as functions of $R$ and $d$. First, we see that $R_e$ is constant, approximately equal to $0.29$, which makes sense since $R_e$ is given by
\begin{align}
R_e = \frac{1}{4\pi}\int_0^{2\pi}\ln\pp{P_x\p{\gamma_0+\abs{\Hmat'\p{\omega}}^2\beta}}\mathrm{d}\omega,
\end{align}
and thus independent of the delay (note that according to \eqref{psiM} $\gamma_0$ is also independent of the delay). Next, let us take a look at $R_d$ given in \eqref{R_gDefcont}
\begin{align}
R_d=&\frac{1}{2} + \beta P_x\frac{1}{2\pi}\int_0^{2\pi}\re\p{\Hmat'^*\p{\omega}\Hmat\p{\omega}}\mathrm{d}\omega\nonumber\\
&+\frac{1}{4\pi}\int_0^{2\pi}\abs{\Hmat'\p{\omega}}^2\beta\p{\frac{\abs{\Hmat'\p{\omega}}^2\beta\p{2+\beta\abs{\Hmat\p{\omega}}^2P_x}+\gamma_0}{\p{\abs{\Hmat'\p{\omega}}^2\beta+\gamma_0}^2}-P_x}\nonumber\\
&-\frac{1}{4\pi}\int_0^{2\pi}\frac{\abs{\Hmat'\p{\omega}}^2\beta\p{3+2P_x\beta\abs{\Hmat\p{\omega}}^2}+\gamma_0}{\p{\abs{\Hmat'\p{\omega}}^2\beta+\gamma_0}}.
\end{align}
In contrast to $R_e$, $R_d$ does depend on the delay via the second term, which in the case considered takes the form $\re\p{\Hmat'^*\p{\omega}\Hmat\p{\omega}} = \abs{\Hmat\p{\omega}}^2\cos\p{\omega d}$. Actually, in the settings considered, it is easy to show that $\gamma_0 =1/P_x = 1$, thus obtaining
\begin{align}
R_d &= \frac{1}{2\pi}\int_0^{2\pi}\re\p{\Hmat'^*\p{\omega}\Hmat\p{\omega}}\mathrm{d}\omega - \frac{1}{2\pi}\int_0^{2\pi}\abs{\Hmat'\p{\omega}}^2\mathrm{d}\omega\\
& = \frac{1}{2\pi}\int_0^{2\pi}\abs{\Hmat'\p{\omega}}^2\cos\p{\omega d}\mathrm{d}\omega - \frac{1}{2\pi}\int_0^{2\pi}\abs{\Hmat'\p{\omega}}^2\mathrm{d}\omega\\
& = \frac{1}{2\pi}\int_0^{2\pi}\abs{\Hmat'\p{\omega}}^2\pp{\cos\p{\omega d}-1}\mathrm{d}\omega\leq0.
\end{align}
Therefore, we obtain that $R_d$ is non-positive, and hence for all $\phi$ (except the trivial case of $\phi=0$) there is a glassy phase. This result is consistent with Figures \ref{fig:9} and \ref{fig:9MSE}. More importantly, it can be observed that the MSE vanishes (or equivalently, the ferromagnetic phase dominates the partition function) only in case $d=0$, namely, zero delay. This is a reasonable result, as a delay of one sample (linear phase) is enough to cause a serious degradation in the MSE. Actually, for any fixed rate the error is constant, independently of the delay, as one would expect. Finally, note that the MSE is larger in the glassy region than in the paramagnetic region\footnote{Note that the MSE, in contrast to the MMSE, must not be monotonically increasing as a function of the rate.}. This is also a reasonable result: As the rate increases, and hence more codewords are possible, since the MSE estimator is actually a weighted average (w.r.t. the posterior) over the codewords, the MSE can only decrease (each codeword in the codebook contributes approximately the same estimation error). Accordingly, for small codebooks (low rates) the MSE is larger, since the averaging is performed over ``fewer" codewords.

\begin{figure}[!t]
\begin{minipage}[b]{1.0\linewidth}
  \centering
  \centerline{\includegraphics[width=12cm,height = 9cm]{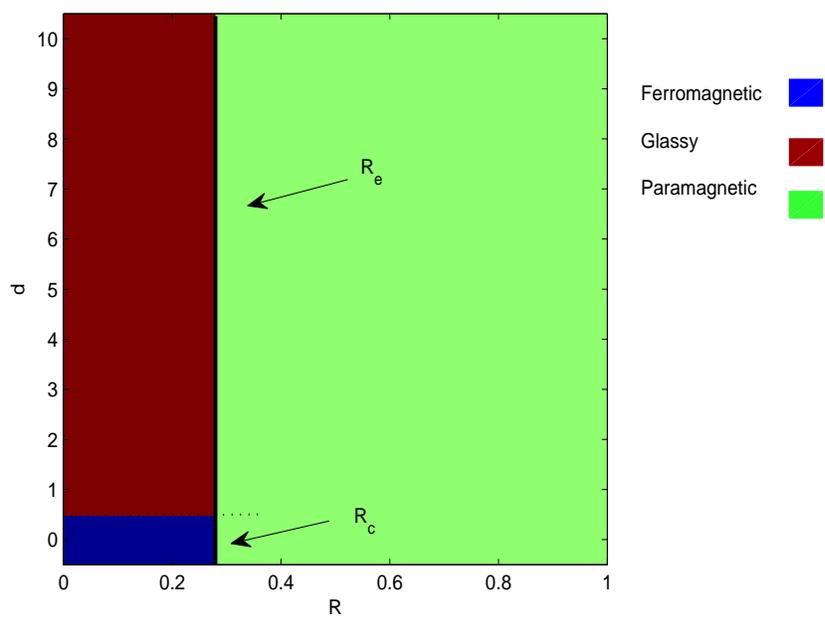}}
\end{minipage}
\caption{Example \ref{exmp4}: Phase diagram in the plane of $R$ vs. $d$.}
\label{fig:9}
\end{figure}
\begin{figure}[!t]
\begin{minipage}[b]{1.0\linewidth}
  \centering
  \centerline{\includegraphics[width=12cm,height = 9cm]{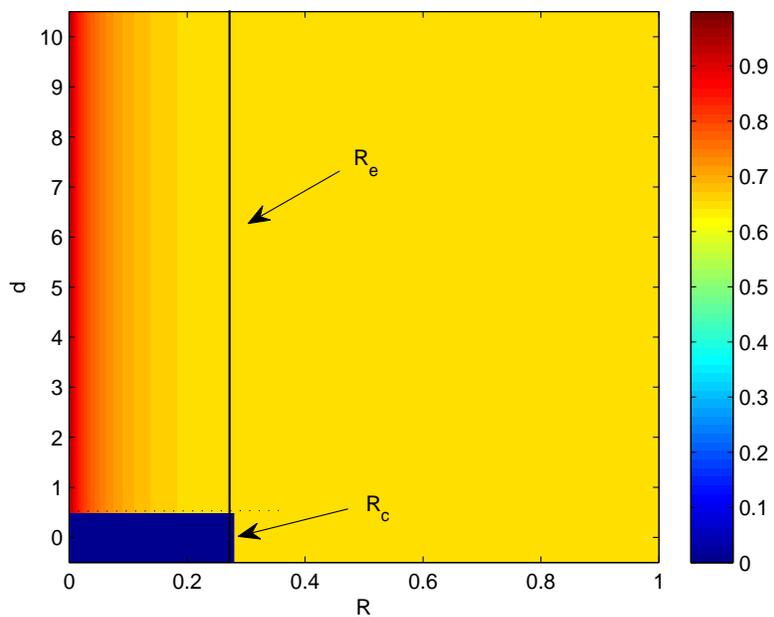}}
\end{minipage}
\caption{Example \ref{exmp4}: Mismatched MSE as a function of $R$ and $d$.}
\label{fig:9MSE}
\end{figure}
\end{example}


\section{Proof Outline and Tools}\label{sec:proofOut}
\subsection{Proof Outline}
In this section, before getting deep into the proof of Theorem \ref{th:2}, we discuss the techniques and the main steps which will be used in Section \ref{sec:proofs}. Generally speaking, the evaluation of the mismatched partition function, $Z'\p{\by,\blam}$, for a typical $\by$, essentially boils down to the evaluation of the exponential order of
\begin{align}
\text{Pr}\ppp{\frac{1}{2}\norm{\by-\bA'\bX_1}^2-\frac{\blam^T\bX_1}{\beta}\approx n\epsilon}
\label{ExpecNFIRpre}
\end{align}
for every value of $\epsilon$ in some range. In case that $\bA' = \bI$ \cite{Neri1,Neri2}, this probability can be calculated fairly easily. Indeed, in this case, the above probability is equivalent to calculating the probability that a randomly chosen vector $\bX$ on the $n$-dimensional hypersphere shell would have an empirical correlation coefficient $\rho$ (induced by the constraint $\norm{\by-\bx}^2/2-\blam^T\bx/\beta\approx n\epsilon$) with a given vector $\by' = \by+\blam/\beta$. Geometrically, this probability is actually the probability that $\bX$ falls within a cone of half angle $\arccos\p{\rho}$ around $\by'$ (for more details, see \cite{Shannon,Wyner}). However, in our case, because of the ``interactions"\footnote{In the considered settings, the posterior, is proportional to $\exp\ppp{-\beta\norm{\by-\bA'\bx}^2/2}$, and after expansion of the norm, the exponent includes an ``external-field term," proportional to $\by^H\bA'\bx$, and a ``pairwise spin-spin interaction term," proportional to $\norm{\Amat'\bx}^2$.} between different components of $\bX$, which are induced by $\bA'$, the methods in the aforementioned papers are not directly applicable. In our case, the purpose is to estimate the probability that a randomly chosen vector $\bX$ on the $n$-dimensional hypersphere shell would fall within the intersection of this hypersphere and the $n$-dimensional hyperellipsoid (which is induced by the event in \eqref{ExpecNFIRpre}). All our attempts to approach this calculation using the ``geometric" route have failed. Thus, we will use a different route. 

The main idea in our approach is, to ``eliminate" the interactions between the different components of $\bX$, by passing to the frequency domain. Since $\bA'$ is a Toeplitz matrix, according to Szeg\"{o}'s theorem \cite{Szego,Widom,Gray,Bottcher}, it is asymptotically diagonalized by the discrete Fourier transform (DFT) matrix (if $\bA'$ is a circulant matrix then the DFT matrix exactly diagonalizes it). Thus, multiplying both sides of \eqref{MatModell} by the DFT matrix, $\bF^H = \ppp{e^{-j2\pi ml/n}/\sqrt{n}}_{m,l=0}^{n-1}$, we ``asymptotically"\footnote{Rigorously, in the proof, we first assume that $\bA'$ is a circulant matrix, and thus \eqref{MatModell2Pre} is exact for any $n$. Then, when taking the limit $n\to\infty$, using Szeg\"{o}'s theorem, this assumption will be dropped. Finally, note that the assumption of the square summability of the generating sequence $\ppp{h_k}$ in the theorems presented earlier, is made in order to use Szeg\"{o}'s theorem.} have that
\begin{align}
\tilde{\bY} = \bSig\tilde{\bX} + \tilde{\bN}
\label{MatModell2Pre}
\end{align}
where $\bSig \define \diag\p{\sigma_1,\ldots,\sigma_n}$, $\tilde{\bX} \define \bF^H\bX$, $\tilde{\bY} \define \bF^H\bY$ and $\tilde{\bN} \define \bF^H\bN$. Accordingly, we evaluate \eqref{ExpecNFIRpre}, using
\begin{align}
\text{Pr}\ppp{\frac{1}{2}\norm{\tilde{\by}-\bSig\tilde{\bX}_1}^2-\frac{\tilde{\blam}^T\tilde{\bX}_1}{\beta}\approx n\epsilon}
\label{ProbTTPre}
\end{align}
where $\tilde{\blam} = \bF^T\blam$. Now, in order to evaluate \eqref{ProbTTPre}, it is desirable to estimate the volume\footnote{Recall that the volume of a set $\calA\subset\mathbb{R}^n$ is defined as $\text{Vol}\ppp{\mathcal{A}}\define\int_{\mathcal{A}}\mathrm{d}\bx$.} of the following set: For a given pair of vectors $\p{\tilde{\bx},\tilde{\by}}$ and $\delta>0$, we define the conditional $\delta$-type of $\tilde{\bx}$ given $\tilde{\by}$ as
\begin{align}
\calT_\delta\p{\tilde{\bx}\mid\tilde{\by}} \define \ppp{\tilde{\bx}\in\mathbb{R}^n:\;\abs{\norm{\tilde{\bx}}^2-nP_x}\leq\delta,\;\abs{\frac{\norm{\tilde{\by}-\bSig\tilde{\bx}}^2}{2}-\frac{\tilde{\blam}^T\tilde{\bx}}{\beta}-n\epsilon}\leq\delta}.
\label{TypeSetPre}
\end{align}
This set is regarded as a conditional type of (wrong) codewords $\tilde{\bx}$ given $\tilde{\by}$ as it contains all vectors which, within $\delta$, have the same energy related to the partition function \eqref{PartFuncmod}. After calculating the volume of \eqref{TypeSetPre}, the probability in \eqref{ProbTTPre} can then be easily estimated. However, as was previously mentioned, calculating the volume of such a set is a tedious task when approaching it directly. We will use instead the following relaxation. We start with partitioning the components of $\tilde{\bx}$ into $k$ \emph{bins}, each of dimension $n_b$, such that $k = n/n_b$, and we approximate the eigenvalues, which are the diagonal elements of $\bSig$, to be piecewise constant over these bins. This partition literally means that we transform the original model in \eqref{MatModell2Pre} into $k$ subchannels, each having the form
\begin{align}
Y_{i,r} &= \sigma_rX_i + N_i,\ \ i = \p{r-1}n_b+1,\ldots,r\cdot n_b,
\end{align}
for $r = 1,\ldots,k$. With this partitioning in mind, at the final stage of the analysis (after taking the limit $n\to\infty$), we take the limit $k\to\infty$. This partitioning will enable to calculate the desired volume. Then, using large deviations considerations, the mismatched partition function will be obtained. Finally, in order to derive the MSE, we will use the tools of \cite{Neri1}, which are briefly presented in the following subsection.

\subsection{Optimum Estimation Relations - Background and Summary}\label{sec:back}
\subsubsection{Matched Case}
Let $\bX=\p{X_1,\ldots,X_n}$ and $\bY=\p{Y_1,\ldots,Y_m}$ be two random vectors, jointly distributed according to a given probability function $P\p{\bx,\by}$. The conditional mean estimator of $\bX$ based on $\bY$, i.e., $\hat{\bX} = \bE\ppp{\bX\mid\bY}$ is well known to minimize the MSE $\bE\p{X_i-\hat{X}_i}^2$ for all $i=1,\ldots,n$. Accordingly, the MMSE in estimating $X_i$ equals to $\bE\ppp{\p{X_i-\bE\ppp{X_i\mid\bY}}^2}$, i.e., the expected conditional variance of $X_i$ given $\bY$. More generally, the MMSE error covariance matrix is an $n\times n$ matrix whose $\p{i,j}$-th element is given by $\bE\ppp{\p{X_i-\bE\ppp{X_i\mid\bY}}\p{X_j-\bE\ppp{X_j\mid\bY}}}$. This matrix can be represented as the expectation (w.r.t. $\bY$) of the conditional covariance matrix of $\bX$ given $\bY$, henceforth denoted by $\text{cov}\p{\bX\mid\bY}$. In particular, using the orthogonality principle, the MMSE error covariance matrix is given by
\begin{align}
\bE\ppp{\text{cov}\p{\bX\mid\bY}} = \bE\ppp{\bX\bX^T} - \bE\ppp{\bE\ppp{\bX\mid\bY}\bE\ppp{\bX^T\mid\bY}}.
\end{align} 
Based on Definition \ref{def:1}, the following relations readily follow
\begin{align}
\bE\ppp{\bX\mid\bY=\by}&=\nabla_0g\p{\blam} \ln Z\p{\by,\blam}\\
\bE\ppp{\text{cov}\p{\bX\mid\bY}} &= \bE\ppp{\nabla_0^2\ln Z\p{\bY,\blam}}
\end{align}
where for a generic function $g$, we use $\nabla_0g\p{\blam}$ and $\nabla_0^2g\p{\blam}$ to designate $\left.\nabla_{\blamt}g\p{\blam}\right|_{\blamt=0}$ and $\left.\nabla_{\blamt}^2g\p{\blam}\right|_{\blamt=0}$, respectively, and $\nabla_{\blamt}$ and $\nabla^2_{\blamt}$ denote the gradient and Hessian operators w.r.t. $\blam$, respectively. 
Finally, it is easy to verify that the following relation holds
\begin{align}
\bE\ppp{\text{cov}\p{\bX\mid\bY}} = \bE\ppp{\bX\bX^T}-\bE\ppp{\pp{\nabla_0\ln Z\p{\bY,\blam}}\pp{\nabla_0\ln Z\p{\bY,\blam}}^T},
\end{align}
and upon taking the trace of the above equation one obtains
\begin{align}
\text{mmse}\p{\bX\mid\bY} &\define \sum_{i=1}^n \bE\ppp{\p{X_i-\bE\ppp{X_i\mid\bY}}^2}\nonumber\\
&= \sum_{i=1}^n\pp{\bE\ppp{X_i^2}-\bE\ppp{\left.\pp{\frac{\partial\ln Z\p{\bY,\blam}}{\partial\lambda_i}}^2\right|_{\blamt=0}}}.
\end{align}
Further relations between information measures and estimation quantities can be found in \cite{Neri1,Neri2}. 

\subsubsection{Mismatched Case}
Consider a mismatched estimator which is the conditional mean of $\bX$ given $\bY$, based on an incorrect joint distribution $P'\p{\bx,\by}$, whereas the true joint distribution continues to be $P\p{\bx,\by}$. Then, the following relation holds
\begin{align}
\bE\ppp{\text{cov}'\p{\bX\mid\bY}} &\define \bE\ppp{\p{\bX-\bE'\ppp{\bX\mid\bY}}\p{\bX-\bE'\ppp{\bX\mid\bY}}^T}\nonumber\\
&=\bE\ppp{\bX\bX^T} - \bE_P\ppp{\bE\ppp{\bX\mid\bY}\bE'\ppp{\bX^T\mid\bY}}\nonumber\\
&\ \ \ -\bE\ppp{\bE'\ppp{\bX\mid\bY}\bE\ppp{\bX^T\mid\bY}}\nonumber\\
&\ \ \ +\bE\ppp{\bE'\ppp{\bX\mid\bY}\bE'\ppp{\bX^T\mid\bY}},
\label{generalMiscov}
\end{align}
where $\text{cov}'\p{\bX\mid\bY} \define \p{\bX-\bE'\ppp{\bX\mid\bY}}\p{\bX-\bE'\ppp{\bX\mid\bY}}^T$. Upon taking the trace of \eqref{generalMiscov}, one obtains
\begin{align}
\text{mse}\p{\bX\mid\bY} &\define \sum_{i=1}^n \bE\ppp{\p{X_i-\bE'\ppp{X_i\mid\bY}}^2}\nonumber\\
&= \sum_{i=1}^n\left[\vphantom{\bE\ppp{\left.\pp{\frac{\partial\ln Z'\p{\bY,\blam}}{\partial\lambda_i}}^2\right|_{\blamt=0}}}\bE\ppp{X_i^2}-2\bE\ppp{\left.\frac{\partial\ln Z\p{\bY,\blam}}{\partial\lambda_i}\right|_{\blamt=0}\cdot\left.\frac{\partial\ln Z'\p{\bY,\blam}}{\partial\lambda_i}\right|_{\blamt=0}}\right.\nonumber\\
&\left.\ \ \ \ \ \ \ \ \ \ +\bE\ppp{\left.\pp{\frac{\partial\ln Z'\p{\bY,\blam}}{\partial\lambda_i}}^2\right|_{\blamt=0}}\right].
\end{align}

\section{Proof of Theorem \ref{th:2}}\label{sec:proofs}
For a given $\by$, the mismatched partition function is given by\footnote{Note that there should be a normalization factor of $\p{2\pi/\beta}^{-n/2}$ in \eqref{partfirst}. Nonetheless, since this constant is independent of $\blam$, it has no effect on the MSE (which is obtained by the gradient of $\ln Z'\p{\by,\blam}$ w.r.t. $\blam$). Hence, for simplicity of notation, it is omitted.}
\begin{align}
\label{partfirst}
Z'\p{\by,\blam} &= \sum_{\bx\in\calC}e^{-nR}\exp\pp{-\beta\norm{\by-\bA'\bx}^2/2+\blam^T\bx}\\
& = e^{-nR}\exp\pp{-\beta\norm{\by-\bA'\bx_0}^2/2+\blam^T\bx_0}\\
&\ \ \ \ \  +\sum_{\bx\in\calC\setminus\ppp{\bx_0}}e^{-nR}\exp\pp{-\beta\norm{\by-\bA'\bx}^2/2+\blam^T\bx}\\
&\define Z'_{c}\p{\by,\blam} + Z'_{e}\p{\by,\blam}
\end{align}
where without loss of generality, the transmitted codeword is assumed to be $\bx_0$, and $Z'_{c}\p{\by,\blam}$ and $Z'_{e}\p{\by,\blam}$ are the partial partition functions induced by the correct codeword and the wrong codewords, respectively. By the law of large numbers (LLN), $\norm{\by-\bA'\bx_0}^2\approx\norm{\p{\bA-\bA'}\bx_0}^2+n/\beta$, and therefore, with high probability
\begin{align}
Z'_{c}\p{\by,\blam}&\exe  e^{-nR}\exp\ppp{-\frac{\beta}{2}\pp{\norm{\p{\bA-\bA'}\bx_0}^2+\frac{n}{\beta}}+\blam^T\bx_0}\\
& = \exp\ppp{-n\p{R+\frac{1}{2}+\frac{\beta\norm{\p{\bA-\bA'}\bx_0}^2}{2n}}+\blam^T\bx_0}.
\label{ZQc}
\end{align}
More precisely, for any $\epsilon>0$,
\begin{align}
&\exp\ppp{-n\p{R+\frac{1}{2}+\frac{\beta\norm{\p{\bA-\bA'}\bx_0}^2}{2n}+\epsilon}+\blam^T\bx_0}\leq Z'_{c}\p{\by,\blam}\nonumber\\
&\ \ \ \ \leq\exp\ppp{-n\p{R+\frac{1}{2}+\frac{\beta\norm{\p{\bA-\bA'}\bx_0}^2}{2n}-\epsilon}+\blam^T\bx_0}
\end{align}
with probability tending to one as $n\to\infty$. As for $Z'_{e}\p{\by,\blam}$, we have
\begin{align}
Z'_{e}\p{\by,\blam} = e^{-nR}\int_{\mathbb{R}}\mathcal{N}\p{\epsilon}e^{-n\beta\epsilon}\mathrm{d}\epsilon
\label{errorPart}
\end{align}
where
\begin{align}
\mathcal{N}\p{\epsilon} \define \sum_{i=1}^{M-1}\Ind{\ppp{\bx_i:\;\frac{\norm{\by-\bA'\bx_i}^2}{2}-\frac{\blam^T\bx_i}{\beta}\approx n\epsilon}},
\end{align}
to wit, $\mathcal{N}\p{\epsilon}$ is the number of codewords $\ppp{\bx_i}$ in $\calC\setminus\ppp{\bx_0}$ for which $\norm{\by-\bA'\bx_i}^2/2-\blam^T\bx_i/\beta\approx n\epsilon$, namely, between $n\epsilon$ and $n\p{\epsilon+\mathrm{d}\epsilon}$. We proceed in two steps: First, the typical exponential order of $\mathcal{N}\p{\epsilon}$ is computed, and then \eqref{errorPart} is calculated.
\newline\itshape Step 1\normalfont: Given $\by$, $\mathcal{N}\p{\epsilon}$ is a sum of $\p{M-1}$ i.i.d. Bernoulli random variables and therefore, its expected value is given by
\begin{align}
\bE\ppp{\mathcal{N}\p{\epsilon}} &=\sum_{i=1}^{M-1}\text{Pr}\ppp{\frac{\norm{\by-\bA'\bX_i}^2}{2}-\frac{\blam^T\bX_i}{\beta}\approx n\epsilon}\\
&= \p{e^{nR}-1}\cdot\text{Pr}\ppp{\frac{\norm{\by-\bA'\bX_1}^2}{2}-\frac{\blam^T\bX_1}{\beta}\approx n\epsilon}.
\label{ExpecNFIR1}
\end{align}
Assuming that $\bA'$ is a circulant matrix\footnote{Recall that this assumption is only an intermediate step in the analysis, and will be dropped later on. Alternatively, instead of this assumption, one could use the spectral decomposition theorem, to find an orthonormal basis which diagonalizes the matrix $\bA'$, and project \eqref{MatModell} on this basis, to obtain the form of \eqref{MatModell2}.}, it is known that the discrete Fourier transform (DFT) matrix diagonalizes it \cite{Szego,Widom,Gray,Bottcher}, and thus multiplying both sides of equation \eqref{MatModell} by the DFT matrix, $\bF^H$, one obtains
\begin{align}
\tilde{\bY} = \bSig'\tilde{\bX} + \tilde{\bN}
\label{MatModell2}
\end{align}
where $\bSig' \define \diag\p{\sigma'_1,\ldots,\sigma'_n}$, $\tilde{\bX} \define \bF^H\bX$, $\tilde{\bY} \define \bF^H\bY$ and $\tilde{\bN} \define \bF^H\bN$. Since a unitary operator is applied on $\bX$, then $\tilde{\bX}$ is still uniformly drawn on the $n$-hyperesphere with radius $\sqrt{nP_x}$ (as in the original setting). Similarly, $\tilde{\bN}$ has the same statistics as before, namely, its components are i.i.d. complex Gaussian random variables with zero mean and variance $1/\beta$. For simplicity of notation, in the following, the ``tilde" sign over the various variables will be omitted, keeping the original notation. Therefore, instead of evaluating \eqref{ExpecNFIR1}, the exponential order of
\begin{align}
\text{Pr}\ppp{\frac{\norm{\by-\bSig'\bX_1}^2}{2}-\frac{\blam^T\bX_1}{\beta}\approx n\epsilon},
\label{ProbTT}
\end{align}
will be evaluated, where $\bF^T\blam\mapsto\blam$\footnote{Note that $\bF^T\blam$ may be a complex quantity (in contrast to $\blam$). This fact will be taken into account later on.}. For a given pair of vectors $\p{\bx,\by}$ and $\delta>0$, define the conditional $\delta$-type of $\bx$ given $\by$ as
\begin{align}
\calT_\delta\p{\bx\mid\by} \define \ppp{\bx\in\mathbb{R}^n:\;\abs{\norm{\bx}^2-nP_x}\leq\delta,\;\abs{\frac{\norm{\by-\bSig'\bx}^2}{2}-\frac{\blam^T\bx}{\beta}-n\epsilon}\leq\delta}.
\label{TypeSet1}
\end{align}
The following lemma is proved in Appendix \ref{app:2}.
\begin{lemma}\label{lem:3}
Let $k$ and $n_b$ be natural numbers such that $k=n/n_b$\footnote{Without loss of generality, it is assumed that $n_b$ (\emph{bin} length) is a divisor of $n$, and that the $k$ various bins have equal sizes.}. Define the sets $\calG_{1,\delta} \define \ppp{\delta\cdot i:\;i=0,1,\ldots,\left\lceil kP_x/\delta\right\rceil}$ and $\calG_{2,\delta} \define \ppp{\delta\cdot i:\;i=-\left\lceil k/\delta\right\rceil,\ldots,-1,0,1,\ldots,\left\lceil k/\delta\right\rceil}$. Also, let
\begin{align}
\hat{\calT}_\delta\p{\bx\mid\by} \define \bigcup\limits_{\substack{\boldsymbol{\mathcal{P}}^{\delta}\cap\boldsymbol{\mathcal{R}}^{\delta}_{\boldsymbol{\mathcal{P}}}}}\ \bigtimes_{m=1}^k\mathscr{B}_{m}^{\delta}\p{P_m,\rho_m}
\end{align}
where $\bigtimes$ designates a Cartesian product, and
\begin{align}
\mathscr{B}_{m}^\delta\p{P_m,\rho_m} &\define \left\{\vphantom{\abs{\re\ppp{\sum_{i\in\mathcal{I}_m}\sigma'_i\bar{y}_i{x}_i}-n_{b}\rho_m{\sqrt{\tilde{P}_{y,m}\tilde{P}_{\sigma,m}}}}}\bx \in\mathbb{R}^{n_b}:\;\abs{\norm{\bx_{\p{m-1}n_b+1}^{mn_b}}^2-n_{b}P_m}\leq\delta,\nonumber\right.\\
&\ \ \ \ \ \ \ \ \ \ \ \ \ \ \; \ \ \ \left. \;\abs{\re\ppp{\sum_{i\in\mathcal{I}_m}\sigma'_i\bar{y}_i{x}_i}-n_{b}\rho_m{\sqrt{\tilde{P}_{y,m}\tilde{P}_{\sigma,m}}}}\leq\delta\right\}
\end{align}
where $\mathcal{I}_m \define \pp{\p{m-1}n_b+1,mn_b}$, $\bar{y}_i\define y_i^*+\frac{\lambda_i}{\beta\sigma'_i}$, $\tilde{P}_{y,m} \define \frac{1}{n_b}\sum_{i\in\mathcal{I}_m}\abs{\bar{y}_i}^2$, $\tilde{P}_{\sigma,m}\define \frac{1}{n_b}\sum_{i\in\mathcal{I}_m}\abs{\sigma_i'x_i}^2$, and\footnote{The purpose of the subscript symbol in $\boldsymbol{\mathcal{R}}^\delta_{\boldsymbol{\mathcal{P}}}$ is to emphasize the dependence of it on $\boldsymbol{\mathcal{P}}$. More precisely, these sets should be understood as joint-power-correlation allocations, which are ``living" in the intersection $\boldsymbol{\mathcal{P}}^\delta\cap\boldsymbol{\mathcal{R}}^\delta_{\boldsymbol{\mathcal{P}}}$. Accordingly, $P_m$ and $\rho_m$ are the power and correlation constraints within the $m$th bin, respectively.}
\begin{align}
\label{powerSetg}
\boldsymbol{\mathcal{P}}^{\delta}&\define\ppp{\bP\in\calG_{1,\delta}^k:\;\abs{\frac{1}{k}\sum_{i=1}^kP_i-P_x} \leq\delta}\\
\boldsymbol{\mathcal{R}}^{\delta}_{\boldsymbol{\mathcal{P}}}&\define\ppp{\boldsymbol{\rho}\in\calG_{2,\delta}^k:\;\abs{\frac{1}{k}\sum_{i=1}^k\rho_i\sqrt{\tilde{P}_{y,i}\tilde{P}_{\sigma,i}}- \tilde{\rho}}\leq\delta}
\end{align}
where $\calG_{1,\delta}^k$ and $\calG_{2,\delta}^k$ are the $k$th Cartesian power of $\calG_{1,\delta}$ and $\calG_{2,\delta}$, respectively, and
\begin{align}
\tilde{\rho} \define \frac{\frac{1}{n}\sum_{i=1}^n\abs{\sigma'_ix_i}^2 + P_y-2\epsilon}{2}
\end{align}
where $P_y \define \frac{1}{n}\sum_{i=1}^n\abs{y_i}^2$. Then,
\begin{align}
\hat{\calT}_{\delta/k}\p{\bx\mid\by}\subseteq\calT_\delta\p{\bx\mid\by}\subseteq\hat{\calT}_\delta\p{\bx\mid\by}.
\end{align}
\end{lemma}
Next, the eigenvalues, $\ppp{\sigma'_i}_i$, are approximated to be piecewise constant over the various bins. At the final stage of the analysis (after taking the limit $n\to\infty$), we will take the limit $k\to\infty$ so that this approximation becomes superfluous. Accordingly, under this approximation, $\tilde{P}_{\sigma,m} = \abs{\sigma_m'}^2P_m$, and (with abuse of notation)
\begin{align}
\hat{\calT}^k_\delta\p{\bx\mid\by} = \bigcup\limits_{\substack{\boldsymbol{\mathcal{P}}^\delta\cap\boldsymbol{\mathcal{R}}^\delta_{\boldsymbol{\mathcal{P}}}}}\ \bigtimes_{m=1}^k\mathscr{B}_{m}^{\delta}\p{P_m,\rho_m}
\end{align}
where now
\begin{align}
\mathscr{B}_{m}^\delta\p{P_m,\rho_m} &\define \left\{\vphantom{\abs{\re\ppp{\sigma'_m\sum_{i\in\mathcal{I}_m}\bar{y}_i{x}_i}-n_{b}\rho_m{\sqrt{\tilde{P}_{y,m}\abs{\sigma_m'}^2P_m}}}}\bx \in\mathbb{R}^{n_b}:\;\abs{\norm{\bx_{\p{m-1}n_b+1}^{mn_b}}^2-n_{b}P_m}\leq\delta, \right.\nonumber\\
&\left.\ \ \ \ \ \ \ \ \ \ \ \ \ \ \ \ \ \ \;\abs{\re\ppp{\sigma'_m\sum_{i\in\mathcal{I}_m}\bar{y}_i{x}_i}-n_{b}\rho_m{\sqrt{\tilde{P}_{y,m}\abs{\sigma_m'}^2P_m}}}\leq\delta\right\}
\label{BsetB}
\end{align}
and
\begin{align}
\boldsymbol{\mathcal{R}}^\delta_{\boldsymbol{\mathcal{P}}}\define\ppp{\boldsymbol{\rho}\in\calG_{2,\delta}^k:\;\abs{\frac{1}{k}\sum_{i=1}^k\abs{\sigma'_i}\rho_i\sqrt{\tilde{P}_{y,i}P_i}-\tilde{\rho}}\leq\delta},
\end{align}
where
\begin{align}
\tilde{\rho} \define \frac{\frac{1}{k}\sum_{i=1}^k\abs{\sigma'_i}^2P_i + P_y-2\epsilon}{2}.
\end{align}

In the following, the volume of $\mathcal{T}_\delta\p{\bx\mid\by}$ is evaluated. On the one hand, using Lemma \ref{lem:3}, one obtains that
\begin{align}
\text{Vol}\ppp{\mathcal{T}_\delta\p{\bx\mid\by}}&\leq\text{Vol}\ppp{\hat{\mathcal{T}}_\delta^k\p{\bx\mid\by}}\\
&\leq \sum\limits_{\substack{\boldsymbol{\mathcal{P}}^{\delta}\cap\boldsymbol{\mathcal{R}}^{\delta}_{\boldsymbol{\mathcal{P}}}}}\text{Vol}\ppp{\bigtimes_{m=1}^k\mathscr{B}_{m}^{\delta}\p{P_m,\rho_m}}\\
&\leq N_{\delta,k}\cdot\operatorname*{max}_{\substack{\boldsymbol{\mathcal{P}}^{\delta}\cap\boldsymbol{\mathcal{R}}^{\delta}_{\boldsymbol{\mathcal{P}}}}}\text{Vol}\ppp{\bigtimes_{m=1}^k\mathscr{B}_{m}^{\delta}\p{P_m,\rho_m}}
\label{volineq1}
\end{align}
where the second inequality follows for the union bound, and $N_{k,\delta}$ is a constant depending on $k$ and $\delta$ (but not on $n$). This constant can be roughly bounded by
\begin{align}
N_{\delta,k}\leq\abs{\boldsymbol{\mathcal{P}}^{\delta}}\abs{\boldsymbol{\mathcal{R}}^\delta_{\boldsymbol{\mathcal{P}}}}\leq\p{\frac{kP_x+2\delta}{\delta}}^k\p{2\cdot\frac{k+\delta}{\delta}}^k.
\end{align}
On the other hand,
\begin{align}
\text{Vol}\ppp{\mathcal{T}_\delta\p{\bx\mid\by}}&\geq\text{Vol}\ppp{\hat{\mathcal{T}}_{\delta/k}^k\p{\bx\mid\by}}\\
&\geq \operatorname*{max}_{\substack{\boldsymbol{\mathcal{P}}^{\delta/k}\cap\boldsymbol{\mathcal{R}}^{\delta/k}_{\boldsymbol{\mathcal{P}}}}}\text{Vol}\ppp{\bigtimes_{m=1}^k\mathscr{B}_{m}^{\delta/k}\p{P_m,\rho_m}}.
\label{volineq2}
\end{align}
The following lemma is proved in Appendix \ref{app:1}.
\begin{lemma}\label{lem:1}
For every $m=1,\ldots,k$ and $\nu>0$,
\begin{align}
\p{1-\nu}\exp\ppp{\frac{n_b}{2}\ln\p{\pi e\vartheta_{o,u}^{2}}}\leq\text{Vol}\ppp{\mathscr{B}_{m}^\delta\p{P_{m},\rho_{m}}}\leq\exp\ppp{\frac{n_b}{2}\ln\p{\pi e\vartheta_o^{2}}}.
\end{align}
where
\begin{align}
\vartheta_{\delta,+}^{2}=P_m+\delta-P_m\p{\rho_m-\delta}^2
\end{align}
and
\begin{align}
\vartheta_{\delta,-}^{2} = P_m-\delta-P_m\p{\rho_m+\delta}^2.
\end{align}
In particular,
\begin{align}
\lim_{\delta\to0}\lim_{n_b\to\infty}\frac{1}{n_b}\ln\text{Vol}\ppp{\mathscr{B}_{m}^\delta\p{P_{m},\rho_{m}}} = \frac{1}{2}\ln\p{\pi eP_{m}\p{1-\rho^2_{m}}}.
\end{align}
\end{lemma}

Now,
\begin{align}
\text{Vol}\ppp{\bigtimes_{m=1}^k\mathscr{B}_{m}^\delta\p{P_m,\rho_m}} = \prod_{m=1}^k\text{Vol}\ppp{\mathscr{B}_{m}^\delta\p{P_m,\rho_m}}.
\label{VolProd}
\end{align}
Whence, using Lemma \ref{lem:1}, \eqref{VolProd}, \eqref{volineq1} and \eqref{volineq2}, one obtains that
\begin{align}
&\text{Vol}\ppp{\mathcal{T}_\delta\p{\bx\mid\by}}\geq\operatorname*{max}_{\substack{\boldsymbol{\mathcal{P}}^{\delta/k}\cap\boldsymbol{\mathcal{R}}^{\delta/k}_{\boldsymbol{\mathcal{P}}}}}\p{\pi e}^{n/2}\exp\ppp{\frac{n_{b}}{2}\sum_{m=1}^k\ln\p{\vartheta_{\delta/k,-}^{2}}}
\end{align}
and
\begin{align}
&\text{Vol}\ppp{\mathcal{T}_\delta\p{\bx\mid\by}}\leq N_{\delta,k}\cdot\operatorname*{max}_{\substack{\boldsymbol{\mathcal{P}}^{\delta}\cap\boldsymbol{\mathcal{R}}^{\delta}_{\boldsymbol{\mathcal{P}}}}}\p{\pi e}^{n/2}\exp\ppp{\frac{n_{b}}{2}\sum_{m=1}^k\ln\p{\vartheta_{\delta,+}^{2}}}.
\end{align}
Thus,
\begin{align}
\lim_{\delta\to0}\lim_{n\to\infty}\frac{1}{n}\ln\text{Vol}\ppp{\mathcal{T}_\delta\p{\bx\mid\by}}=\frac{1}{2}\ln\p{\pi e}+\operatorname*{max}_{\boldsymbol{\mathcal{P}}\cap\boldsymbol{\mathcal{R}}_{\boldsymbol{\mathcal{P}}}}\ppp{\frac{h}{2}\sum_{m=1}^k\ln\p{P_{m}\p{1-\rho^2_{m}}}}
\label{VolProd2}
\end{align}
where
\begin{align}
\label{powerSet1}
\boldsymbol{\mathcal{P}}&\define\ppp{\bP\in\mathbb{R}^k:\;\frac{1}{k}\sum_{i=1}^kP_i=P_x}\\
\boldsymbol{\mathcal{R}}_{\boldsymbol{\mathcal{P}}}&\define\ppp{\boldsymbol{\rho}\in\mathbb{R}^k:\;\frac{1}{k}\sum_{i=1}^k\abs{\sigma'_i}\rho_i\sqrt{\tilde{P}_{y,i}P_i}= \tilde{\rho}}.
\label{corrSet1}
\end{align}
Finally, the probability in \eqref{ProbTT}, is given by
\begin{align}
\lim_{n\to\infty}\frac{1}{n}\ln\text{Pr}\ppp{\frac{\norm{\by-\bSig'\bX_1}^2}{2}-\frac{\blam^T\bX_1}{\beta}\approx n\epsilon}=\lim_{h\to0}\lim_{\delta\to0}\lim_{n\to\infty}\frac{1}{n}\ln\p{\frac{\text{Vol}\ppp{\mathcal{T}_\delta^k\p{\bx\mid\by}}}{{\text{Vol}}\ppp{\mathcal{T}_{x,\delta}^n}}},
\label{propbability1}
\end{align}
in which $\mathcal{T}_{x,\delta}^n$ is the set of $n$-dimensional $\bx$-complex vectors with norm $\sqrt{nP_x}$. 
\begin{lemma}\label{lem:2}
The volume of $\mathcal{T}_{x,\delta}^n$ is given by
\begin{align}
\lim_{\delta\to0}\lim_{n\to\infty}\frac{1}{n}\ln{\text{Vol}}\ppp{\mathcal{T}_{x,\delta}^n}=\frac{1}{2}\ln\p{\pi eP_x}.
\end{align}
\end{lemma}
\begin{proof}
Readily follows by using almost the same proof of Lemma \ref{lem:1} (see Appendix \ref{app:1}).
\end{proof}
Thus, applying Lemma \ref{lem:2} on \eqref{propbability1}, one obtains\footnote{Note that at this stage, using once again the dominated convergence theorem (DCT) \cite{rudin} and Szeg\"{o}'s theorem \cite{Szego,Widom,Gray,Bottcher}, we can refine the bin sizes by taking the limit $h\define n_b/n = 1/k\to0$, and then to solve a variational problem. However, it turns out that it is better to refine the bin sizes only at the last stage of the analysis.}
\begin{align}
\text{Pr}\ppp{\frac{\norm{\by-\bSig'\bX_1}^2}{2}-\frac{\blam^T\bX_1}{\beta}\approx n\epsilon}\exe\exp\ppp{n\tilde{\Gamma}\p{\epsilon}}
\label{probGen}
\end{align}
with probability tending to one as $n\to\infty$, and
\begin{align}  
\tilde{\Gamma}\p{\epsilon} \define \lim_{h\to0}\operatorname*{max}_{\boldsymbol{\mathcal{P}},\boldsymbol{\mathcal{R}}_{\boldsymbol{\mathcal{P}}}}\ppp{\frac{h}{2}\sum_{m=1}^k\ln\p{\frac{P_{m}}{P_x}\p{1-\rho^2_{m}}}}.
\label{Gammadaffreq2}
\end{align}
Therefore, using \eqref{ExpecNFIR1}
\begin{align}
\bE\ppp{\mathcal{N}\p{\epsilon}} \exe \exp\ppp{n\p{R+\tilde{\Gamma}\p{\epsilon}}}.
\end{align}
To finish step 1, the following lemma is proposed and proved in Appendix \ref{app:4}\footnote{Lemma \ref{lem:4} simply states that, if we chose $\epsilon$ such that, $R+\tilde{\Gamma}\p{\epsilon}>0$, then the energy level $\epsilon$ will be ``typically" populated with an exponential number of codewords, concentrated very strongly around its mean $\bE\ppp{\mathcal{N}\p{\epsilon}}$. Otherwise (which means that $\bE\ppp{\mathcal{N}\p{\epsilon}}$ is exponentially small), the energy level $\epsilon$ will not be populated by any codewords ``typically".}.
\begin{lemma}\label{lem:4}
Let 
\begin{align}
\mathscr{E} \define \ppp{\epsilon\in\mathbb{R}:\;R+\tilde{\Gamma}\p{\epsilon}>0}.
\label{EEE}
\end{align}
Then, 
\begin{align}
\lim_{n\to\infty}\frac{1}{n}\ln\calN\p{\epsilon} = 
\begin{cases}
R+\tilde{\Gamma}\p{\epsilon},\ \ &\epsilon\in\mathscr{E}\\
-\infty,\ &\text{else}
\end{cases}
\label{statement}
\end{align}
with probability (w.p.) 1. 
\end{lemma}
\itshape Step 2\normalfont: Using Lemma \ref{lem:4}, \eqref{errorPart}, and Varadhan's theorem \cite{Dembo}, one obtains that \cite[Ch. 2]{Galavotti,Mezard},
\begin{align}
Z'_{e}\p{\bY,\blam} &\exe e^{-nR}\;\operatorname*{max}_{\epsilon\in\mathscr{E}}\;\exp\ppp{n\p{R+\tilde{\Gamma}\p{\epsilon}-\beta\epsilon}}\\
&=\exp\ppp{n\pp{\operatorname*{max}_{\epsilon\in\mathscr{E}}\ppp{\tilde{\Gamma}\p{\epsilon}-\beta\epsilon}}},
\label{maxEprob1}
\end{align}
namely, w.p. 1,
\begin{align}
\lim_{n\to\infty}\frac{\ln Z'_{e}\p{\bY,\blam}}{n} =\operatorname*{max}_{\epsilon\in\mathscr{E}}\ppp{\tilde{\Gamma}\p{\epsilon}-\beta\epsilon}.
\label{highprob}
\end{align}
Let $\Gamma\p{\epsilon}$ be defined as in \eqref{Gammadaffreq2}, but without the limit over $h$. It is verified in Appendix \ref{app:order} that the maximization and the limit over $h$ can be interchanged, namely, \eqref{highprob} can be rewritten as follows\footnote{Another approach to ``handle" the limit over $h$ is, to first prove the theorem for a linear system whose frequency response is a staircase function (namely, ``ignoring evaluate" the limit over $h$ in \eqref{propbability1}). Then, using the fact that every frequency response can be approximated arbitrarily well by a sequence of staircase functions with sufficiently small spacing between jumps (Szeg\"{o}'s theorem), the main theorem is proved. Note that \eqref{orderchange} literally means that the partition function for any transfer function is obtained via a limit (w.r.t. $h$) of a sequence of partition functions corresponding to staircase functions with spacings $h$.}
\begin{align}
\frac{\ln Z'_{e}\p{\bY,\blam}}{n}\sim\lim_{h\to0}\operatorname*{max}_{\epsilon\in\mathscr{E}}\ppp{\Gamma\p{\epsilon}-\beta\epsilon}
\label{orderchange}
\end{align}
with probability tending to one. For simplicity of notation, in the following, the notion of \emph{typical} sequences is used to describe an event that is happening with high probability. For example, we say that for a \emph{typical} realization of $\by$, $Z'_{e}\p{\by,\blam}$ is given by the right hand side of \eqref{highprob}, with the meaning that it happens with probability tending to one as $n\to\infty$. Also, in the following, in order not to drag the limit over $h$, it will be omitted and then reverted when it has a role.

Next, an explicit expression for $Z'_{e}\p{\by,\blam}$ is derived. Based on \eqref{powerSet1}, \eqref{corrSet1}, and \eqref{Gammadaffreq2}, $\Gamma\p{\epsilon}$ can be rewritten as
\begin{align}
&\operatorname*{max}_{\ppp{P_i}_{i=1}^k,\ppp{\rho_i}_{i=1}^k}\;\frac{h}{2}\sum_{m=1}^k\ln\p{\frac{P_{m}}{P_x}\p{1-\rho^2_{m}}}\nonumber\\
&\ \ \  \ \ \ \text{s.t.}\ \ \ \ \frac{1}{k}\p{\sum_{m=1}^k\abs{\sigma'_m}\rho_m\sqrt{P_m\tilde{P}_{y,m}}-\frac{1}{2}\abs{\sigma_{m}}^2P_x -\frac{1}{2}\abs{\sigma'_m}^2P_m -\frac{1}{2\beta}} =-\epsilon \nonumber\\
&\ \ \ \ \ \ \ \ \ \ \ \ \ \ \ \frac{1}{k}\sum_{m=1}^kP_m = P_x.
\label{optimizationfff}
\end{align}
\begin{prop}\label{prop1}
Let $\ppp{\mu_i}_{i=1}^k$ be a vector of real scalars such that $\sum_i\mu_i = k$. Then, \eqref{optimizationfff} can be transformed into
\begin{align}
&\operatorname*{max}_{\ppp{P_i}_{i=1}^k,\ppp{\rho_i}_{i=1}^k,\ppp{\mu_i}_{i=1}^k}\;\frac{h}{2}\sum_{m=1}^k\ln\p{\frac{P_{m}}{P_x}\p{1-\rho^2_{m}}}\nonumber\\
&\ \ \  \ \ \ \ \ \ \ \text{s.t.}\ \ \ \ \abs{\sigma'_i}\rho_i\sqrt{P_i\tilde{P}_{y,i}}-\frac{1}{2}\abs{\sigma_{i}}^2P_x -\frac{1}{2}\abs{\sigma'_i}^2P_i-\frac{1}{2\beta} = -\mu_i\epsilon,\;\;i = 1,\ldots,k \nonumber\\
&\ \ \ \ \ \ \ \ \ \ \ \ \ \ \ \frac{1}{k}\sum_{m=1}^kP_m = P_x;\ \ \frac{1}{k}\sum_{m=1}^k\mu_m = 1.
\label{maxNew76}
\end{align}
\end{prop}
\begin{proof}[Proof of Proposition \ref{prop1}]
Given a solution of \eqref{maxNew76}, it is to verify that it is feasible for the optimization problem given by \eqref{optimizationfff}. Conversely, given a solution, $\ppp{P^*_m,\rho_m^*}$, of \eqref{optimizationfff}, by taking 
\begin{align}
\mu_m^* = -\frac{\abs{\sigma'_m}\rho_m^*\sqrt{P_m^*\tilde{P}_{y,m}}-\frac{1}{2}\abs{\sigma_{m}}^2P_x -\frac{1}{2}\abs{\sigma'_m}^2P_m^*-\frac{1}{2\beta}}{\epsilon},
\end{align}
it can be seen that $\ppp{P^*_m,\rho_m^*,\mu_m^*}$ is feasible for \eqref{maxNew76}. Thus, the two problems are equivalent.
\end{proof}
Using the first constraint in \eqref{maxNew76}, the optimization problem in \eqref{maxNew76} can be transformed into
\begin{align}
&\operatorname*{max}_{\ppp{P_i}_{i=1}^k,\ppp{\mu_i}_{i=1}^k}\;\frac{h}{2}\sum_{m=1}^k\ln\ppp{\frac{P_{m}}{P_x}\pp{1-\p{\frac{\frac{1}{2}\abs{\sigma_{m}}^2P_x +\frac{1}{2}\abs{\sigma'_m}^2P_m+\frac{1}{2\beta}-\mu_m\epsilon}{\abs{\sigma'_m}\sqrt{P_m\tilde{P}_{y,m}}}}^2}}\nonumber\\
&\ \ \  \ \ \ \; \text{s.t.}\ \ \ \ \frac{1}{k}\sum_{i=1}^kP_i = P_x;\ \ \frac{1}{k}\sum_{m=1}^k\mu_m = 1.
\label{maxNew77}
\end{align}
Therefore, for a typical realization of the vector $\by$, $Z_e'\p{\by,\blam}$ is given by
\begin{align}
&\frac{\ln Z'_{e}\p{\by,\blam}}{n} \sim\operatorname*{max}_{\boldsymbol{\mathcal{P}}}\operatorname*{max}_{\mathscr{E},\ppp{\mu_i}\in\mathscr{M}_k}\nonumber\\
&\frac{h}{2}\sum_{m=1}^k\p{\ln\ppp{\frac{P_{m}}{P_x}\pp{1-\p{\frac{\frac{1}{2}\abs{\sigma_{m}}^2P_x +\frac{1}{2}\abs{\sigma'_m}^2P_m+\frac{1}{2\beta}-\mu_m\epsilon}{\abs{\sigma'_m}\sqrt{P_m\tilde{P}_{y,m}}}}^2}}-2\beta\mu_m\epsilon},
\label{maxNew88}
\end{align}
in which
\begin{align}
\mathscr{M}_k \define \ppp{\p{\mu_1,\ldots,\mu_k}\in\mathbb{R}^k:\ \frac{1}{k}\sum_{i=1}^k\mu_i = 1}.
\label{conMu2}
\end{align}
Using the subadditivity property of the maximum norm one obtains (for typical $\by$)
\begin{align}
&\frac{\ln Z'_{e}\p{\by,\blam}}{n} \lesssim\operatorname*{max}_{\boldsymbol{\mathcal{P}},\ppp{\mu_i}}\frac{h}{2}\sum_{m=1}^k\nonumber\\
&\operatorname*{max}_{\mathscr{E}}\ln\ppp{\frac{P_{m}}{P_x}\pp{1-\p{\frac{\frac{1}{2}\abs{\sigma_{m}}^2P_x +\frac{1}{2}\abs{\sigma'_m}^2P_m+\frac{1}{2\beta}-\mu_m\epsilon}{\abs{\sigma'_m}\sqrt{P_m\tilde{P}_{y,m}}}}^2}}-2\beta\mu_m\epsilon.
\label{OOff}
\end{align}
Note that except the subadditivity, in the above optimization the maximization is carried over $\ppp{\mu_i}\in\mathbb{R}^k$ rather than $\ppp{\mu_i}\in\mathscr{M}_k$ (as it should be), hence increasing further the bound. Changing the variables, $\mu_m\epsilon\mapsto\epsilon_m$, the values of $\epsilon_m$ for which the derivative vanishes are the solutions of the following equation 
\begin{align}
\frac{2\p{\frac{1}{2\beta}-\epsilon_m+\frac{1}{2}\abs{\sigma_{m}}^2P_x +\frac{1}{2}\abs{\sigma'_m}^2P_m}}{\abs{\sigma'_m}^2\tilde{P}_{y,m}P_m\p{1-\p{\frac{\frac{1}{2}\abs{\sigma_{m}}^2P_x +\frac{1}{2}\abs{\sigma'_m}^2P_m+\frac{1}{2\beta}-\mu_m\epsilon}{\abs{\sigma'_m}\sqrt{P_m\tilde{P}_{y,m}}}}^2}}-2\beta=0,
\end{align}
which after simple algebra, boils down to a quadratic equation whose solutions are
\begin{align}
\label{epsOpt}
\epsilon_{m,1}^* = \frac{2+\abs{\sigma'_m}^2\beta P_m+\abs{\sigma_{m}}^2\beta P_x+\sqrt{1+4\beta^2\abs{\sigma'_m}^2\tilde{P}_{y,m}P_m}}{2\beta}\\
\epsilon_{m,2}^* = \frac{2+\abs{\sigma'_m}^2\beta P_m+\abs{\sigma_{m}}^2\beta P_x-\sqrt{1+4\beta^2\abs{\sigma'_m}^2\tilde{P}_{y,m}P_m}}{2\beta}.
\label{epsOpt4}
\end{align}
Substitution of $\epsilon_{m,1}^*$ in the objective function of \eqref{OOff} reveals that $\epsilon_{m,1}^*$ is not in the objective function domain, and thus only $\epsilon_{m,2}^*$ is considered. In the following, the case $\epsilon_{m,2}^*\in\mathscr{E}$ is first analyzed. Substituting $\epsilon_{m,2}^*$ in \eqref{OOff}, one obtains (for typical $\by$)
\begin{align}
&\frac{\ln Z'_{e}\p{\by,\blam}}{n} \lesssim\operatorname*{max}_{\boldsymbol{\mathcal{P}},\ppp{\mu_i}}\frac{h}{2}\sum_{m=1}^k\ln\ppp{\frac{P_{m}}{P_x}\pp{1-\p{\frac{\frac{1}{2}\abs{\sigma_{m}}^2P_x +\frac{1}{2}\abs{\sigma'_m}^2P_m+\frac{1}{2\beta}-\epsilon_{m,2}^*}{\abs{\sigma'_m}\sqrt{P_mP_{y,m}}}}^2}}-2\beta\epsilon_{m,2}^*.
\label{OOff2}
\end{align}
Let $\gamma$ be the Lagrange multiplier associated with the power constraint. Then, the derivative of the objective function in \eqref{OOff2} w.r.t. $P_m$ is given by
\begin{align}
-\abs{\sigma'_m}^2\beta + \frac{1-\sqrt{1+4\abs{\sigma'_m}^2\beta^2P_m^*\tilde{P}_{y,m}}}{2P_m^*}-\gamma=0,
\end{align}
which vanishes at
\begin{align}
P_m^* = \frac{\abs{\sigma'_m}^2\beta\p{1+\beta\tilde{P}_{y,m}}+\gamma}{\p{\abs{\sigma'_m}^2\beta+\gamma}^2},
\label{PtermAnyLamb}
\end{align}
independently of $\epsilon_{m,2}^*$, and $\gamma$ is chosen such that $\sum_iP_i = kP_x$. Therefore (for typical $\by$),
\begin{align}
\frac{\ln Z'_{e}\p{\by,\blam}}{n} &\lesssim\frac{h}{2}\sum_{m=1}^k\ln\ppp{\frac{P_m^*}{P_x}\pp{1-\p{\frac{\frac{1}{2}\abs{\sigma_{m}}^2P_x +\frac{1}{2}\abs{\sigma'_m}^2P_m^*+\frac{1}{2\beta}-\epsilon_{m}^*}{\abs{\sigma'_m}\sqrt{P_m^*\tilde{P}_{y,m}}}}^2}}-2\beta\epsilon_{m}^*\nonumber\\
&\define F_{\text{par}},
\label{ineqPart}
\end{align}
where $\epsilon_{m}^*\define\epsilon_{m,2}^*\p{P_m^*}$. Hence, an upper bound, $F_{\text{par}}$, on $\ln Z_e'\p{\by,\blam}/n$ is obtained. On the other hand, by taking
\begin{align}
\epsilon^* &= \frac{\sum_{i=1}^k\epsilon_i^*}{k}\\
\mu_m^* &= \frac{k\epsilon_m^*}{\sum_{i=1}^k\epsilon_i^*},
\end{align}
and \eqref{PtermAnyLamb}, this bound is achieved. Summarizing the above results, $Z'_{e}\p{\by,\blam}$ is given by (for typical $\by$)
\begin{align}
&\frac{\ln Z'_{e}\p{\by,\blam}}{n} \sim
\begin{cases}
F_{\text{par}},\ \ &\ \ \Gamma\p{\epsilon^*}+R>0\\
\Gamma\p{\epsilon_s}-\beta\epsilon_s,\ \ &\ \ \Gamma\p{\epsilon^*}+R\leq 0\\ 
\end{cases}.
\end{align}
Since at the final step of the calculation, the partition function (or its derivative w.r.t. $\blam$) is evaluated at $\blam = \bze$, the range $\Gamma\p{\epsilon^*}+R>0$ should be computed at the vicinity of $\blam = \bze$. First, note that $\tilde{P}_{y,i}$, given in Lemma \ref{lem:3}, can be written as
\begin{align}
\tilde{P}_{y,i} = \abs{\sigma_{i}}^2P_x + \frac{1}{\beta} + \frac{2}{\beta}\frac{1}{n_b}\re\p{\frac{1}{\sigma'_i}\sum_{r\in i\text{-th}\text{ bin}}y_r\lambda_r} +\frac{1}{\beta^2\abs{\sigma'_i}^2}\frac{1}{n_b}\sum_{r\in i\text{-th}\text{ bin}}\abs{\lambda_r}^2.
\label{P_yEqautionMis}
\end{align}
Hence, substituting $\blam = \bze$ in \eqref{PtermAnyLamb}, one obtains
\begin{align}
\label{Pmvicinity}
\left.P_m^*\right|_{\blamt=0} &= \frac{\abs{\sigma'_m}^2\beta\p{1+\beta\left.\tilde{P}_{y,m}\right|_{\blamt=0}}+\gamma_0}{\p{\abs{\sigma'_m}^2\beta+\gamma_0}^2}\\
& = \frac{\abs{\sigma'_m}^2\beta\p{2+\beta\abs{\sigma_{m}}^2P_x}+\gamma_0}{\p{\abs{\sigma'_m}^2\beta+\gamma_0}^2}
\label{P_mEq2}
\end{align}
where $\gamma_0$ is chosen such that
\begin{align}
kP_x = \sum_m\left.P_m^*\right|_{\blamt=0}.
\end{align}
Substitution of $\left.P_m^*\right|_{\blamt=0}$ and $\left.\epsilon_m^*\p{\left.P_m^*\right|_{\blamt=0}}\right|_{\blamt=0}$ in $\Gamma\p{\epsilon}$, reveals that 
\begin{align}
&\left.\Gamma\p{\epsilon^*}\right|_{\blamt=0} = \frac{h}{2}\sum_{m=1}^k\ln\ppp{\frac{\left.P_m^*\right|_{\blamt=0}}{P_x}\pp{1-\p{\frac{\frac{1}{2}\abs{\sigma_{m}}^2P_x +\frac{1}{2}\abs{\sigma'_m}^2\left.P_m^*\right|_{\blamt=0}+\frac{1}{2\beta}-\left.\epsilon_{m}^*\right|_{\blamt=0}}{\abs{\sigma'_m}\sqrt{\left.P_m^*\right|_{\blamt=0}\left.\tilde{P}_{y,m}\right|_{\blamt=0}}}}^2}},
\label{GammaSumm}
\end{align}
and that
\begin{align}
\left.\epsilon_m^*\right|_{\blamt=0} =\frac{2+\abs{\sigma'_m}^2\beta \left.P_m^*\right|_{\blamt=0}+\abs{\sigma_{m}}^2\beta P_x-\sqrt{1+4\beta^2\abs{\sigma'_m}^2\tilde{P}_{y,m}\left.P_m^*\right|_{\blamt=0}}}{2\beta}.
\label{epsilonlambda0}
\end{align}
Then, substituting \eqref{epsilonlambda0} in the $m$th term of the sum in \eqref{GammaSumm}, it becomes
\begin{align}
\ln\p{\frac{\sqrt{1+4\abs{\sigma'_m}^2\beta^2\left.P_m^*\right|_{\blamt=0}\left.\tilde{P}_{y,m}\right|_{\blamt=0}}-1}{2P_x\abs{\sigma'_m}^2\beta^2\left.\tilde{P}_{y,m}\right|_{\blamt=0}}},
\label{m_termSum}
\end{align}
which after substitution of \eqref{Pmvicinity}, boils down to
\begin{align}
\frac{1}{P_x\gamma_0+\abs{\sigma'_m}^2P_x\beta}.
\label{finResinner}
\end{align}
Hence, substituting \eqref{finResinner} in \eqref{GammaSumm}, one obtains
\begin{align}
\left.\Gamma\p{\epsilon^*}\right|_{\blamt=0} &=-\frac{h}{2}\sum_{m=1}^k\ln\p{P_x\gamma_0+\abs{\sigma'_m}^2P_x\beta}.
\label{gamgam}
\end{align}
Accordingly, the region $\Gamma\p{\epsilon^*}+R\leq0$ is equivalent to
\begin{align}
R&\leq\frac{h}{2}\sum_{m=1}^k\ln\p{P_x\gamma_0+\abs{\sigma'_m}^2P_x\beta} \define R_e,
\label{R_edef}
\end{align} 
and hence
\begin{align}
&\frac{\ln Z'_{e}\p{\by,\blam}}{n} \sim
\begin{cases}
F_{\text{par}},\ \ &\ \ R>R_e\\
\Gamma\p{\epsilon_s}-\beta\epsilon_s,\ \ &\ \ R\leq R_e\\ 
\end{cases}.
\end{align}
The next step in the evaluation of $Z'\p{\by,\blam}$, is taking into account $Z'_{c}\p{\by,\blam}$. To this end, the following relation is used
\begin{align}
\lim_{n\to\infty}\frac{\ln\p{e^{-na}+e^{-nb}}}{n}= -\min\p{a,b}.
\end{align}
Accordingly, within the range $R>R_e$, for a typical code and realizations of the vector $\by$, we search rates for which $Z'_{c}\p{\by,\bze}>Z'_{e}\p{\by,\bze}$, namely,
\begin{align}
\frac{\ln Z'_{c}\p{\by,\bze}}{n}>\left.F_{\text{par}}\right|_{\blamt=0}.
\label{rangeccc}
\end{align}
Recall that $\left.F_{\text{par}}\right|_{\blamt=0}$ is given by
\begin{align}
\left.F_{\text{par}}\right|_{\blamt=0} &= \Gamma\p{\epsilon^*} - \beta\epsilon^* \\
& = -\frac{h}{2}\sum_{m=1}^k\ppp{\ln\p{P_x\gamma_0+P_x\abs{\sigma'_m}^2\beta}+\left.2\beta\epsilon_m^*\right|_{\blamt=0}},
\end{align}
and that 
\begin{align}
\frac{\ln Z'_{c}\p{\by,\bze}}{n}&=-\p{R+\frac{1}{2}+\frac{\beta}{2n}\norm{\p{\bA'-\bA}\bx_0}^2}\\
& = -R-\frac{1}{2}-\frac{h\beta}{2}\sum_{m=1}^k\abs{\sigma'_m-\sigma_{m}}^2P_x.
\label{ToepProd}
\end{align}
Hence the inequality in \eqref{rangeccc} becomes 
\begin{align}
R&<\frac{h}{2}\sum_{m=1}^k\ln\p{P_x\gamma_0+P_x\abs{\sigma'_m}^2\beta}-\frac{1}{2}-\frac{h\beta}{2}\sum_{m=1}^k\abs{\sigma'_m-\sigma_{m}}^2P_x\nonumber\\
&\ \ \ \ +\frac{h}{2}\sum_{m=1}^k
\ppp{2+\abs{\sigma'_m}^2\beta \left.P_m^*\right|_{\blamt=0}+\abs{\sigma_{m}}^2\beta P_x-\sqrt{1+4\beta^2\abs{\sigma'_m}^2\tilde{P}_{y,m}\left.P_m^*\right|_{\blamt=0}}}\\
&=\frac{h}{2}\sum_{m=1}^k\ln\p{P_x\gamma_0+P_x\abs{\sigma'_m}^2\beta}+ \frac{1}{2} + h\beta\sum_{m=1}^k\re\p{\sigma_m^{'*}\sigma_{m}}P_x\nonumber\\
&\ \ \ \ +\frac{h\beta}{2}\sum_{m=1}^k\abs{\sigma'_m}^2\p{\left.P_m^*\right|_{\blamt=0}-P_x}  - \frac{h}{2}\sum_{m=1}^k\sqrt{1+4\beta^2\abs{\sigma'_m}^2\tilde{P}_{y,m}\left.P_m^*\right|_{\blamt=0}}.
\label{ineqRR}
\end{align}
Substituting $P_m^*$, given in \eqref{P_mEq2}, in the last two terms of \eqref{ineqRR}, one obtains
\begin{align}
R&<\frac{h}{2}\sum_{m=1}^k\ln\p{P_x\gamma_0+P_x\abs{\sigma'_m}^2\beta}+ \frac{1}{2} + h\beta\sum_{m=1}^k\re\p{\sigma_m^{'*}\sigma_{m}}P_x\nonumber\\
&\ \ \ +\frac{h\beta}{2}\sum_{m=1}^k\abs{\sigma'_m}^2\beta\p{\frac{\abs{\sigma'_m}^2\beta\p{2+\beta\abs{\sigma_{m}}^2P_x}+\gamma_0}{\p{\abs{\sigma'_m}^2\beta+\gamma_0}^2}-P_x}\nonumber\\
&\ \ \ -\frac{h}{2}\sum_{m=1}^k\frac{\abs{\sigma'_m}^2\beta\p{3+2P_x\beta\abs{\sigma_{m}}^2}+\gamma_0}{\p{\abs{\sigma'_m}^2\beta+\gamma_0}}.
\label{FerroDom}
\end{align}
Refining the bin sizes by taking the limit $h\to0$, while using Szeg\"{o}'s theorem, it is shown in Appendix \ref{app:5} that \eqref{FerroDom} becomes
\begin{align}
&R<R_e+R_d\define R_c
\label{R_cdef}
\end{align}
where
\begin{align}
R_e\define\frac{1}{4\pi}\int_0^{2\pi}\ln\p{P_x\gamma_0+\abs{\Hmat'\p{\omega}}^2P_x\beta}\mathrm{d}\omega,
\end{align}
and 
\begin{align}
R_d\define&\frac{1}{2} + \beta P_x\frac{1}{2\pi}\int_0^{2\pi}\re\p{\Hmat'^*\p{\omega}\Hmat\p{\omega}}\mathrm{d}\omega\nonumber\\
&+\frac{1}{4\pi}\int_0^{2\pi}\abs{\Hmat'\p{\omega}}^2\beta\p{\frac{\abs{\Hmat'\p{\omega}}^2\beta\p{2+\beta\abs{\Hmat\p{\omega}}^2P_x}+\gamma_0}{\p{\abs{\Hmat'\p{\omega}}^2\beta+\gamma_0}^2}-P_x}\nonumber\\
&-\frac{1}{4\pi}\int_0^{2\pi}\frac{\abs{\Hmat'\p{\omega}}^2\beta\p{3+2P_x\beta\abs{\Hmat\p{\omega}}^2}+\gamma_0}{\p{\abs{\Hmat'\p{\omega}}^2\beta+\gamma_0}}.
\label{R_ddef}
\end{align}
Hence, within the range $R>R_e$, $Z'_{c}\p{\by,\bze}>Z'_{e}\p{\by,\bze}$ (again, typical code and realization vector $\by$) for 
\begin{align}
\ppp{R<R_c}\cap\ppp{R>R_e} = \ppp{R_e<R<R_e+R_d = R_c},
\end{align}
which is a non-empty set if $R_d$ is positive. Next, within the range $R\leq R_e$, $Z'_{c}\p{\by,\bze}>Z'_{e}\p{\by,\bze}$ for rates which satisfy (for typical code and realization of $\by$)
\begin{align}
\frac{\ln Z'_{c}\p{\by,\bze}}{n}>\Gamma\p{\epsilon_s}-\left.\beta\epsilon_s\right|_{\blamt=0}.
\label{rangeccc1}
\end{align}
First, recall that $\epsilon_s$ satisfies $R+\Gamma\p{\epsilon_s}=0$, and hence $\Gamma\p{\epsilon_s} = -R$. Thus, \eqref{rangeccc1} can be rewritten as
\begin{align}
-R-\frac{1}{2}-\frac{h\beta}{2}\sum_{m=1}^k\abs{\sigma'_m-\sigma_{m}}^2P_x > -R-\left.\beta\epsilon_s\right|_{\blamt=0},
\end{align}
which is equivalent to
\begin{align}
\left.\epsilon_s\right|_{\blamt=0}>\frac{1}{2\beta}+\frac{h}{2}\sum_{m=1}^k\abs{\sigma'_m-\sigma_{m}}^2P_x.
\label{glassDom}
\end{align}
Applying $\Gamma\p{\cdot}$ to \eqref{glassDom}, one obtains
\begin{align}
\Gamma\p{\left.\epsilon_s\right|_{\blamt=0}}>\Gamma\p{\frac{1}{2\beta}+\frac{h}{2}\sum_{m=1}^k\abs{\sigma'_m-\sigma_{m}}^2P_x},
\end{align}
and hence
\begin{align}
R<-\left.\Gamma\p{\frac{1}{2\beta}+\frac{h}{2}\sum_{m=1}^k\abs{\sigma'_m-\sigma_{m}}^2P_x}\right|_{\blamt=0}\define R_g,
\label{Rgdefde}
\end{align}
where
\begin{align}
&\Gamma\p{\frac{1}{2\beta}+\frac{h}{2}\sum_{m=1}^k\abs{\sigma'_m-\sigma_{m}}^2P_x} =\nonumber\\ 
&\operatorname*{max}_{\ppp{P_i}_{i=1}^k,\ppp{\rho_i}_{i=1}^k}\;\frac{h}{2}\sum_{m=1}^k\ln\p{\frac{P_{m}}{P_x}\p{1-\rho^2_{m}}}\nonumber\\
&\ \ \  \ \ \ \; \text{s.t.}\ \ \ \ \frac{1}{k}\p{\sum_{m=1}^k\abs{\sigma'_m}\rho_m\sqrt{P_m\tilde{P}_{y,m}} -\frac{1}{2}\abs{\sigma'_m}^2\p{P_m-P_x}-\re\p{\sigma_m^{'*}\sigma_{m}}P_x} =0 \nonumber\\
&\ \ \ \ \ \ \ \ \ \ \ \ \ \ \ \frac{1}{k}\sum_{m=1}^kP_m = P_x.
\label{R_gdef}
\end{align}
To conclude, $Z'\p{\by,\blam}$ is given by (for a typical code and $\by$)
\begin{align}
&\frac{\ln Z'\p{\by,\blam}}{n} \sim
\begin{cases}
F_{\text{par}},\ \ &\ \ R>R_e\vee R_c\\
\Gamma\p{\epsilon_s}-\beta\epsilon_s,\ \ &\ \ R_g<R\leq R_e\\ 
\ln Z'_{c}\p{\by,\blam}/n,\ \ &\ \ \ppp{R_e\leq R\leq R_c}\cup\ppp{R\leq R_g\wedge R_e}\\ 
\end{cases}
\end{align}
where $a\vee b\define\max\p{a,b}$ and $a\wedge b\define\min\p{a,b}$. In the following, the relation
\begin{align}
R_d>0\ \Longrightarrow\ R_e<R_g,
\end{align}
is verified. Recall that $R_d$ follows from the requirement that
\begin{align}
-\p{R+\frac{1}{2}+\frac{\beta}{2n}\norm{\p{\bA-\bA'}\bx_0}^2}&=\frac{\ln Z'_{c}\p{\by,\bze}}{n}\\
&\geq\left.F_{\text{par}}\right|_{\blamt=0} \\
&= \Gamma\p{\epsilon^*} - \left.\beta\epsilon^*\right|_{\blamt=0}=-R_e- \left.\beta\epsilon^*\right|_{\blamt=0},
\end{align}
which can be rewritten as
\begin{align}
R\leq R_e + \left.\beta\epsilon^*\right|_{\blamt=0} - \p{\frac{1}{2}+\frac{\beta}{2n}\norm{\p{\bA'-\bA}\bx_0}^2},
\end{align}
and thus $R_d$ is given by
\begin{align}
R_d = \left.\beta\epsilon^*\right|_{\blamt=0} - \p{\frac{1}{2}+\frac{\beta}{2n}\norm{\p{\bA'-\bA}\bx_0}^2}.
\end{align}
Accordingly, $R_d>0$ is equivalent to
\begin{align}
\left.\beta\epsilon^*\right|_{\blamt=0} > \frac{1}{2}+\frac{\beta}{2n}\norm{\p{\bA-\bA'}\bx_0}^2.
\end{align}
Now, within the range $R\leq R_e$, $Z'_{c}\p{\by,\bze}\geq Z'_{e}\p{\by,\bze}$ if \eqref{glassDom} 
\begin{align}
\left.\beta\epsilon_s\right|_{\blamt=0}>\frac{1}{2}+\frac{\beta}{2n}\norm{\p{\bA-\bA'}\bx_0}^2.
\label{dff}
\end{align}
However, $R\leq R_e$ is equivalent to $\epsilon^*\notin\mathscr{E}$, and thus $\left.\epsilon_s\right|_{\blamt=0}\geq\left.\epsilon^*\right|_{\blamt=0}$. Therefore, if $R_d>0$, the following holds 
\begin{align}
\left.\beta\epsilon_s\right|_{\blamt=0}\geq\left.\beta\epsilon^*\right|_{\blamt=0}> \frac{1}{2}+\frac{\beta}{2n}\norm{\p{\bA-\bA'}\bx_0}^2.
\end{align}
Whence, \eqref{dff} holds true within the whole region $R\leq R_e$, and therefore $R_e< R_g$. Thus, for $R_d>0$, $Z'\p{\by,\blam}$ becomes (for a typical code realization $\by$)
\begin{align}
&\frac{\ln Z'\p{\by,\blam}}{n} \sim
\begin{cases}
F_{\text{par}},\ \ &\ \ R>R_c\\
\ln Z'_{c}\p{\by,\blam}/n,\ \ &\ \ R\leq R_c 
\end{cases}.
\end{align}
If however, $R_d<0$, then $R_g\leq R_e$, and hence (for a typical code realization $\by$)
\begin{align}
&\frac{\ln Z'\p{\by,\blam}}{n} \sim
\begin{cases}
F_{\text{par}},\ \ &\ \ R>R_e\\
-R-\beta\epsilon_s,\ \ &\ \ R_g<R\leq R_e\\ 
\ln Z'_{c}\p{\by,\blam}/n,\ \ &\ \ R\leq R_g
\end{cases}.
\label{PartitionFunctionLambdaDep}
\end{align}
Recall that $\epsilon_s$ is the solution of the equation
\begin{align}
\Gamma\p{\epsilon_s} + R=0,
\label{gammaepsilonEql}
\end{align}
where $\Gamma\p{\epsilon_s}$ is given by
\begin{align}
&\operatorname*{max}_{\ppp{P_i}_{i=1}^k,\ppp{\rho_i}_{i=1}^k}\;\frac{h}{2}\sum_{m=1}^k\ln\p{\frac{P_{m}}{P_x}\p{1-\rho^2_{m}}}\nonumber\\
&\ \ \  \ \ \ \; \text{s.t.}\ \ \ \ \frac{1}{k}\p{\sum_{m=1}^k\abs{\sigma'_m}\rho_m\sqrt{P_m\tilde{P}_{y,m}}-\frac{1}{2}\abs{\sigma_{m}}^2P_x -\frac{1}{2}\abs{\sigma'_m}^2P_m -\frac{1}{2\beta}} =-\epsilon_s \nonumber\\
&\ \ \ \ \ \ \ \ \ \ \ \ \ \ \ \frac{1}{k}\sum_{m=1}^kP_m = P_x.
\end{align}
Similarly to the optimization problem in \eqref{maxNew77}, the above maximization problem can be rewritten as
\begin{align}
&\operatorname*{max}_{\ppp{P_i}_{i=1}^k,\ppp{\rho_i}_{i=1}^k}\;\frac{h}{2}\sum_{m=1}^k\ln\ppp{\frac{P_{m}}{P_x}\pp{1-\p{\frac{\frac{1}{2}\abs{\sigma_{m}}^2P_x +\frac{1}{2}\abs{\sigma'_m}^2P_m +\frac{1}{2\beta}-\mu_m\epsilon_s}{\sigma'_m\sqrt{P_m\tilde{P}_{y,m}}}}^2}}\nonumber\\
&\ \ \  \ \ \ \; \text{s.t.}\ \ \ \ \frac{1}{k}\sum_{m=1}^kP_m = P_x,\;\frac{1}{k}\sum_{m=1}^k\mu_m = 1.
\end{align}
Accordingly, the derivative of the objective function w.r.t. $P_m$ vanishes at
\begin{align}
P_m^* = \frac{4\alpha_1\epsilon_s^2+\abs{\sigma'_m}^2\alpha_2\p{\tilde{P}_{y,m}\alpha_2+2\epsilon_s}}{\p{\abs{\sigma'_m}^2\alpha_2+2\alpha_1\epsilon_s}^2}
\label{opglassysol1l}
\end{align}
and the derivative w.r.t. $\mu_m$ it vanishes at
\begin{align}
\mu_m^* =& \frac{4\alpha_1^2\epsilon_s^2\p{1+P_x\beta\abs{\sigma_{m}}^2}+4\abs{\sigma'_m}^2\alpha_1\epsilon_s\p{\alpha_2-\beta\tilde{P}_{y,m}\alpha_2+\beta\epsilon_s+P_x\beta\alpha_2\abs{\sigma_{m}}^2}}{2\beta\epsilon_s\p{\abs{\sigma'_m}^2\alpha_2+2\alpha_1\epsilon_s}^2}\nonumber\\
& + \frac{\abs{\sigma'_m}^4\alpha_2\p{\alpha_2-\beta\tilde{P}_{y,m}\alpha_2 + 2\beta\epsilon_s+P_x\beta\alpha_2\abs{\sigma_{m}}^2}}{2\beta\epsilon_s\p{\abs{\sigma'_m}^2\alpha_2+2\alpha_1\epsilon_s}^2}
\label{opglassysol2l}
\end{align}
where $\alpha_1$ is chosen such that $kP_x = \sum_mP_m^*$, and $\alpha_2$ is chosen such that $k = \sum_m\mu_m^*$. Substituting the above maximizers in the objective function one obtains
\begin{align}
\Gamma\p{\epsilon_s} =\frac{h}{2}\sum_{m=1}^k\ln\p{\frac{2\epsilon_s}{P_x\abs{\sigma'_m}^2\alpha_2+2P_x\alpha_1\epsilon_s}}.
\label{gammaSolglassy}
\end{align}
For completeness, a closed-form expression for $R_g$ is derived. Based on \eqref{Rgdefde}
\begin{align}
R_g = -\left.\Gamma\p{\frac{1}{2\beta}+\frac{h}{2}\sum_{m=1}^k\abs{\sigma'_m-\sigma_{m}}^2P_x}\right|_{\blamt=0}.
\end{align}
Using \eqref{gammaSolglassy}, and upon taking the limit $h\to0$ (while using Szeg\"{o}'s theorem, as was done in \eqref{R_ddef})
\begin{align}
R_g = -\frac{1}{4\pi}\int_0^{2\pi}\ln\p{\frac{2\tilde{\epsilon}}{P_x\abs{\Hmat'\p{\omega}}^2\tilde{\alpha}_2+2P_x\tilde{\alpha}_1\tilde{\epsilon}}}\mathrm{d}\omega
\label{RgSzego}
\end{align}
where $\tilde{\alpha}_1$ and $\tilde{\alpha}_2$ solve the simultaneous equations
\begin{align}
&\frac{1}{2\pi}\int_0^{2\pi}\frac{4\tilde{\alpha}_1\tilde{\epsilon}^2+\abs{\Hmat'\p{\omega}}^2\tilde{\alpha}_2\pp{\p{\abs{\Hmat\p{\omega}}^2P_x+\frac{1}{\beta}}\tilde{\alpha}_2+2\tilde{\epsilon}}}{\p{\abs{\Hmat'\p{\omega}}^2\tilde{\alpha}_2+2\tilde{\alpha}_1\tilde{\epsilon}}^2}\mathrm{d}\omega = P_x\\
&\frac{1}{2\pi}\int_0^{2\pi}\frac{4\tilde{\alpha}_1^2\tilde{\epsilon}^2\p{1+P_x\beta\abs{\Hmat\p{\omega}}^2}+4\abs{\Hmat'\p{\omega}}^2\tilde{\alpha}_1\tilde{\epsilon}^2\beta+2\abs{\Hmat'\p{\omega}}^4\tilde{\alpha}_2\tilde{\epsilon}\beta}{2\beta\tilde{\epsilon}\p{\abs{\Hmat'\p{\omega}}^2\tilde{\alpha}_2+2\tilde{\alpha}_1\tilde{\epsilon}}^2}\mathrm{d}\omega = 1,
\end{align}
and
\begin{align}
\tilde{\epsilon} = \frac{1}{2\beta}+\frac{P_x}{4\pi}\int_0^{2\pi}\abs{\Hmat'\p{\omega}-\Hmat\p{\omega}}^2\mathrm{d}\omega.
\end{align}

Obtaining $Z'\p{\by,\blam}$, using the tools presented in Subsection \ref{sec:back}, the MSE is now derived. The MSE estimator of the $i$th component (\emph{chip}) of $\bx'$, within the $q$th bin, is given by the derivative of $Z'\p{\by,\blam}$ w.r.t. $\lambda_{q_i}$ evaluated at $\blam=\bze$\footnote{A very similar analysis applies also to the derivative $\frac{\partial}{\partial\lambda_i}\ln Z\p{\by,\blam}$, which is essentially a weighted average over $x_i$ with weights proportional to $\bE{\calN\p{\epsilon}}e^{-\beta\epsilon}$ for $\epsilon\in\mathscr{E}$. Thus, the exponentially dominant weight is due to the term that maximizes the exponent \cite{Neri1,Neri2}. Hence, in this case, the commutativity between the derivative w.r.t. $\blam$ and the limit $n\to\infty$ is legitimate. Another approach to justify the interchange of the order of these operations is to use well-known results (for example, \cite[Ch. 16]{Vladimir},\cite{Frink,Steinlage}) on functional properties of a limit function, which are applicable in our case due to the uniform convergence of the various relevant terms (see Appendix \ref{app:order}).}. The derivative of $F_{\text{par}}$ is given by
\begin{align}
\left.\frac{\partial F_{\text{par}}}{\partial\lambda_{q_i}}\right|_{\blamt=0} = -\left.\frac{h}{2}\sum_{l=1}^k\ppp{\frac{P_x}{P_x\gamma_0+P_x\abs{\sigma'_l}^2\beta}\frac{\partial \gamma}{\partial\lambda_{q_i}}\right|_{\blamt=0} + \left.\frac{\partial \psi_l}{\partial\lambda_{q_i}}\right|_{\blamt=0}}.
\label{ParDiv}
\end{align}
Let $x_q\define\left.\partial \gamma/\partial\lambda_{q_i}\right|_{\blamt=0}$. Using \eqref{PtermAnyLamb}, one obtains
\begin{align}
\left.\frac{\partial P_q^*}{\partial\lambda_{q_i}}\right|_{\blamt=0} = &\frac{\p{\abs{\sigma'_q}^2\beta^2\left.\frac{\partial \tilde{P}_{y,q}}{\partial\lambda_{q_i}}\right|_{\blamt=0}+x_q}\p{\abs{\sigma'_q}^2+\gamma_0}^2}{\p{\abs{\sigma'_q}^2+\gamma_0}^4}\nonumber\\
&-\frac{2\p{\abs{\sigma'_q}^2+\gamma_0}\p{\abs{\sigma'_q}^2\beta\p{1+\left.\beta\tilde{P}_{y,q}\right|_{\lambda = 0}}+\gamma_0}x_q}{\p{\abs{\sigma'_q}^2+\gamma_0}^4},
\label{optimalpq}
\end{align}
and for $l\neq q$
\begin{align}
\left.\frac{\partial P_l^*}{\partial\lambda_{q_i}}\right|_{\blamt=0} =\frac{\p{\abs{\sigma'_l}^2+\gamma_0}^2x_q-2\p{\abs{\sigma'_l}^2+\gamma_0}\p{\abs{\sigma'_l}^2\beta\p{1+\left.\beta\tilde{P}_{y,l}\right|_{\lambda = 0}}+\gamma_0}x_q}{\p{\abs{\sigma'_l}^2+\gamma_0}^4}
\label{optimalpl}
\end{align}
where by using \eqref{P_yEqautionMis}
\begin{align}
&\left.\frac{\partial \tilde{P}_{y,q}}{\partial\lambda_{q_i}}\right|_{\blamt=0} = \frac{2}{\beta\sigma'_q}\frac{y_{q_i}}{n_b}\\
&\left.\beta\tilde{P}_{y,q}\right|_{\lambda = 0} = 1+\abs{\sigma_{q}}^2P_x\beta.
\end{align}
Since $\gamma_0$ is chosen to satisfy $\left.\sum_rP_r^*\right|_{\blamt=0} = kP_x$, it follows that
\begin{align}
0 &= \left.\frac{\partial}{\partial\lambda_{q_i}}\sum_{r=1}^kP_r^*\right|_{\blamt=0} = \left.\sum_{r=1}^k\frac{\partial P_r^*}{\partial\lambda_{q_i}}\right|_{\blamt=0}\\
& = \frac{\abs{\sigma'_q}^2\beta^2\left.\frac{\partial \tilde{P}_{y,q}}{\partial\lambda_{q_i}}\right|_{\blamt=0}\p{\abs{\sigma'_q}^2+\gamma_0}^2}{\p{\abs{\sigma'_q}^2+\gamma_0}^4}\nonumber\\
& + x_q\sum_{r=1}^k\frac{\p{\abs{\sigma'_r}^2+\gamma_0}^2-2\p{\abs{\sigma'_r}^2+\gamma_0}\p{\abs{\sigma'_r}^2\beta\p{1+\left.\beta\tilde{P}_{y,r}\right|_{\lambda = 0}}+\gamma_0}}{\p{\abs{\sigma'_r}^2+\gamma_0}^4},
\end{align}
and thus
\begin{align}
x_q = \frac{\abs{\sigma'_q}^2\beta^2\left.\frac{\partial \tilde{P}_{y,q}}{\partial\lambda_{q_i}}\right|_{\blamt=0}}{\p{\abs{\sigma'_q}^2+\gamma_0}^2C} = \frac{2\sigma_q^{'*}\beta}{n_b\p{\abs{\sigma'_q}^2+\gamma_0}^2C}y_{q_i},
\label{div11}
\end{align}
where
\begin{align}
C \define \sum_{r=1}^k\frac{\p{\abs{\sigma'_r}^2+\gamma_0}-2\p{\abs{\sigma'_r}^2\beta\p{1+\left.\beta\tilde{P}_{y,r}\right|_{\lambda = 0}}+\gamma_0}}{\p{\abs{\sigma'_r}^2+\gamma_0}^3}.
\end{align}
Next, $\left.\partial \psi_l/\partial\lambda_{q_i}\right|_{\blamt=0}$, is calculated. Using the definition of $\psi_m$ in \eqref{psiM} one obtains
\begin{align}
\left.\frac{\partial \psi_q}{\partial\lambda_{q_i}}\right|_{\blamt=0} = \abs{\sigma'_q}^2\beta\left.\frac{\partial P_q^*}{\partial\lambda_{q_i}}\right|_{\blamt=0}-\frac{2\beta^2\abs{\sigma'_q}^2\p{\left.\frac{\partial P_{y,q}^*}{\partial\lambda_{q_i}}\right|_{\blamt=0}\left.P_q^*\right|_{\blamt=0}+\left.\tilde{P}_{y,q}\right|_{\blamt=0}\left.\frac{\partial P_q^*}{\partial\lambda_{q_i}}\right|_{\blamt=0}}}{\sqrt{1+4\beta^2\abs{\sigma'_q}^2\left.\tilde{P}_{y,q}\right|_{\blamt=0}\left.P_q^*\right|_{\blamt=0}}},
\label{div21}
\end{align}
and for $l\neq q$
\begin{align}
\left.\frac{\partial \psi_l}{\partial\lambda_{q_i}}\right|_{\blamt=0} = \abs{\sigma'_l}^2\beta\left.\frac{\partial P_l^*}{\partial\lambda_{q_i}}\right|_{\blamt=0}-\frac{2\beta^2\abs{\sigma'_l}^2\left.\tilde{P}_{y,l}\right|_{\blamt=0}\left.\frac{\partial P_l^*}{\partial\lambda_{q_i}}\right|_{\blamt=0}}{\sqrt{1+4\beta^2\abs{\sigma'_l}^2\left.\tilde{P}_{y,l}\right|_{\blamt=0}\left.P_l^*\right|_{\blamt=0}}}.
\label{div22}
\end{align}
Substituting \eqref{div11}, \eqref{div21} and \eqref{div22} in \eqref{ParDiv}, the MSE estimator in the range $R>R_c$ and $R>R_e$, for $R_d>0$ and $R_d<0$, respectively, (note that all the terms are dependent on $y_{q_i}$ linearly via $x_q$) is given by\footnote{The relation between the right and the left hand sides of \eqref{asyeqrand} is an asymptotic equality between two random variables, in the sense that the difference between them converges to zero w.p. 1.}
\begin{align}
\bE'\ppp{X_{q_i}\mid\bY}\sim\left.\frac{\partial nF_{\text{par}}}{\partial\lambda_{q_i}}\right|_{\blamt=0} &= \xi_{1,q}Y_{q_i}
\label{asyeqrand}
\end{align}
where
\begin{align}
&\xi_{1,q}=-\frac{1}{2}\frac{2\sigma_q^{'*}\beta}{\p{\abs{\sigma'_q}^2+\gamma_0}^2C}\sum_{l=1}^k\ppp{\frac{P_x}{P_x\gamma_0+P_x\abs{\sigma'_l}^2\beta}+B_l-C_l}\nonumber\\
&\ \ \ \ \ \ \ \ -\frac{1}{2}\left(1-\frac{2\p{\abs{\sigma'_q}^2+\gamma_0}^2\left.P_q^*\right|_{\blamt=0}}{\sqrt{1+4\beta^2\abs{\sigma'_q}^2\left.\tilde{P}_{y,q}\right|_{\blamt=0}\left.P_q^*\right|_{\blamt=0}}}\right.\nonumber\\
&\ \ \ \ \ \ \ \ \ \ \ \ \ \ \ \ \left.-\frac{2\beta^2\abs{\sigma'_q}^2\left.\tilde{P}_{y,q}\right|_{\blamt=0}}{\sqrt{1+4\beta^2\abs{\sigma'_q}^2\left.\tilde{P}_{y,q}\right|_{\blamt=0}\left.P_q^*\right|_{\blamt=0}}}\right)\frac{2\sigma_q^{'*}\beta}{\p{\abs{\sigma'_q}^2+\gamma_0}^2},
\label{xi1coeff}
\end{align}
with
\begin{align}
B_l &= \frac{\p{\abs{\sigma'_l}^2+\gamma_0}^2-2\p{\abs{\sigma'_l}^2+\gamma_0}\p{\abs{\sigma'_l}^2\beta\p{1+\left.\beta\tilde{P}_{y,l}\right|_{\lambda = 0}}+\gamma_0}}{\p{\abs{\sigma'_l}^2+\gamma_0}^4}\\
C_l &= \frac{2\beta^2\abs{\sigma'_l}^2\left.\tilde{P}_{y,l}\right|_{\blamt=0}B_l}{\sqrt{1+4\beta^2\abs{\sigma'_l}^2\left.\tilde{P}_{y,l}\right|_{\blamt=0}\left.P_l^*\right|_{\blamt=0}}}.
\label{CoeffPara}
\end{align}
Next, the MSE estimator in the region $R_g<R\leq R_e$ for $R_d<0$ is derived. The derivative of the partition function w.r.t. $\lambda_{q_i}$ is given by
\begin{align}
\left.\frac{\partial F_{\text{glas}}}{\partial\lambda_{q_i}}\right|_{\blamt=0} = -\left.\beta\frac{\partial \epsilon_s}{\partial\lambda_{q_i}}\right|_{\blamt=0}.
\label{GlsDiv}
\end{align}
Recall that $\epsilon_s$ is the solution of the equation
\begin{align}
\Gamma\p{\epsilon_s} + R=0,
\label{gammaepsilonEq}
\end{align}
where $\Gamma\p{\epsilon_s}$ is given as
\begin{align}
&\operatorname*{max}_{\ppp{P_i}_{i=1}^k,\ppp{\rho_i}_{i=1}^k}\;\frac{h}{2}\sum_{m=1}^k\ln\p{\frac{P_{m}}{P_x}\p{1-\rho^2_{m}}}\nonumber\\
&\ \ \  \ \ \ \; \text{s.t.}\ \ \ \ \frac{1}{k}\p{\sum_{m=1}^k\abs{\sigma'_m}\rho_m\sqrt{P_m\tilde{P}_{y,m}}-\frac{1}{2}\abs{\sigma_{m}}^2P_x -\frac{1}{2}\abs{\sigma'_m}^2P_m -\frac{1}{2\beta}} =-\epsilon_s \nonumber\\
&\ \ \ \ \ \ \ \ \ \ \ \ \ \ \ \frac{1}{k}\sum_{m=1}^kP_m = P_x.
\label{nnn}
\end{align}
Similarly to the optimization problem in \eqref{maxNew77}, the maximization problem in \eqref{nnn} can be rewritten as
\begin{align}
&\operatorname*{max}_{\ppp{P_i}_{i=1}^k,\ppp{\rho_i}_{i=1}^k}\;\frac{h}{2}\sum_{m=1}^k\ln\ppp{\frac{P_{m}}{P_x}\pp{1-\p{\frac{\frac{1}{2}\abs{\sigma_{m}}^2P_x +\frac{1}{2}\abs{\sigma'_m}^2P_m +\frac{1}{2\beta}-\mu_m\epsilon_s}{\sigma'_m\sqrt{P_m\tilde{P}_{y,m}}}}^2}}\nonumber\\
&\ \ \  \ \ \ \; \text{s.t.}\ \ \ \ \frac{1}{k}\sum_{m=1}^kP_m = P_x,\;\frac{1}{k}\sum_{m=1}^k\mu_m = 1.
\label{nnn2}
\end{align}
The derivative of the objective function w.r.t. $P_m$ vanishes at
\begin{align}
P_m^* = \frac{4\alpha_1\epsilon_s^2+\abs{\sigma'_m}^2\alpha_2\p{\tilde{P}_{y,m}\alpha_2+2\epsilon_s}}{\p{\abs{\sigma'_m}^2\alpha_2+2\alpha_1\epsilon_s}^2}
\label{opglassysol1}
\end{align}
and the derivative w.r.t. $\mu_m$, vanishes at
\begin{align}
\mu_m^* =& \frac{4\alpha_1^2\epsilon_s^2\p{1+P_x\beta\abs{\sigma_{m}}^2}+4\abs{\sigma'_m}^2\alpha_1\epsilon_s\p{\alpha_2-\beta\tilde{P}_{y,m}\alpha_2+\beta\epsilon_s+P_x\beta\alpha_2\abs{\sigma_{m}}^2}}{2\beta\epsilon_s\p{\abs{\sigma'_m}^2\alpha_2+2\alpha_1\epsilon_s}^2}\nonumber\\
& + \frac{\abs{\sigma'_m}^4\alpha_2\p{\alpha_2-\beta\tilde{P}_{y,m}\alpha_2 + 2\beta\epsilon_s+P_x\beta\alpha_2\abs{\sigma_{m}}^2}}{2\beta\epsilon_s\p{\abs{\sigma'_m}^2\alpha_2+2\alpha_1\epsilon_s}^2}
\label{opglassysol2}
\end{align}
where $\alpha_1$ is chosen to such that $kP_x = \sum_mP_m^*$, and $\alpha_2$ is chosen such that $k = \sum_m\mu_m^*$. Substituting the above maximizers in the objective function of \eqref{nnn2} one obtains
\begin{align}
\Gamma\p{\epsilon_s} =\frac{h}{2}\sum_{m=1}^k\ln\p{\frac{2\epsilon_s}{P_x\abs{\sigma'_m}^2\alpha_2+2P_x\alpha_1\epsilon_s}}.
\end{align} 
Thus, \eqref{gammaepsilonEq} becomes
\begin{align}
\frac{h}{2}\sum_{m=1}^k\ln\p{\frac{2\epsilon_s}{P_x\abs{\sigma'_m}^2\alpha_2+2P_x\alpha_1\epsilon_s}} + R=0.
\label{gammaepsilonEq2}
\end{align}
Let $x_q \define \left.\partial\epsilon_s/\partial\lambda_{q_i}\right|_{\blamt=0}$, $\dot{\alpha}_{1,q} \define \left.\partial\alpha_1/\partial\lambda_{q_i}\right|_{\blamt=0}$, $\dot{\alpha}_{2,q} \define \left.\partial\alpha_2/\partial\lambda_{q_i}\right|_{\blamt=0}$, $\alpha_{1,0} \define \left.\alpha_1\right|_{\blamt=0}$, $\alpha_{2,0} \define \left.\alpha_2\right|_{\blamt=0}$ and $\epsilon_{s,0} \define \left.\epsilon_s\right|_{\blamt=0}$. Differentiating \eqref{gammaepsilonEq2} w.r.t. $\lambda_{q_i}$ one obtains
\begin{align}
0&=\sum_{m=1}^k\frac{P_x\abs{\sigma'_m}^2\alpha_{2,0}+2P_x\alpha_{1,0}\epsilon_{s,0}}{2\epsilon_{s,0}}\nonumber\\
&\times\frac{2\p{P_x\abs{\sigma'_m}^2\alpha_{2,0}+2P_x\alpha_{1,0}\epsilon_{s,0}}x_q-2\epsilon_{s,0}\p{P_x\abs{\sigma'_m}^2\dot{\alpha}_{2,q}+2P_x\dot{\alpha}_{1,q}\epsilon_{s,0}+2P_x\alpha_{1,0}x_q}}{\p{P_x\abs{\sigma'_m}^2\alpha_{2,0}+2P_x\alpha_{1,0}\epsilon_{s,0}}^2}\\
&=\sum_{m=1}^k\frac{P_x\abs{\sigma'_m}^2\alpha_{2,0}x_q-\epsilon_{s,0}\p{P_x\abs{\sigma'_m}^2\dot{\alpha}_{2,q}+2P_x\dot{\alpha}_{1,q}\epsilon_{s,0}}}{\epsilon_{s,0}\p{P_x\abs{\sigma'_m}^2\alpha_{2,0}+2P_x\alpha_{1,0}\epsilon_{s,0}}},
\end{align}
and thus
\begin{align}
x_q = \frac{U_1}{U_2}
\label{xsolsol}
\end{align}
where
\begin{align}
U_1\define\sum_{m=1}^k\frac{P_x\abs{\sigma'_m}^2\dot{\alpha}_{2,q}+2P_x\dot{\alpha}_{1,q}\epsilon_{s,0}}{P_x\abs{\sigma'_m}^2\alpha_{2,0}+2P_x\alpha_{1,0}\epsilon_{s,0}}
\end{align}
and
\begin{align}
U_2\define\sum_{m=1}^k\frac{P_x\abs{\sigma'_m}^2\alpha_{2,0}}{\epsilon_{s,0}\p{P_x\abs{\sigma'_m}^2\alpha_{2,0}+2P_x\alpha_{1,0}\epsilon_{s,0}}}.
\end{align}
Hence, in order to calculate $x_q$ one needs to find $\epsilon_{s,0},\alpha_{1,0},\alpha_{2,0},\dot{\alpha}_{1,q},\dot{\alpha}_{2,q}$. The terms $\epsilon_{s,0},\alpha_{1,0},\alpha_{2,0}$ are calculated using the set of simultaneous equations
\begin{subequations}
\begin{align}
&\left.\Gamma\p{\epsilon_s}\right|_{\blamt=0}+R=0\\
&\frac{1}{k}\sum_{m=1}^k\left.P_m^*\right|_{\blamt=0} = P_x\\
&\frac{1}{k}\sum_{m=1}^k\left.\mu_m^*\right|_{\blamt=0} = 1,
\end{align}
\end{subequations}
and accordingly, the terms $\dot{\alpha}_{1,q},\dot{\alpha}_{2,q}$ are calculated using the set of equations
\begin{subequations}
\begin{align}
\label{PowerLageq}
&\sum_{m=1}^k\left.\frac{\partial P_m^*}{\partial\lambda_{q_i}}\right|_{\blamt=0} = 0\\
&\sum_{m=1}^k\left.\frac{\partial \mu_m^*}{\partial\lambda_{q_i}}\right|_{\blamt=0} = 0.
\label{MuLageq}
\end{align}
\end{subequations}
Given $\epsilon_{s,0},\alpha_{1,0},\alpha_{2,0}$, closed-form expressions for $\dot{\alpha}_{1,q},\dot{\alpha}_{2,q}$ are now derived. Using \eqref{opglassysol1}, \eqref{PowerLageq} can be written as
\begin{align}
\eta_1\dot{\alpha}_{1,q}+\eta_2\dot{\alpha}_{2,q}+\eta_3x_q + \eta_q\abs{\sigma'_q}^2\dot{\tilde{P}}_{y,q} = 0
\label{diffeq1}
\end{align}
where $\dot{\tilde{P}}_{y,q}\define \left.\partial\tilde{P}_{y,q}/\partial\lambda_{q_i}\right|_{\blamt=0}$, and
\begin{align}
&\eta_1 \define \sum_{m=1}^k\frac{4D_m\epsilon_{s,0}^2-4R_m\epsilon_{s,0}}{D_m^3}\\
&\eta_2 \define \sum_{m=1}^k\frac{D_m\abs{\sigma'_m}^2\pp{\p{\abs{\sigma_{m}}^2P_x+\frac{1}{\beta}}\alpha_{2,0}+2\epsilon_{s,0}}}{D_m^3}\nonumber\\
&\ \ \ \ \  \ \ \ \ \ \ \ \ +\frac{\abs{\sigma'_m}^2\alpha_{2,0}\p{\abs{\sigma_{m}}^2P_x+\frac{1}{\beta}}D_m-2R_m\abs{\sigma'_m}^2}{D_m^3}\\
&\eta_3 \define \sum_{m=1}^k\frac{8D_m\alpha_{1,0}\epsilon_{s,0}+2D_m\abs{\sigma'_m}^2\alpha_{2,0}-4R_m\alpha_{1,0}}{D_m^3}\\
&\eta_q \define \frac{\alpha_{2,0}^2}{D_q^2},
\end{align}
in which
\begin{align}
&D_m \define \abs{\sigma'_m}^2\alpha_{2,0}+2\alpha_{1,0}\epsilon_{s,0}\\
&R_m \define 4\alpha_{1,0}\epsilon_{s,0}^2 + \abs{\sigma'_m}^2\alpha_{2,0}\pp{\alpha_{2,0}\p{\abs{\sigma_{m}}^2P_x+2\epsilon_{s,0}}}.
\end{align}
Similarly, using \eqref{opglassysol2}, \eqref{MuLageq} can be written as
\begin{align}
\gamma_1\dot{\alpha}_{1,q}+\gamma_2\dot{\alpha}_{2,q}+\gamma_3x_q + \gamma_q\abs{\sigma'_q}^2\dot{\tilde{P}}_{y,q} = 0
\label{diffeq2}
\end{align}
where
\begin{align}
&\gamma_1 \define \sum_{m=1}^k\frac{8\alpha_{1,0}\epsilon_{s,0}^2\p{1+P_x\beta\abs{\sigma_{m}}^2}+4\beta\abs{\sigma'_m}^2\epsilon_{s,0}^2}{K_m}\nonumber\\
& \ \ \ \ \ \ \ \ \ \ \ \ \ -\frac{8T_m\beta\epsilon_{s,0}^2\p{\abs{\sigma'_m}^2\alpha_{2,0}+2\alpha_{1,0}\epsilon_{s,0}}}{K_m^2}\\
&\gamma_2 \define \sum_{m=1}^k\frac{2K_m\beta\epsilon_{s,0}\abs{\sigma'_m}^4-4T_m\beta\epsilon_{s,0}\p{\abs{\sigma'_m}^2\alpha_{2,0}+2\alpha_{1,0}\epsilon_{s,0}}\abs{\sigma'_m}^2}{K_m^2}\\
&\gamma_3 \define \sum_{m=1}^k\frac{8\alpha_{1,0}^2\epsilon_{s,0}\p{1+P_x\beta\abs{\sigma_{m}^2}}+8\beta\epsilon_{s,0}\abs{\sigma'_m}^2\alpha_{1,0}+2\beta\alpha_{2,0}\abs{\sigma'_m}^4}{K_m}\nonumber\\
&-\frac{T_m\pp{2\beta\p{\abs{\sigma'_m}^2\alpha_{2,0}+2\alpha_{1,0}\epsilon_{s,0}}^2+8\beta\epsilon_{s,0}\alpha_{1,0}\p{\abs{\sigma'_m}^2\alpha_{2,0}+2\alpha_{1,0}\epsilon_{s,0}}}}{K_m^2}\\
&\gamma_q \define \frac{-4\beta\alpha_{1,0}\epsilon_{s,0}\alpha_{2,0}-\beta\alpha_{2,0}^2\abs{\sigma'_q}^2}{K_q^2},
\end{align}
in which
\begin{align}
&K_m \define 2\beta\epsilon_{s,0}\p{\abs{\sigma'_m}^2\alpha_{2,0}+2\alpha_{1,0}\epsilon_{s,0}}^2\\
&T_m \define 4\alpha_{1,0}^2\epsilon_{s,0}^2\p{1+P_x\beta\abs{\sigma_{m}}^2}+4\beta\abs{\sigma'_m}^2\alpha_{1,0}\epsilon_{s,0}^2+2\abs{\sigma'_m}^4\alpha_{2,0}\beta\epsilon_{s,0}.
\end{align}
Thus, solving the pair of equations, \eqref{diffeq1} and \eqref{diffeq2}, one obtains
\begin{align}
&\dot{\alpha}_{1,q} = \frac{\gamma_3/\gamma_2-\eta_3/\eta_2}{\eta_1/\eta_2-\gamma_1/\gamma_2}x_q + \frac{\gamma_q/\gamma_2-\eta_q/\eta_2}{\eta_1/\eta_2-\gamma_1/\gamma_2}\abs{\sigma'_q}^2\dot{\tilde{P}}_{y,q} \define r_1x_q + J_{1q}\abs{\sigma'_q}^2\dot{\tilde{P}}_{y,q}\\
&\dot{\alpha}_{2,q} = \frac{\gamma_3/\gamma_1-\eta_3/\eta_1}{\eta_2/\eta_1-\gamma_2/\gamma_1}x_q + \frac{\gamma_q/\gamma_1-\eta_q/\eta_1}{\eta_2/\eta_1-\gamma_2/\gamma_1}\abs{\sigma'_q}^2\dot{\tilde{P}}_{y,q}\define r_2x_q + J_{2q}\abs{\sigma'_q}^2\dot{\tilde{P}}_{y,q}.
\end{align}
Substituting $\dot{\alpha}_{1,q}$ and $\dot{\alpha}_{2,q}$ in \eqref{xsolsol}, simple rearrangement of terms reveals that
\begin{align}
x_q = \frac{J_{1,q}\sum_{m=1}^k2\epsilon_{s,0}^2P_x/Q_m+J_{2,q}\sum_{m=1}^k2\epsilon_{s,0}P_x\abs{\sigma'_m}^2/Q_m}{V\p{1-F}}\abs{\sigma'_q}^2\dot{\tilde{P}}_{y,q}
\end{align}
where
\begin{align}
&V \define \sum_{m=1}^k\frac{P_x\abs{\sigma'_m}^2\alpha_{2,0}}{\epsilon_{s,0}\p{P_x\abs{\sigma'_m}^2\alpha_{2,0}+2P_x\alpha_{1,0}\epsilon_{s,0}}}\\
&F \define \frac{1}{V}\sum_{m=1}^k\frac{\epsilon_{s,0}\p{P_x\abs{\sigma'_m}^2r_2 + 2P_x\epsilon_{s,0}r_1}}{\epsilon_{s,0}\p{P_x\abs{\sigma'_m}^2\alpha_{2,0}+2P_x\alpha_{1,0}\epsilon_{s,0}}}\\
&Q_m \define \epsilon_{s,0}\p{P_x\abs{\sigma'_m}^2\alpha_{2,0}+2P_x\alpha_{1,0}\epsilon_{s,0}}.
\end{align}
Let
\begin{align}
J_q &\define  \frac{J_{1,q}\sum_{m=1}^k2\epsilon_{s,0}^2P_x/Q_m+J_{2,q}\sum_{m=1}^k2\epsilon_{s,0}P_x\abs{\sigma'_m}^2/Q_m}{V\p{1-F}},
\label{GlassEstCoee}
\end{align}
and so
\begin{align}
x_q = J_q\abs{\sigma'_q}^2\dot{\tilde{P}}_{y,q} = J_q\frac{2\sigma'^*_q}{n_b\beta}y_{q_i}.
\end{align}
Therefore,
\begin{align}
\bE'\ppp{X_{q_i}\mid \bY}&\sim\left.\frac{\partial nF_{\text{glas}}}{\partial\lambda_{q_i}}\right|_{\blamt=0} \\
&= -\left.n\beta\frac{\partial \epsilon_s}{\partial\lambda_{q_i}}\right|_{\blamt=0}=\xi_{q,2}\cdot y_{q_i} .
\end{align}
where
\begin{align}
\xi_{q,2} \define -J_q\frac{2\sigma'^*_q}{h}.
\label{xi2coeff}
\end{align}
Finally, the mismatched MSE estimator in the region $R\leq R_g$ and $R\leq R_c$ for $R_d<0$ and $R_d>0$, respectively, is derived. Based on \eqref{ZQc}, it readily follows that
\begin{align}
\bE'\ppp{X_{q_i}\mid \bY}&\sim\left.\frac{\partial \ln Z_{Q,c}}{\partial\lambda_{q_i}}\right|_{\blamt=0}= X_{q_i}.
\end{align}

To conclude, the mismatched MSE estimator is given as follows.
\newline For $R_d\geq0$ 
 \begin{align}
&\bE'\ppp{X_{q_i}\mid\bY} \sim
\begin{cases}
X_{q_i},\ \ &\ \ R\leq R_c\\
\xi_{q,1} Y_{q_i},\ \ &\ \ R> R_c 
\end{cases}.
\label{MSEestimator1pp}
\end{align}
For $R_d<0$ 
\begin{align}
&\bE'\ppp{X_{q_i}\mid\bY} \sim
\begin{cases}
X_{q_i},\ \ &\ \ R\leq R_g\\
\xi_{q,2} Y_{q_i},\ \ &\ \ R_g<R\leq R_e\\ 
\xi_{q,1} Y_{q_i},\ \ &\ \ R> R_e
\end{cases}
\label{MSEestimator2app}
\end{align}
where the above equalities are asymptotic equalities between two random variables, in the sense that the difference between them converges to zero in probability.

The mismatched MSE is given by
\begin{align}
\text{mse}\p{\bX\mid\bY} &= \sum_{i=1}^n\bE\ppp{X_i^2}-2\re\p{\bE\ppp{\bE\p{X_i\mid\bY}\bE'^*\p{X_i\mid\bY}}}\nonumber\\
&\ \ \ \ \ \ \ \ \ \ \ \ \ \ \ +\bE\ppp{\abs{\bE'\p{X_i\mid\bY}}^2}.
\label{MSEqGen}
\end{align}
Therefore, based on \eqref{MSEqGen}, in order to calculate the MSE, the MMSE estimator should be obtained first. Substituting $\bA = \bA'$ in $R_d$, given in \eqref{R_ddef}, one can see that $R_d=0$. Thus, the MMSE estimator is given by
\begin{align}
\bE\ppp{X_{q_i}\mid\bY}\sim
\begin{cases}
\xi_{1,q}Y_{q_i},\ \ &\ \ R>R_{e}\\
X_{q_i},\ \ &\ \ R\leq R_{e}\\ 
\end{cases}.
\end{align}
In order to find $R_e$, according to \eqref{R_edef}, $\gamma_0$ is needed. However, in this case it can readily be verified that $\gamma_0 = 1/P_x$, and thus
\begin{align}
R_{c,M} \define R_e = \frac{1}{4\pi}\int_0^{2\pi}\ln\p{1+\abs{\Hmat\p{\omega}}^2\beta P_x}\mathrm{d}\omega.
\end{align}
Finally, substitution of $\sigma_{m} = \sigma'_m$ in \eqref{xi1coeff}, reveals that
\begin{align}
\xi_{1,q} = \frac{\beta\sigma_{q}^*P_x}{1+\abs{\sigma_{q}}^2P_x\beta},
\end{align}
and thus
\begin{align}
\bE\ppp{X_{q_i}\mid\bY}\sim
\begin{cases}
\frac{\beta\sigma_{q}^*P_x}{1+\abs{\sigma_{q}}^2P_x\beta}Y_{q_i},\ \ &\ \ R>R_{c,M}\\
X_{q_i},\ \ &\ \ R\leq R_{c,M}\\ 
\end{cases}.
\end{align}
Based on the second term of the sum in \eqref{MSEqGen}, several cases should be considered. For $R_d>0$, since $R_c<R_{c,M}$, there are three regions: $R<R_c$, $R_{c}<R<R_{c,M}$ and $R>R_{c,M}$. For $R<R_c$, both the matched and the mismatched estimators are asymptotically equal to $X_{q_i}$ with high probability, and thus
\begin{align}
\text{mse}\p{\bX\mid\bY} = 0.
\end{align}
For $R_{c}<R<R_{c,M}$ one readily obtains
\begin{align}
\frac{\text{mse}\p{\bX\mid\bY}}{n} = P_x-2h\re\p{\sum_{m=1}^k\xi_{m,1}^*\sigma^*_{m}P_x}+h\sum_{m=1}^k\abs{\xi_{m,1}}^2\p{\abs{\sigma_{m}}^2P_x+\frac{1}{\beta}},
\end{align}
and similarly, for $R_{c,M}<R$,
\begin{align}
\frac{\text{mse}\p{\bX\mid\bY}}{n} &= P_x-2h\re\p{\sum_{m=1}^k\xi_{m,1}^*\frac{\beta\sigma_{m}^*P_x}{1+\abs{\sigma_{m}}^2P_x\beta}\p{\abs{\sigma_{m}}^2P_x+\frac{1}{\beta}}}\nonumber\\
&+h\sum_{m=1}^k\abs{\xi_{m,1}}^2\p{\abs{\sigma_{m}}^2P_x+\frac{1}{\beta}}\\
& = P_x-2h\re\p{\sum_{m=1}^k\xi_{m,1}^*\sigma^*_{m}P_x}+h\sum_{m=1}^k\abs{\xi_{m,1}}^2\p{\abs{\sigma_{m}}^2P_x+\frac{1}{\beta}}.
\end{align}
Thus, the MSE's in the last two ranges are the same. In the same way, the MSE for $R_d<0$ is calculated. For $R\leq R_g$
\begin{align}
\text{mse}\p{\bX\mid\bY} = 0.
\end{align}
For $R_g<R\leq R_e$
\begin{align}
\frac{\text{mse}\p{\bX\mid\bY}}{n} = P_x-2h\re\p{\sum_{m=1}^k\xi_{m,2}^*\sigma^*_{m}P_x}+h\sum_{m=1}^k\abs{\xi_{m,2}}^2\p{\abs{\sigma_{m}}^2P_x+\frac{1}{\beta}}\define\text{mse}_{1},
\end{align}
and for $R> R_e$
\begin{align}
\frac{\text{mse}\p{\bX\mid\bY}}{n} = P_x-2h\re\p{\sum_{m=1}^k\xi_{m,1}^*\sigma^*_{m}P_x}+h\sum_{m=1}^k\abs{\xi_{m,1}}^2\p{\abs{\sigma_{m}}^2P_x+\frac{1}{\beta}}\define\text{mse}_{2}.
\end{align}
Finally, take the limit $h\to0$ (after $n\to\infty$). Using Szeg\"{o}'s theorem (as was done in \eqref{R_ddef}), one obtains ($i=1,2$)
\begin{align}
\lim_{n\to\infty}\text{mse}_{i} = P_x-\frac{P_x}{\pi}\int_0^{2\pi}\re\p{\bXi_i^*\p{\omega}\Hmat^*\p{\omega}}\mathrm{d}\omega+\frac{1}{2\pi}\int_0^{2\pi}\abs{\bXi_i\p{\omega}}^2\p{\abs{\Hmat\p{\omega}}^2+\frac{1}{\beta}}\mathrm{d}\omega
\label{freqrep}
\end{align}
where $\bXi_i\p{\omega}$, for $i=1,2$, are given in \eqref{Chi1} and \eqref{Chi2}. 

In the matched case, for $R>R_{c,M}$
\begin{align}
\text{mmse}\p{\bX\mid\bY} &= \sum_{i=1}^n \bE\ppp{X_i^2}-\bE\ppp{\abs{\bE\ppp{X_i\mid\bY}}^2}\\
& = nP_x-\sum_{m=1}^k\sum_{i_m = 1}^{n_b}\bE\ppp{\abs{\bE\ppp{X_{i_m}\mid\bY}}^2}\\
& = nP_x-\sum_{m=1}^k\abs{\frac{\beta\sigma_m^*P_x}{1+\abs{\sigma_m}^2P_x\beta}}^2n_b\p{\abs{\sigma_m}^2P_x+\frac{1}{\beta}}\\
& = nP_x-n\sum_{m=1}^kh\frac{\abs{\sigma_m}^2P_x^2}{\frac{1}{\beta}+\abs{\sigma_m}^2P_x} = n\sum_{m=1}^kh\frac{P_x}{1+\abs{\sigma_m}^2P_x\beta},
\end{align}
which upon taking the limit $h\to0$, becomes
\begin{align}
\lim_{n\to\infty}\frac{\text{mmse}\p{\bX\mid\bY}}{n} = \frac{1}{2\pi}\int_0^{2\pi}\frac{P_x}{1+\abs{\Hmat\p{\omega}}^2P_x\beta}\mathrm{d}\omega.
\end{align} 

\begin{remark}[Generalization to Any Input Spectral Distribution]
As was mentioned in Section \ref{sec:body}, the above analysis can be modified to hold for any input spectral density $S_x\p{\omega}$. Technically speaking, the following modification should be considered: Let $P_{x,m}$ be the (real) transmitted power over the $m$th bin. Then, because of the separable form of the partition function over the bins, we will essentially obtain exactly the same results with the exception of $P_{x,m}$ instead of $P_x$. Precisely, instead of $P_x$ which appears in the numerator of the logarithm function in \eqref{maxNew88}, one should simply replace it to $P_{x,m}$. Following the same lines of derivation, at the final stage of the refinement of the bin sizes, we will finally obtain the spectral density $S_x\p{\omega}$ as a limit function of $\ppp{P_{x,m}}_m$. 
\end{remark}	

\section{Conclusion}
\label{sec:Conclusion}
In this paper, we considered the problem of mismatched estimation of codewords corrupted by a Gaussian vector channel. The derivation was build upon a simple relation between the MSE and a certain function, which can be viewed as a partition function, and hence be analyzed using methods of statistical mechanics. As a special case, the MMSE estimator and its respective estimation error was derived. In particular, it was shown that the MSE essentially separated into two cases each exhibiting a different behavior: In one case, the MSE exhibits single phase transition, which divides the MSE into ferromagnetic and paramagnetic phases. In the other case, the MSE exhibits two phase transitions, which divide the MSE into three phases consisting of the two previous phases and a third glassy phase. Then, using the theoretical results obtained, a few numerical examples were analyzed, by exploring the phase diagrams and the MSE's as functions of the mismatched parameters in each problem. This leads to physical intuitions regarding the threshold effects and the role of the mismatched measure in creating them. Indeed, it was shown that the aforementioned separation of the MSE is linked to pessimism and optimism behaviors of the receiver, according to its mismatched assumption on the channel. Note that in contrast to previous related papers \cite{Neri1,Neri2}, in which the explored examples did not completely emphasize the necessity of the use of the analysis techniques of statistical physics for deriving the MSE, we believe that the considered problem in this paper does, as standard information theoretic approaches do not lend themselves to rigorous analysis. Finally, we believe that the tools developed in this paper for handling optimum estimation problems, can be used in other applications. One such application, which has been already considered for a simple model is estimation of signals of partial support \cite[Section V. D]{Neri2} which has motivation in compressed sensing applications. It would be natural to generalize the model considered in \cite[Section V. D]{Neri2} to a much more rich and applicable one (in the spirit of the considered model in this paper), and perhaps assessing the MSE using the concepts developed in this paper.


\appendices
\numberwithin{equation}{section}
\section{Proof of Lemma \ref{lem:3}}
\label{app:2}
\begin{proof}
We first show the inclusion
\begin{align}
\calT_\delta\p{\bx\mid\by}\subseteq\hat{\calT}_\delta\p{\bx\mid\by},
\end{align}
namely, for any $\bx\in \calT_\delta\p{\bx\mid\by}$ also $\bx\in\hat{\calT}_\delta\p{\bx\mid\by}$. Recall that
\begin{align}
\mathscr{B}_{m}^\delta\p{P_m,\rho_m} &\define \left\{\vphantom{\abs{\frac{\sum_{i=\p{m-1}n_b+1}^{mn_b}\tilde{y}_i\tilde{x}_i}{\sqrt{\p{\frac{1}{n_b}\sum_{i=\pi_{m-1}+1}^{\pi_{m}}\tilde{y}_i^2}P_m}}-n_{b}\rho_m}\leq\delta}\bx \in\mathbb{R}^{n_b}:\;\abs{\norm{\bx_{\p{m-1}n_b+1}^{mn_b}}^2-n_{b}P_m}\leq\delta,\right.\nonumber\\
&\left.\ \ \ \ \ \ \ \ \ \ \ \ \ \ \ \ \ \ \;\abs{\re\ppp{\sum_{i\in\mathcal{I}_m}\sigma'_i\bar{y}_i{x}_i}-n_{b}\rho_m{\sqrt{\tilde{P}_{y,m}\tilde{P}_{\sigma,m}}}}\leq\delta\right\},
\end{align}
and that
\begin{align}
\calT_\delta\p{\bx\mid\by} \define \ppp{\bx\in\mathbb{R}^n:\;\abs{\norm{\bx}^2-nP_x}\leq\delta,\;\abs{\frac{\norm{\by-\bSig'\bx}^2}{2}-\frac{\blam^T\bx}{\beta}-n\epsilon}\leq\delta}.
\label{TypeSeta}
\end{align}
First, note that the second constraint in \eqref{TypeSeta} can be rewritten as
\begin{align}
\abs{\re\ppp{\frac{1}{n}\sum_{i=1}^n\sigma'_i\bar{y}_i{x}_i} - \rho}\leq\delta
\label{rhodefa}
\end{align}
where
\begin{align}
\tilde{\rho}=\frac{\frac{1}{n}\sum_{i=1}^n\abs{\sigma'_ix_i}^2 + P_y-2\epsilon}{2}.
\end{align}
Then, for any $\bx\in \calT_\delta\p{\bx\mid\by}$, we first show that there exist a sequence $\ppp{P_m}_{m=1}^k\in\boldsymbol{\mathcal{P}}^{\delta}$ such that for any $1\leq m\leq k$,
\begin{align}
\abs{\norm{\bx_{\p{m-1}n_b+1}^{mn_b}}^2-n_{b}P_m}\leq\delta.
\label{alocatePP}
\end{align}
To this end, for each $1\leq m\leq k$, $P_m$ is chosen to be the nearest point to $\norm{\bx_{\p{m-1}n_b+1}^{mn_b}}^2$ in the set $\calG_{1,\delta}^k$, namely $P_m = \left\lfloor \norm{\bx_{\p{m-1}n_b+1}^{mn_b}}^2/\p{n_b\delta}\right\rfloor\cdot\delta$. Under this choice, obviously, \eqref{alocatePP} holds, and $\ppp{P_m}_{m=1}^k\in\boldsymbol{\mathcal{P}}^{\delta}$, since
\begin{align}
\abs{\frac{1}{k}\sum_{m=1}^kP_m-P_x} &= \abs{\frac{1}{k}\sum_{m=1}^k\left\lfloor \frac{\norm{\bx_{\p{m-1}n_b+1}^{mn_b}}^2}{n_b\delta}\right\rfloor\delta-P_x}\\
&\leq\abs{\frac{1}{n}\sum_{m=1}^k\norm{\bx_{\p{m-1}n_b+1}^{mn_b}}^2\delta-P_x}\leq\delta
\end{align}
where the last equality follows from the fact that $\bx\in \calT_\delta\p{\bx\mid\by}$. Next, we show that there exist a sequence $\ppp{\rho_m}_{m=1}^k\in\boldsymbol{\mathcal{R}}^{\delta}_{\boldsymbol{\mathcal{P}}}$ such that for any $1\leq m\leq k$,
\begin{align}
\abs{\re\ppp{\sum_{i\in\mathcal{I}_m}\sigma'_i\bar{y}_i{x}_i}-n_{b}\rho_m{\sqrt{\tilde{P}_{y,m}\tilde{P}_{\sigma,m}}}}\leq\delta.
\label{alocaterr}
\end{align} 
Similarly, by taking 
\begin{align}
\rho_m = \left\lfloor \frac{\re\ppp{\sum_{i\in\mathcal{I}_m}\sigma'_i\bar{y}_i{x}_i}}{n_b\delta\sqrt{\tilde{P}_{y,m}\tilde{P}_{\sigma,m}}}\right\rfloor\cdot\delta\in\calG_{2,\delta}^k,
\end{align}
obviously, \eqref{alocaterr} holds, and also $\ppp{P_m,\rho_m}\in\boldsymbol{\mathcal{P}}^{\delta}\cap\boldsymbol{\mathcal{R}}^{\delta}_{\boldsymbol{\mathcal{P}}}$, since
\begin{align}
\abs{\frac{1}{k}\sum_{m=1}^k\rho_m\sqrt{\tilde{P}_{y,m}\tilde{P}_{\sigma,m}}-\tilde{\rho}} &= \abs{\frac{1}{k}\sum_{m=1}^k\left\lfloor \frac{\re\ppp{\sum_{i\in\mathcal{I}_m}\sigma'_i\bar{y}_i{x}_i}}{n_b\delta\sqrt{\tilde{P}_{y,m}\tilde{P}_{\sigma,m}}}\right\rfloor\cdot\delta\sqrt{\tilde{P}_{y,m}\tilde{P}_{\sigma,m}}-\tilde{\rho}}\\
&\leq \abs{\frac{1}{n}\sum_{m=1}^k\re\ppp{\sum_{i\in\mathcal{I}_m}\sigma'_i\bar{y}_i{x}_i}-\tilde{\rho}}\\
& = \abs{\re\ppp{\frac{1}{n}\sum_{i=1}^n\sigma'_i\bar{y}_i{x}_i} - \rho}\leq\delta
\end{align}
where the last equality follows from the fact that $\bx\in \calT_\delta\p{\bx\mid\by}$. For the second inclusion, we need to show that $\hat{\calT}_{\delta/k}\p{\bx\mid\by}\subseteq\calT_\delta\p{\bx\mid\by}$. For any $\bx\in\hat{\calT}_{\delta/k}\p{\bx\mid\by}$
\begin{align}
\abs{\norm{\bx}^2-nP_x} &= \abs{\sum_{m=1}^k\norm{\bx_{\p{m-1}n_b+1}^{mn_b}}^2-nP_x}\nonumber\\
&= \abs{\sum_{m=1}^k\norm{\bx_{\p{m-1}n_b+1}^{mn_b}}^2-\sum_{m=1}^kn_bP_m}\nonumber\\
&\leq\sum_{m=1}^k\abs{\norm{\bx_{\p{m-1}n_b+1}^{mn_b}}^2-n_bP_m}\leq k\frac{\delta}{k} = \delta
\end{align}
where the second equality follows from the definition of $\boldsymbol{\mathcal{P}}^\delta$, the third inequality follows from the triangle inequality, and the forth inequality follows from the definition of $\mathscr{B}_{m}^\delta\p{P_m,\rho_m}$. In the same way, for any $\bx\in\hat{\calT}_{\delta/k}\p{\bx\mid\by}$
\begin{align}
\abs{\re\ppp{\sum_{i=1}^n\sigma'_i\bar{y}_i{x}_i} - n\rho{\sqrt{P_{\bar{y}}\p{\frac{1}{n}\sum_{i=1}^n\abs{\sigma'_ix_i}^2}}}} &= \abs{\sum_{m=1}^k\re\ppp{\sum_{i\in\calI_m}\sigma'_i\bar{y}_i{x}_i} - n_b\sum_{m=1}^k\rho_m\sqrt{\tilde{P}_{y,m}\tilde{P}_{\sigma,m}}}\nonumber\\
&\leq k\frac{\delta}{k} = \delta
\end{align}
where the first equality follows from the definition of $\boldsymbol{\mathcal{R}}_{\boldsymbol{\mathcal{P}}}^\delta$, and the second inequality follows from the triangle inequality and the definition of $\mathscr{B}_{m}^\delta\p{P_m,\rho_m}$. Thus $\hat{\calT}_{\delta/k}\p{\bx\mid\by}\subseteq\calT_\delta\p{\bx\mid\by}\subseteq\hat{\calT}_{\delta}\p{\bx\mid\by}$.
\end{proof}

\section{Proof of Lemma \ref{lem:1}}
\label{app:1}
\begin{proof}
For simplicity of notation, the following conventions are used. Calculating the volume of $\mathscr{B}_{m}^\delta\p{P_m,\rho_m}$ is equivalent to calculating the volume of the set
\begin{align}
\mathcal{F}_\delta\p{P_x,\rho}\define\ppp{\bx \in\mathbb{R}^{n}:\;\abs{\norm{\bx}^2-nP_x}\leq\delta, \;\abs{\re\ppp{\sum_{i=1}^nx_iy_i^*}-n\rho{\sqrt{P_xP_y}}}\leq\delta}
\label{fcals}
\end{align}
where $P_x \define \norm{\bx}^2/n$ and $P_y \define \norm{\by}^2/n$, for a given vector $\by\in\mathbb{C}^n$. Due to the symmetry of the vectors $\bx$ and $\by$ in the DFT domain (recall that in the time domain the considered vectors are real), i.e., $x_i = x^*_{n-i}$ for $i=1,\ldots,n$ (and similarly for $\by$), for the volume calculation of \eqref{fcals}, only vectors with dimension $n/2$ should be considered, while the other half is fixed. Accordingly, the constraints in \eqref{fcals} take the form
\begin{align}
\abs{\sum_{i=1}^{n/2}\abs{x_i}^2-\frac{n}{2}P_x}\leq\delta,
\end{align}
and
\begin{align}
\abs{\re\ppp{\sum_{i=1}^{n/2}x_iy_i^*}-\frac{n}{2}\rho{\sqrt{P_xP_y}}}\leq\delta.
\end{align}
Let $m \define n/2$. Consider the following Gaussian measure
\begin{align}
\mathrm{d}\gamma_G^{m}\p{\bx} \define \frac{1}{\p{\pi\vartheta^2}^{m}}\exp\ppp{-\frac{1}{\vartheta^2}\sum_{i=1}^{m}\abs{x_i-ay_i}^2}\mathrm{d}\bx
\label{GaussianMeasure}
\end{align}
where $a,\vartheta^2\in\mathbb{R}$. Then,
\begin{align}
1 = \gamma_G^{m}\ppp{\mathbb{R}^m}&\geq\gamma_G^{m}\ppp{\mathcal{F}_\delta\p{P_x,\rho}}\\
& = \int_{\mathcal{F}_\delta}\frac{1}{\p{\pi\vartheta^2}^{m}}\exp\ppp{-\frac{1}{\vartheta^2}\sum_{i=1}^{m}\abs{x_i-ay_i}^2}\mathrm{d}\bx\\
&\geq\int_{\mathcal{F}_\delta}\frac{1}{\p{\pi\vartheta^2}^{m}}\exp\ppp{-\frac{m}{\vartheta^2}\pp{P_x+\delta-2a\sqrt{P_xP_y}\p{\rho-\delta}+a^2P_y}}\mathrm{d}\bx\\
& = \text{Vol}\ppp{\mathcal{F}_\delta\p{P_x,\rho}}\frac{1}{\p{\pi\vartheta^2}^{m}}\exp\ppp{-\frac{m}{\vartheta^2}\pp{P_x+\delta-2a\sqrt{P_xP_y}\p{\rho-\delta}+a^2P_y}}.
\label{ap2}
\end{align}
It is easy to verify that 
\begin{align}
a_o &\define \sqrt{\frac{P_x}{P_y}}\p{\rho-\delta},
\end{align}
and
\begin{align}
\vartheta_o^{2}&\define P_x+\delta-2a\sqrt{P_xP_y}\p{\rho-\delta}+a^2P_y\\
& = P_x+\delta-P_x\p{\rho-\delta}^2
\end{align}
maximize the right hand side of \eqref{ap2} (w.r.t. $a$ and $\vartheta^2$). Thus, on the one hand,
\begin{align}
\text{Vol}\ppp{\mathcal{F}_\delta\p{P_x,\rho}}\leq\exp\ppp{m\ln\p{\pi e\vartheta_o^{2}}}.
\label{eeqq1}
\end{align}
On the other hand,
\begin{align}
1 &= \gamma_G^{m}\ppp{\mathcal{F}_\delta\p{P_x,\rho}\cup\mathcal{F}^c_\delta\p{P_x,\rho}}\\
& = \int_{\mathcal{F}_\delta}\frac{1}{\p{\pi\vartheta^2}^{m}}\exp\ppp{-\frac{1}{\vartheta^2}\sum_{i=1}^{m}\abs{x_i-ay_i}^2}\mathrm{d}\bx + \gamma_G^{m}\ppp{\mathcal{F}^c_\delta\p{P_x,\rho}}\\
&\leq\text{Vol}\ppp{\mathcal{F}_\delta\p{P_x,\rho}}\exp\ppp{-m\ln\p{\pi e\vartheta_{o,u}^{2}}}+ \gamma_G^{m}\ppp{\mathcal{F}^c_\delta\p{P_x,\rho}}
\end{align}
where the last inequality follows by the same considerations as before, and
\begin{align}
\vartheta_{o,u}^{2} = P_x-\delta-P_x\p{\rho+\delta}^2.
\end{align}
Using Boole's inequality
\begin{align}
\gamma_G^{m}\ppp{\mathcal{F}^c_\delta\p{P_x,\rho}}&\leq\gamma_G^{m}\ppp{\bx:\abs{\norm{\bx}^2-mP_x}>\delta}+\gamma_G^{m}\ppp{\bx:\abs{\re\ppp{\sum_{i=1}^nx_iy_i^*}-m\rho{\sqrt{P_xP_y}}}>\delta}.
\label{twomeasures}
\end{align}
It is easy to verify that the parameters $a$ and $\vartheta$ that are maximizing the Gaussian measure are given by
\begin{align}
a_M &= \frac{\re\p{\sum_{i=1}^mx_iy_i^*}}{P_y}\define \frac{\tilde{\rho}}{P_y}\\
\vartheta_M^2 &= \frac{1}{m}\sum_{i=1}^m\abs{x_i}^2-\frac{\re\p{\sum_{i=1}^mx_iy_i^*}^2}{P_y}\define \tilde{P}_x-\frac{\tilde{\rho}^2}{P_y}
\end{align}
where $\tilde{\rho}$ and $\tilde{P}_x$ are the empirical correlation and the input variance, respectively. Let $\gamma_{G,M}$ denote the Gaussian measure associated with the parameters $a_M,\vartheta_M$, namely, $\gamma_{G,M}$ is given by \eqref{GaussianMeasure} with $a = a_M$ and $\vartheta^2 = \vartheta_M^2$. Accordingly, it is easy to verify that under $\gamma_{G,M}$, the following hold
\begin{align}
&\bE_{\gamma_{G,M}}\ppp{\norm{\bX}^2} = mP_x,
\end{align}
and
\begin{align}
&\bE_{\gamma_{G,M}}\ppp{\re\pp{\sum_{i=1}^n{X_iy_i^*}}} = m\sqrt{P_xP_y}\rho.
\end{align} 
Thus, using the LLN, the two terms on the right hand side of \eqref{twomeasures} are negligible as $m\to\infty$, namely,
\begin{align}
\gamma_G^{m}\ppp{\mathcal{F}^c_\delta\p{P_x,\rho}}&\leq \epsilon
\label{twomeasures2}
\end{align}
for any $\epsilon>0$. Thus,
\begin{align}
\text{Vol}\ppp{\mathcal{F}_\delta\p{P_x,\rho}}\geq\p{1-\epsilon}\exp\ppp{m\ln\p{\pi e\vartheta_{o,u}^{2}}}.
\label{eeqq2}
\end{align} 
Finally, combining \eqref{eeqq1}, \eqref{eeqq2}, and taking the limit $\delta\to0$, the lemma follows.
\end{proof}

\section{Proof of Lemma \ref{lem:4}}
\label{app:4}
\begin{proof}
Recall that
\begin{align}
\bE\ppp{\mathcal{N}\p{\epsilon}} \exe \exp\ppp{n\p{R+\Gamma\p{\epsilon}}},
\end{align}
and that \footnote{Given $\by$, $\mathcal{N}\p{\epsilon}$ is a sum of $M-1$ i.i.d. Bernoulli random variables and therefore its variance is $\p{M-1}p\p{1-p}$, where $p$ is the success probability, which in our case, was shown to be given by $p = \exp\ppp{n\Gamma\p{\epsilon}}$.}
\begin{align}
\text{var}\ppp{\calN\p{\epsilon}} \exe \exp\ppp{n\p{R+\Gamma\p{\epsilon}}}\p{1- \exp\ppp{n\Gamma\p{\epsilon}}}.
\end{align}
Thus,
\begin{align}
\frac{\text{var}\ppp{\calN\p{\epsilon}}}{\p{\bE\ppp{\mathcal{N}\p{\epsilon}}}^2}\exe\exp\ppp{-n\p{R+\Gamma\p{\epsilon}}}.
\end{align}
For any $\epsilon\notin\mathscr{E}$, the expectation of $\calN\p{\epsilon}$ can be written as $\bE\ppp{\mathcal{N}\p{\epsilon}} \exe e^{-nC_1}$ where $C_1 = R+\Gamma\p{\epsilon}>0$. Thus, by Markov inequality (since $\calN\p{\epsilon}\in\mathbb{N}\cup\ppp{0}$)
\begin{align}
\mathbb{P}\ppp{\calN\p{\epsilon}>0}\leq\bE\ppp{\mathcal{N}\p{\epsilon}}\exe e^{-nC_1}.
\label{est1}
\end{align}
On the other hand, for any $\epsilon\in\mathscr{E}$ and $\delta>0$, using Chebyshev's inequality
\begin{align}
\mathbb{P}\ppp{\abs{\frac{\calN\p{\epsilon}}{\bE\ppp{\mathcal{N}\p{\epsilon}}}-1}>\delta}\leq\frac{\text{var}\ppp{\calN\p{\epsilon}}}{\gamma\p{\bE\ppp{\mathcal{N}\p{\epsilon}}}^2}\exe e^{-nC_2}
\label{est2}
\end{align}
where $C_2 = R+\gamma\p{\epsilon}>0$. Thus, in this case, $\calN\p{\epsilon}$ is concentrated very strongly around $\bE\ppp{\mathcal{N}\p{\epsilon}}$. Finally, let $\calA_n \define \ppp{\abs{\calN\p{\epsilon}-\bE\ppp{\mathcal{N}\p{\epsilon}}\Ind\ppp{\mathscr{E}}}>\delta}$. Then, using \eqref{est1} and \eqref{est2}, it is easy to verify that
\begin{align}
\sum_{i=1}^\infty\mathbb{P}\p{\calA_n}<\infty.
\end{align}
Thus, using Borel-Cantelli Lemma, one obtains that
\begin{align}
\mathbb{P}\ppp{\limsup_{n\to\infty}\calA_n} = 0,
\end{align}
and hence \eqref{statement} follows.
\end{proof}

\section{Proof of \eqref{orderchange}}
\label{app:order}
Equation \eqref{orderchange} follows by the following lemma.
\begin{lemma}\label{lem:5}
Let $f:\mathbb{R}\times\mathbb{R}\to\mathbb{R}$ be a smooth function such that 
\begin{align}
g\p{x} = \lim_{h\to a}f\p{x,h},
\label{lemAsu}
\end{align}
uniformly for every $x\in\mathbb{R}$. Assume that $\lim_{h\to a}\operatorname*{max}_xf\p{x,h}$ exist. Then,
\begin{align}
\lim_{h\to a}\operatorname*{max}_xf\p{x,h} = \operatorname*{max}_x\lim_{h\to a}f\p{x,h}.
\end{align}
\end{lemma}
\begin{proof}[Proof of Lemma \ref{lem:5}]
Let 
\begin{align}
\lambda \define \lim_{h\to a}\operatorname*{max}_xf\p{x,h},
\label{limlem}
\end{align}
and
\begin{align}
g\p{x_0} \define \operatorname*{max}_xg\p{x}= \operatorname*{max}_x\lim_{h\to a}f\p{x,h}.
\end{align}
Based on \eqref{limlem}, $\forall\epsilon_1>0$ there exist $\delta_1>0$ such that 
\begin{align}
\abs{\operatorname*{max}_xf\p{x,h}-\lambda}<\epsilon_1
\end{align}
whenever $0<h-a<\delta_1$. Accordingly, by \eqref{lemAsu}, $\forall\epsilon_2>0$ there exist $\delta_2>0$ such that 
\begin{align}
\abs{f\p{x,h}-g\p{x}}<\epsilon_2
\label{lemAsu2}
\end{align}
whenever $0<h-a<\delta_2$. Let us assume by contradiction that (without loss of generality)
\begin{align}
\Delta \define \abs{g\p{x_0}-\lambda}>0.
\end{align}
However, by using the triangle inequality, one obtains that
\begin{align}
0<\Delta = \abs{g\p{x_0}-\lambda} \leq \abs{g\p{x_0}-\operatorname*{max}_xf\p{x,h}}+\abs{\operatorname*{max}_xf\p{x,h}-\lambda},
\end{align}
and hence 
\begin{align}
0<\Delta-\epsilon_1\leq\abs{g\p{x_0}-\operatorname*{max}_xf\p{x,h}}\leq\abs{g\p{x_0}-f\p{x_0,h}},
\end{align}
for $0<h-a<\text{min}\p{\delta_1,\delta_2}$, which contradicts the assumption in \eqref{lemAsu} (or \eqref{lemAsu2}). Thus, $\delta=0$.
\end{proof}
\begin{remark}
As the proof shows, Lemma \ref{lem:5} remains valid for functions $f:\calX\times\calY\to\mathbb{R}$.
\end{remark}
In our case, the assumptions of Lemma \ref{lem:5} hold true: the uniform convergence is due to the absolutely (square) summability of the sequence $\ppp{h_k}$ and Szeg\"{o}'s theorem, and the existence the limit over the maximization problem indeed exists as was obtained. Thus, the order of limit over $h$ and the maximization over $\epsilon$ in \eqref{orderchange} can be interchanged. 

\section{Derivation of \eqref{R_ddef}}
\label{app:5}
Szeg\"{o}'s theorem \cite{Szego,Widom,Gray,Bottcher} basically states that, for a sequence of Toeplitz matrices $\bT_n = \ppp{t_{i-j}}_{i,j}$ with dimension $n\times n$, for which $\ppp{t_k}$ is absolutely (square) summable, the following holds
\begin{align}
\lim_{n\to\infty}\frac{1}{n}\sum_{k=0}^{n-1}F\p{\tau_{n,k}} = \frac{1}{2\pi}\int_0^{2\pi}F\p{\Tmat\p{\omega}}\mathrm{d}\omega
\label{szegoTh}
\end{align}
where $\ppp{\tau_{n,k}}_k$ are the eigenvalues of $\bT_n$, $\Tmat\p{\omega}$ is the Fourier transform of $\ppp{t_k}$, and $F\p{\cdot}$ is some polynomial function. Furthermore, if $\bT_n$ are Hermitian, then \eqref{szegoTh} holds true for any continuous function $F\p{\cdot}$. 

In our case, however, the matrices $\bA$ and $\bA'$ are not necessarily Hermitian. Nevertheless, based on \eqref{FerroDom}, it can be seen that the dependency of the various non-linear terms (except the third term) on the eigenvalues is only via $\abs{\sigma_m'}^2$, which can be regarded as eigenvalues of the Hermitian matrix $\bA^H\bA$, and so Szeg\"{o}'s theorem can be applied. Regarding the third term in the right hand side of \eqref{FerroDom}, it can be shown \cite{Gray} that a product of Toeplitz matrices also satisfies Szeg\"{o}'s theorem, namely,
\begin{align}
\lim_{n\to\infty}\frac{1}{n}\sum_{k=0}^{n-1}F\p{\rho_{n,k}} = \frac{1}{2\pi}\int_0^{2\pi}F\p{\Tmat\p{\omega}\Smat\p{\omega}}\mathrm{d}\omega
\label{szegoTh2}
\end{align}
where $\ppp{\rho_{n,k}}_k$ are the eigenvalues of product of the Toeplitz matrices, $\bT_n\bS_n$, and $\Tmat\p{\omega}$ and $\Smat\p{\omega}$ are the respective Fourier transforms. Accordingly, since the third term in \eqref{FerroDom} is originated from a product of Toeplitz matrices \eqref{ToepProd}, \eqref{szegoTh2} can be used. Therefore, a direct application of \eqref{szegoTh} and \eqref{szegoTh2} on \eqref{FerroDom}, we finally obtain \eqref{R_ddef}. Finally, note that these considerations are utilized to justify the other places in the paper (for example, \eqref{RgSzego} and \eqref{freqrep}) in which Szeg\"{o}'s theorem is applied. 
\ifCLASSOPTIONcaptionsoff
  \newpage
\fi
\bibliographystyle{IEEEtran}
\bibliography{strings}
\end{document}